\documentclass[a4paper,12pt]{article}
\usepackage{authblk}

\usepackage{graphicx}
\usepackage[mathscr]{eucal}
\usepackage{amsmath,amsfonts,amssymb,amsthm}
\usepackage{times}
\usepackage{hyperref}

\newcommand{\maxim}[1]{\textsc{#1}}

\numberwithin{equation}{section} 

\newtheorem{thm}{Theorem}[section]
\newtheorem{rem}[thm]{Remark}

\newtheorem{theorem}[thm]{\textbf{Theorem}}
 \newtheorem{prop}[thm]{Proposition}
 \newtheorem{lemma}[thm]{Lemma}
 \newtheorem{definition}[thm]{Definition}
 \newtheorem{example}[thm]{Example}

\newtheorem{corollary}[thm]{Corollary}

\renewcommand{\tilde}{\widetilde}
\renewcommand{\hat}{\widehat}

\newcommand{\bref}[1]{\textbf{\ref{#1}}}

\newcommand{\gh}[1]{\mathrm{gh}(#1)}
\newcommand{\fdeg}[1]{\mathrm{fdeg}(#1)}

\newcommand{\dx}{\mathrm{d}_X}

\renewcommand{\d}{\partial}

\newcommand{\cM}{\mathcal{M}}

\renewcommand{\geq}{\,{\geqslant}\,}
\renewcommand{\leq}{\,{\leqslant}\,}

\newcommand{\binner}[2]{%
  {\langle}\kern-4.15pt{\langle}#1{,}\,#2{\rangle}\kern-4.15pt{\rangle}}
\newcommand{\commut}[2]{[#1{,}\,#2]}

\newcommand{\half}{\mathchoice{%
    \ffrac{1}{2}}{\frac{1}{2}}{\frac{1}{2}}{\frac{1}{2}}}

\newcommand{\ffrac}[2]{\raisebox{.5pt}%
  {\footnotesize$\displaystyle\frac{#1}{#2}$}\kern1pt}

\newcommand{\red}{\mathrm{red}}

\newcommand{\dl}[1]{\mathchoice{\ffrac{\d}{\d #1}}{\frac{\d}{\d #1}}{\ffrac{\d}{\d #1}}{\ffrac{\d}{\d #1}}}

\newcommand{\ddl}[2]{\ffrac{\d #1}{\d #2}}

\newcommand{\Liealg}{\mathfrak} 
\newcommand{\algg}{\Liealg{g}}

\newcommand{\cC}{\mathcal{C}}

\newcommand{\fR}{\mathbb{R}}

\newcommand{\fZ}{\mathbb{Z}}

 \def\cE{\mathcal{E}}

 \def\cI{\mathcal{I}}

 \def\cP{\mathcal{P}}
 
\newcommand\blfootnote[1]{%
  \begingroup
  \renewcommand\thefootnote{}\footnote{#1}%
  \addtocounter{footnote}{-1}%
  \endgroup
}

\usepackage[left=3cm,right=2cm,
    top=3cm,bottom=3cm,bindingoffset=0cm]{geometry}
\usepackage{bm}
\usepackage{comment}

\newcommand{\dg}{\mathrm{d}_\algg}
\newcommand{\ev}{\mathrm{ev}}
\newcommand{\eI}{\mathcal{E}}
\newcommand{\comm}[1]{}
\usepackage{autobreak}
\usepackage{authblk}
\usepackage{tikz-cd}
\usepackage{mathtools}
\usepackage{musicography}
    \newcommand{\ms}[1]{^{\scalebox{0.78}{$\scriptstyle [#1]$}}}

\newcommand{\fop}{\mathrm{fop}}
\newcommand{\cA}{\mathcal{A}}
\newcommand{\cR}{\mathcal{R}}

\newcommand{\maxw}{{\mathcal{M}}}
\newcommand{\body}[1]{{\mathrm{body}(#1)}}

\title{Asymptotic boundary structure of Lagrangian \\ gauge theories}
\author{Ivan Dneprov}
\author{Maxim Grigoriev${}^\dagger$}
\author{Mikhail Markov}

\affil{\textsl{Service de Physique de l'Univers, Champs et Gravitation, \protect\\ Universit\'e de Mons, 20 place du Parc, 7000 Mons, 
Belgium \vspace{5pt}}}

\date{}

\begin{document}

\maketitle
\begin{abstract}

Given a local gauge theory on spacetime with boundary, it naturally defines another gauge theory which can be regarded as a theory of the boundary values. For Lagrangian theories,  it comes equipped with the presymplectic structure which can be used to define one or another version of Hamiltonian-like formulation of the initial model. This relation is especially manifest for AKSZ sigma models and more-generally gauge PDEs with compatible presymplectic structure in which case the boundary system is again a gauge PDE with presymplectic structure.  In the context of (flat space) holography one is interested in boundaries at infinity, also known as asymptotic boundaries. The gauge PDE framework naturally extends to this setup, resulting in the notion of gauge PDE with asymptotic boundaries. Although this works perfectly well at the level of equations of motion, the extension to Lagrangian systems appears quite subtle because the presymplectic structure capturing the Lagrangian is divergent at the boundary. We show that any $Q$-cocycle in the bulk (and presymplectic structure in particular) determines a pair of compatible $Q$-cocycles of  the boundary gauge PDE: the renormalized one of the same ghost-degree, and the anomaly cocycle of degree one lower. For the latter, the construction is somewhat analogous to the residue map known in the context of b-geometry. The general formalism is exemplified by scalar and Maxwell fields on AdS and Minkowski spaces. It turns out that in the AdS case the natural action determined by the anomaly presymplectic structure is precisely the one known as the holographic Weyl anomaly in the AdS/CFT context while its null-infinity counterpart was known in a few very particular cases only.
\end{abstract}

\blfootnote{\noindent Supported by the ULYSSE Incentive
Grant for Mobility in Scientific Research [MISU] F.6003.24, F.R.S.-FNRS, Belgium.}
\blfootnote{${}^{\dagger}$ Also at Lebedev Physical Institute and Institute for Theoretical and Mathematical Physics, Lomonosov MSU}

\newpage
\setcounter{tocdepth}{2}
\tableofcontents

\section{Introduction}

The concepts of boundary at infinity and conformal compactification of spacetime play a prominent  role
in the physical interpretation of gravity and gauge theories of direct physical interest as they allow one to define isolated  systems and conserved charges for gravity~\cite{Penrose:1962ij,Penrose:1964ge}. A well-known example of the spacetime with boundary at infinity is the (asymptotically) AdS space and its conformal boundary. This provides a playground for the  AdS/CFT correspondence~\cite{Maldacena:1997re,Aharony:1999ti} and is also well-known in conformal geometry where it is naturally incorporated within the celebrated Fefferman-Graham ambient metric construction~\cite{Fefferman-Graham:1985ambient,Fefferman:2007rka}.

If a (gauge) theory is defined on a manifold with a boundary or a codimension 1 submanifold, the natural question is what structure is naturally induced therein? There are various approaches to this issue among which the apparently simplest one is to regard the boundary as an analog of a Cauchy surface so that the boundary structure can be described as a (version of) the Hamiltonian formulation of the initial theory. In the case of gauge theories this is generally a constrained Hamiltonian system so that  the boundary configurations describing the corresponding phase space are subject to differential (algebraic in the case of mechanics) equations known as constraints. For instance, considering the Maxwell field on $\fR^4$ with coordinates $t\equiv x^0, x^i$ and focusing on the submanifold $t=const$ one finds that the induced phase space is parametrized by configurations of $A_i(x)$ and its conjugate momenta $\pi^i(x)\equiv F^{0i}(x)$ subject to the constraint equation $\d_i\pi^i=0$. The message is that the induced structure on the boundary can be naturally seen as a theory defined on the boundary at the level of equations of motion. Moreover, if the starting point system is Lagrangian then the phase-space comes equipped with the (pre)symplectic structure~\cite{Kijowski:1979dj}, e.g. equivalent to $\int \delta\pi^i\wedge \delta A_i$ in the above example of Maxwell. In the case of general partial differential equations (PDEs) this relation can be made very geometrical if the PDE in question is defined in intrinsic terms as a fibre bundle over the spacetime, equipped with Cartan distribution~\cite{Vinogradov1981}, see~\cite{Krasil?shchik-Lychagin-Vinogradov,Krasil'shchik:2010ij} for the review and systematic exposition. In this case the boundary system is simply given by the pullback of the bundle to the boundary, which is again a PDE in a natural way.

In physics we are mostly interested in gauge theories, in which case the extra structures are present. These structures are best described using one or another version of Batalin-Vilkovisky (BV) formalism. It has been known for quite some time that the constrained Hamiltonian system induced on the boundary can be directly found in terms of the BV formulation in the bulk if the initial theory is formulated as the AKSZ sigma-model. To illustrate this point consider an AKSZ model
over a one-dimensional spacetime. In this case the target space is given by a degree 0 symplectic $Q$-manifold $(F,Q,\omega)$. Of course, $(F,Q,\omega)$ has a natural interpretation of BFV extended phase space with the Hamiltonian $H$ of $Q$ being a BRST charge. The BV system determined by the above AKSZ model is precisely the correct BV formulation for the extended Hamiltonian action determined by the constrained system in question~\cite{Grigoriev:1999qz}. In other words, for the 1d AKSZ sigma model seen as a Lagrangian field theory the Hamiltonian formulation is simply given by the target space BFV system. The same relation holds true for more general AKSZ systems, where the BFV formulation is given by the AKSZ-like model in 1-dimension less~\cite{Barnich:2003wj,Grigoriev:2010ic,Grigoriev:2012xg}. More geometrically, one could say that the Hamiltonian BFV formulation is associated to any codimension-1 submanifold. The relation between BV system and its induced associated BFV-system on the boundary of spacetime can be also established  without resorting to the AKSZ framework using the BV-BFV formalism~\cite{Cattaneo:2012qu,Cattaneo:2015vsa}, see also~\cite{Rejzner:2020xid} for its applications to theories on spacetimes with asymptotic boundaries. Let us also mention a systematic study~\cite{Barnich:2001jy} of the boundary structure in local BRST cohomology framework.

The AKSZ sigma model approach to boundary structure  has a natural generalisation, based on representing the theory in question as a so-called gauge PDE~\cite{Grigoriev:2019ojp}. The latter can be seen, on the one hand, as the generalisation of the nonlagrangian version of AKSZ model to the case of nontopological theories and, on the other hand, as the minimal BV extension of the geometrical approach to PDEs~\cite{Vinogradov1981}. As a geometrical object gPDE is a $Q$-bundle~\cite{Kotov:2007nr} over the de Rham complex of spacetime, seen as a $Q$-manifold,
and just like in the case of PDE or AKSZ sigma-model, the boundary structure is obtained as the pullback of gPDE to the boundary de Rham complex seen as a $Q$-submanifold~\cite{Bekaert:2012vt,Grigoriev:2022zlq,Grigoriev:2023kkk}.\footnote{Although the derivation of the boundary structure of gPDE via pullback to the boundary de Rham complex appeared in~\cite{Bekaert:2012vt}, the precise formulation of the procedure was given\cite{Grigoriev:2022zlq,Grigoriev:2023kkk}. Let us also note that closely related ideas and methods were employed in a less general setting in the context of higher-spin theories~\cite{Vasiliev:2001zy,Vasiliev:2012vf}.}  
The gauge PDE approach to boundary structure of gauge fields has proved fruitful in the study of general AdS space gauge fields. In this context the conformal boundary of spacetime can be turned into a submanifold
of the  ambient space and the boundary structure is obtained by a pullback of the corresponding gPDE~\cite{Bekaert:2012vt,Bekaert:2013zya,Chekmenev:2015kzf,Bekaert:2017bpy}. 

The Lagrangian field theories are described in the gPDE formalism by gPDEs equipped with compatible presymplectic structures. Under some regularity conditions, these data define not only a natural action functional on the space of sections but also the full-scale local BV formulation on the symplectic quotient of the corresponding space of supersections. The resulting BV system is generally weaker (has more gauge-inequivalent solutions) than the starting point gPDE so that the additional condition ensuring the equivalence is to be imposed. Such a gPDE with compatible presymplectic structure can be pulled back to the de Rham complex of the boundary (or more general submanifold) giving a new gPDE with the compatible presymplectic structure. However the degree of this presymplectic structure equals the spacetime dimension so that the resulting BV-like formulation induced on the quotient of the space of supersections is naturally interpreted as the BFV one~\cite{Grigoriev:2022zlq}, see also~\cite{Grigoriev:2025vsl}. This extends the AKSZ approach to boundary structure to the case of general gauge theories. It is important to stress, however, that in contrast to AKSZ setup it is not a priori guaranteed that the boundary BFV system obtained this way is a correct BFV formulation of the system in question and some extra regularity conditions are to be imposed. 

In this work we are concerned with gauge theories on spacetimes with asymptotic boundaries. The asymptotic boundary, also known as boundary at infinity, is not a submanifold of the initial spacetime but is a boundary of its compactification i.e. a spacetime with boundary such that the initial spacetime is identified with the interior of the compactified spacetime. Given a compactified spacetime, the asymptotic boundary can be treated as a usual one. The important subtlety, however, is that if we are interested in a gauge theory its extension to the compactified space-time is not given by default and has to be decided and constructed. Moreover, the choice of space-time compactification  typically involves space-time metric which itself can be a dynamical field of the theory under consideration. All in all this leads to the notion of  compactification of a gauge theory.

It turns out that the language of gauge PDEs is well-suited not only to study boundary structures of gauge theories but also to study their compactification. More precisely, we introduce a notion of a gauge PDE with asymptotic boundary. This is roughly speaking a $Q$-bundle where both the spacetime and the fiber are manifolds with boundaries and solutions are by definition sections whose restrictions to the boundary of spacetime belong to the boundary of the field space. In plain terms, this means that among the fields there is one describing a boundary defining function while gauge transformations describe its ambiguity.\footnote{In the context of gravity closely related field-theoretical systems were studied, for example, ~\cite{Friedrich:1981wx,friedrich1983cauchy,Kroon:2016ink}.} Such a system naturally defines at least two gauge PDEs. The first one, the interior, is obtained by restricting to the interior of the base followed by the restriction to the interior of the resulting total space. The second one is the boundary of the restriction of the initial bundle to the boundary of the base. The first one is an equivalent of the initial system to be compactified while the second one is regarded as the gPDE describing the corresponding boundary structure. In slightly different terms the notion of gPDE with asymptotic boundaries was  already in~\cite{Grigoriev:2023kkk,Grigoriev:2025gvk}, see also~\cite{Bekaert:2012vt} where the related considerations were put forward using the ambient space formalism.

Although  gPDEs with asymptotic boundary give a natural setup to study compactifications and boundary structures associated to asymptotic boundaries, the extension to Lagrangian gauge theories appears to be subtle and at the same time interesting. The point is that while the compactification at the level of equations of motion can be almost systematically performed for usual models of gravity and gauge fields, the presymplectic structure encoding the Lagrangian is usually polynomially divergent at the boundary and does not determine a presymplectic structure for the boundary theory in a natural way. Nevertheless, as we show in this work this difficulty can be overcome systematically using solely the invariant geometrical tools. More precisely, we demonstrate that any $Q$-cocycle that diverges polynomially at the boundary generally gives rise to a pair of $Q$-cocycles
of the boundary gPDE. One of the same ghost degree (called renormalised) and that of degree 1 lower (called anomaly); both having the same form-degree. While the appearance of the renormalised cocycle is expected, the anomaly cocycle is less obvious and the mechanism it arises is somewhat similar to the residue map in the context of b-geometry~\cite{melrose1993atiyah}, where the de Rham cocycle in the bulk defines a de Rham cocycle on the singular boundary, which is 1 form-degree lower than the initial cocycle.

We demonstrate that when applied to the presymplectic structure on gPDE with asymptotic boundary  the corresponding renormalised presymplectic structure on the boundary gPDE defines a BFV system, which can be considered as a counterpart of the boundary BFV system in the non-asymptotic case. At the same time the anomaly, if nontrivial, defines a new BV system on the boundary. In the case of asymptotically-AdS space the corresponding action functional coincides with the well-known holographic Weyl anomaly which in the usual AdS/CFT considerations emerges as a log-divergent piece of the effective action. For instance, for the critical scalar of Weyl weight $w=\ell-d/2$, $\ell=1,2,\ldots$ on $d+1$-dimensional AdS space, the action determined by the anomaly presymplectic structure is precisely the 1st order action of the higher-order singleton field on the boundary, whose equations of motion in Cartesian coordinates read as $\Box^\ell \varphi=0$. 

As applications of the general framework proposed in the paper we consider scalar field and Maxwell field on  AdS and Minkowski space and explicitly derive their boundary structure, including the renormalised and anomaly presymplectic structures and their associated BFV and BV systems. While in the AdS case the resulting objects are known in one or another approach, in the case of Minkowski space and its asymptotic boundary at null-infinity the full-scale gauge PDE description of the boundary structure was not known, see however~\cite{Grigoriev:2023kkk}, and the analogs of holographic Weyl anomalies were not studied systematically. In this work we fill this gap by explicitly constructing a gPDE with asymptotic boundaries describing the compactification of flat Minkowski space together with the scalar and the Maxwell field defined on its background.  Moreover, we compute the renormalised and the anomaly presymplectic structures on the boundary and explicitly derive the action functional determined by the anomaly. It turns out that the latter gives an infinite tower of manifestly conformally-Carroll invariant theories which generalize the magnetic limits of the scalar and Maxwell fields~\cite{Henneaux:2021yzg}. 

Although the examples explicitly discussed in this work are limited to constant curvature backgrounds (or, better, Cartan-flat ones) the general  approach perfectly applies to full-fledged gravity and gauge fields defined on its background. More precisely, in the present terms the construction of~\cite{Grigoriev:2025gvk} is the gauge PDE with boundaries for asympotically-AdS gravity and gauge fields on its background considered at the level of equations of motion while its analog from~\cite{Grigoriev:2023kkk} can be readily completed to the gauge PDE with asymptotic boundaries for asymptotically flat gravity. As for the presymplectic structures for these theories and their associated structures on the boundary, we leave their analysis for a future work.

The paper is organized as follows: Section~\bref{sec:prelim} is a brief recap of the gauge PDE approach and its presymplectic extension. It also explains the generalisation to gauge PDEs over background and the gauge PDE description of boundary structure in the case of usual (as opposed to asymptotic) boundaries. The general construction is presented in Section~\bref{sec:general}. In particular, there we introduce the basic notions of the approach such as $Q$-manifolds with $Q$-boundaries and gauge PDEs with asymptotic boundaries and define the map sending polynomially divergent $Q$-cocycles in the bulk to pairs of $Q$-cocycles on the boundary.  The applications to fields on AdS space are given in Section~\bref{sec:AdS}. The important ingredient is the description of the compactified AdS space as a gauge PDE with asymptotic boundaries, which can be also identified as the vanishing-curvature limit of the corresponding gPDE from~\cite{Grigoriev:2025gvk}. Section~\bref{sec:flat} is devoted to the boundary structure of gauge fields on Minkowski space on its boundary at  null-infinity. Just like in the AdS case, we describe the compactified Minkowski as a gauge PDE with asymptotic boundary
and derive a family of Carroll-invariant theories on the boundary determined by the anomaly presymplectic structure.

\section{Preliminaries}\label{sec:prelim}

\subsection{Gauge PDEs} \label{prelim:gPDE-backgr}
To describe local gauge theories we use the language of gPDEs, which heavily relies on the notion of $Q$-manifolds and $Q$-bundles. A $Q$-manifold is a $\mathbb{Z}$-graded manifold equipped with a homological vector field $Q$, i.e. degree 1 vector field satisfying $Q^2 = 0$. Following the tradition in the literature on BV formalism we denote the degree by $\gh{\cdot}$, e.g. $\gh{Q}=1$. Although general considerations involve $\fZ$-grading all the examples considered in this work are limited to nonnegative degree so one can safely assume that the degree is nonnegative.

Natural morphisms between $Q$-manifolds, called $Q$-maps, are degree preserving maps which are compatible with the $Q$ structure, i.e. $\phi: (M,Q) \xrightarrow{} (N,q)$ is a $Q$-map iff $Q\phi^* = \phi^* q$. A $Q$-bundle is a bundle where both the total space and the base are $Q$-manifolds and  the projection is a $Q$-map~\cite{Kotov:2007nr}. Recall that this assumes that both the total space and the base are $\fZ$-graded and $\pi$ is also compatible with the degree.

The basic example of a $Q$-manifold is the shifted tangent bundle $T[1]X$ of a real manifold $X$. The algebra of functions on $T[1]X$ can be identified with the algebra of differential forms on $X$ and under this identification the de Rham differential on $X$ determines a canonical $Q$ structure $\dx$ on $T[1]X$ while the form-degree gives a $\fZ$-grading. If $\{x^\mu, \theta^\mu\}$ are local coordinates on $T[1]X$ such that $x^\mu$ are coordinates on the base $X$ and $\theta^\mu$ are the induced coordinates on the fibers then the canonical $Q$ structure has the form $\mathrm{d}_{X} = \theta^\mu \ddl{}{x^\mu}$ and $\gh{x^\mu}=0$, $\gh{\theta^\mu}=1$.

A $Q$-manifold $(E,Q)$ is called an equivalent reduction of $(E',Q')$ iff $(E',Q')$ is a locally-trivial $Q$-bundle
over $(E,Q)$ with a contractible fiber. In other words, there exists a local trivialization such that  locally  $(E',Q')= (E,Q)\times (T[1]V, \mathrm{d}_V)$, where $V$ is a linear graded manifold. Two $Q$-manifolds $(M_1,q_1), (M_2, q_2)$ are called equivalent if there exists $(M,Q)$ such that $(M_1,q_1), (M_2, q_2)$ are equivalent reductions of $(M,Q)$. Further details can be found in e.g.~\cite{Grigoriev:2019ojp}.

\begin{definition}\cite{Grigoriev:2019ojp}
A gauge PDE is a $Q$-bundle $(E,Q) \xrightarrow{\pi} (T[1]X,\mathrm{d}_{X})$, where $X$ is interpreted as the spacetime manifold $X$ (manifold of independent variables).
    \begin{enumerate}
    \item Sections of $E$ are interpreted as field configurations.
    \item Solutions are $Q$-sections, i.e. sections $\sigma$ satisfying\footnote{By section we mean what is often called a degree preserving sections. These should not be confused with super-sections.} 
    \begin{equation}
        \mathrm{d}_{X} \sigma^* = \sigma^* Q\,.
    \end{equation}

    \item Gauge parameters are vertical vector fields $Y$, $gh(Y) = -1$. 
    
    \item Infinitesimal gauge transformations of a section $\sigma$ are given by
    \begin{equation}
        \delta \sigma^* = \sigma^* [Q,Y] \,.
    \end{equation}

    \item Gauge for gauge parameters are vertical vector fields $Y_k$, $k>1$, $\gh{Y_k}=-k$.  The corresponding infinitesimal transformations do not act on sections but act on (gauge-for) gauge parameters, e.g. $\delta_{Y_2} Y=\commut{Q}{Y_2}$. 
    \end{enumerate}
\end{definition}
The notion of equivalence of $Q$-manifolds naturally extends  to $Q$-bundles over $T[1]X$ giving the notion of equivalence of gPDEs~\cite{Grigoriev:2019ojp}. To exclude nonlocal systems it is also natural to require a gPDE to be equivalent to a local BV system. An important property of gPDEs is that they 
automatically contain a BV formulation of the underlying theory. More precisely, the restriction to $X\subset T[1]X$ of the bundle $SJ^\infty(E)\to T[1]X$ of jets of supersections of $E$ is naturally a local BV system at the level of equations of motion~\cite{Barnich:2010sw}. For further details on gPDEs see~\cite{Grigoriev:2019ojp}. 

\subsection{Gauge PDEs with compatible presymplectic structures}

If we are dealing with Lagrangian gauge systems in the gPDE approach, an extra structure on the gPDE is required.
\begin{definition}
    A compatible presymplectic structure on  a gPDE $(E,Q) \xrightarrow{\pi} (T[1]X,\mathrm{d}_{X})$ is a $2$-form $\omega$ of $gh(\omega) = k$, obeying 
    \begin{equation}
    \label{presimpcomp}
    \begin{gathered}
        d\omega = 0 \,, \qquad L_Q \omega \in \cI_{T[1]X}
    \end{gathered}
    \end{equation}
where $\cI_{T[1]X} \subset \bigwedge^\bullet(E)$ is the differential ideal generated by the forms $\pi^*\alpha, \alpha \in \bigwedge^{1}(T[1]X)$. Unless otherwise specified, it is assumed  $\gh{\omega}=\dim{X}-1$.
\end{definition}
Given a gPDE equipped with compatible presymplectic structure, these data determine a natural action functional on the space of sections of $E$. More precisely, we pick a presymplectic potential $\chi$  and covariant Hamiltonian $H$ such that 
\begin{equation}
\label{chiH-axioms}
\omega=d\chi\,, \qquad i_Q\omega+d H \in \cI_{T[1]X}\,.
\end{equation}
The existence of   $\chi$  and  $H$ is guaranteed locally by~\eqref{presimpcomp} and for $\dim X>1$ globally as well.
The ambiguity in $H$ is given by a function of the form $\pi^*h, h \in \cC^\infty(T[1]X)$ while that in $\chi$ by 
an exact 1 form. The action functional on the space of sections is given by:
\begin{equation}
\label{AKSZ-action}
S[\sigma]=  \int_{T[1]X} {\sigma}^*(\chi)(\mathrm{d}_{X}) + {\sigma}^*(H)\,,
\end{equation}
and is refereed to as the intrinsic action. In the case where $E$ is a graded-geometrical reformulation of a PDE equipped with the presymplectic current, \eqref{AKSZ-action} is precisely the intrinsic action introduced in~\cite{Grigoriev:2016wmk}. Of course in the case of $\dim{X}=1$ this is just the usual Hamiltonian action or its presymplectic generalisation~\cite{Souriau1970,Kijowski:1973gi,Gotay1979}.
\maxim{Correct citation}

It is easy to see that the ambiguity in $H$ results in adding a constant term to the above action while the ambiguity in $\chi$ gives boundary terms. Because we are interested in systems on manifolds with boundaries we consider $\chi$ and $H$ as part of the data defining the system and hence in what follows by gPDE with presymplectic structure we understand $(E,Q,T[1]X)$ together with $\chi,H $ satisfying \eqref{chiH-axioms}.

An alternative representation of the gPDE equipped with compatible presymplectic structure is as follows. One defines a $0$ form $l$ and $1$-form $\chi$ on $E$ such that 
\begin{equation}
\label{chi-l-axioms}
\omega=(d+L_Q)(\chi+l)+\alpha\,, \quad \alpha \in \cI^1_{T[1]X}\,,
\end{equation}
where $\cI^1_{T[1]X}$ denotes a subspace of 1-forms in $\cI_{T[1]X}$. It is a matter of direct computation that \eqref{chiH-axioms} and \eqref{chi-l-axioms} are equivalent via $l=i_Q\chi+H$. In particular $Ql=0$. Let us note that if $\sigma$ is a solution to $(E,Q,T[1]X)$ then the intrinsic  action \eqref{AKSZ-action} evaluated on $\sigma$ can be represented as:
\begin{equation}
    S[\sigma]=\int_{T[1]X}\sigma^*l\,.
\end{equation}

It is important to stress that the Euler-Lagrange equations of the intrinsic action are generally much weaker than the equations of motion encoded in the starting point gauge PDE. However, if one takes into account algebraic gauge transformations and auxiliary fields the two systems can be equivalent at the level of equations of motion. In this case, the presymplectic structure is referred to as complete. A more precise formulation of the equivalence of the BV system is done in terms of the underlying local BV systems, see~\cite{Dneprov:2024cvt,Dneprov:2025eoi} for more details.

A crucial property of gPDEs equipped with presymplectic structure is that they encode a BV formulation of the underlying systems. More precisely, let $S\Gamma(E)$ be the space of supersections and $ev: T[1]X \times S\Gamma(E) \xrightarrow{} E$ the corresponding evaluation map. The total space differential $Q$ acts on $S\Gamma(E)$ in a  natural way and hence defines a $Q$-structure $s$ therein. Moreover, the presymplectic structure determines a presymplectic structure on $S\Gamma(E)$ via
\begin{equation}
    \bar\omega = \int_{T[1]X} \ev^* \omega\,, \qquad \gh{\bar\omega}=\gh{\omega}-\dim{X}\,.
\end{equation}
Moreover, the data of $Q,\chi,H$ defines the following function on $S\Gamma(E)$ 
\begin{equation}
    S_{BV} = \int_{T[1]X} {\ev}^*(\chi)(\mathrm{d}_{X}) + {\ev}^*(H)\,.
\end{equation}
We refer to $S_{BV}$ as the presymplectic BV-AKSZ action or just BV-AKSZ action. It is clear that when restricted to section (as opposed to supersections) BV-AKSZ action coincides with~\eqref{AKSZ-action}.  If in addition $\bar\omega$ happens to be regular, the respective  symplectic quotient of $S\Gamma(E)$ acquires all the structures of the BV system. In particular, $(S_{BV},S_{BV})=0$ modulo boundary terms, where $(\cdot,\cdot)$ is the BV antibracket determined by $\bar\omega$.  The above arguments are a straightforward generalization of the AKSZ construction of BV action, see e.g.~\cite{Roytenberg:2006qz}.  A more detailed analysis is done using the super-jet bundle of $E$ and  gives an explicit derivation of the jet-bundle BV formulation of the underlying gauge theory and confirms its locality, see~\cite{Dneprov:2024cvt} for details.

Let us note that all the above considerations apply to the case where $\gh{\omega}$ is not necessarily $\dim(X)-1$. Of the prime importance is the case $\gh{\omega}=\dim(X)$. In this case $S[\sigma]$ given by \eqref{AKSZ-action}  vanishes identically for degree reasons but BV-AKSZ action $S_{BV}$ is naturally interpreted as the BRST charge of the BFV formulation. For instance, if $\Sigma\subset X$ is a Cauchy surface then the pullback of $(E,Q,T[1]X)$ equipped with $\chi,H$ to $T[1]\Sigma \subset T[1]X$ is again a gPDE with presymplectic structure but now $\gh{\omega}=\dim\Sigma$ and hence the analog of BV system  on the space of supersections of $E|_{T[1]\Sigma} \to T[1]\Sigma$ is naturally a BFV system, see~\cite{Grigoriev:2022zlq} for more details. Under certain regularity conditions this system can be shown to be equivalent to a correct BFV system describing the same gauge theory, see~\cite{Grigoriev:2022zlq} for more details.

\subsection{Background fields}

In this work, we are often interested in gauge fields coupled to a background geometry. The latter can be understood as a particular configuration of background fields. In the gPDE approach, background fields are again described by another gPDE $(B,\gamma,T[1]X)$, giving rise to the notion of a gPDE over a background~\cite{Dneprov:2025eoi}: namely, a bundle $(E,Q) \to (B,\gamma)$ where both $E$ and $B$ are bundles over $T[1]X$. Pulling back $E$ along a solution $\sigma_0: X \to B$, viewed as a submanifold in $T[1]X$, results in a usual gPDE $(\sigma_0^* E,Q,T[1]X)$, which describes the original gauge theory in the specific background $\sigma_0$. Informally, one can think of the bundle $(E,Q) \xrightarrow{} (B, \gamma)$ as a family of gPDEs parameterized by formal solutions of the background system $(B, \gamma)$.

Gauge PDEs over background are especially useful in describing gauge fields on homogeneous spaces and invariant under the underlying spacetime symmetry algebra (e.g. Poincar\'e or (anti) de Sitter). In this case the dependence of all the structures on spacetime coordinates  
is through the background fields only. This implies that locally the underlying gPDE over background takes the following factorised form:
\begin{equation}
\label{backg-factor}
(E,Q)=(E_0,Q_0)\times (T[1]X,\dx)\,,\qquad (B,\gamma)=(B_0,\gamma_0)\times (T[1]X,\dx)
\end{equation}
and $(E_0,Q_0)$ is a $Q$-bundle over $(B_0,\gamma_0)$. It is clear that in this case all the information about the system is encoded in the $Q$-bundle $(E_0,Q_0) \to (B_0,\gamma_0)$ so that one can effectively forget about the spacetime. Moreover, for fields defined on homogeneous spaces, $(B_0,\gamma_0)$ can be taken to be $(\algg[1],\dg)$, where $\algg$ is an underlying Lie algebra in which the respective flat Cartan connection 1-form takes values, see~\cite{Dneprov:2024cvt}.

Informally, the gPDE over $T[1]X$
describing the gauge fields in question is replaced by the
gauge PDE with the same fiber but the spacetime $T[1]X$ replaced by $(B_0,\gamma_0)$. Such description is very economical, more invariant, and captures the information of how fields transform under the underlying spacetime symmetry in a manifest way. This is precisely what we employ in our study of gauge fields on asymptotic boundaries of homogeneous spaces. Note that gauge PDEs describing gauge fields coupled to the gravity background also have the factorized structure~\eqref{backg-factor}.

The extension of gPDEs over background to
Lagrangian systems is a straightforward generalisation of gPDEs with presymplectic structures. One asks that $(E,Q) \xrightarrow{}(B,\gamma)$ be equipped with extra data $\omega,\chi,H$ obeying the following compatibility conditions: 
\begin{equation}
\begin{gathered}
    \omega = d \chi\,, \qquad  gh(\omega) = dim X - 1 \, ,\qquad     i_Q \omega + dH \in \cI_B
\end{gathered}
\end{equation}
analogous to~\eqref{chiH-axioms} with the only modification being that $\cI_{T[1]X}$ is replaced with $\cI_B$, i.e. the ideal generated by forms ${\pi_B}^*\alpha$, $ \alpha \in \bigwedge^{1}(B)$. 
Restricting to a solution $\sigma_0$ of the background gPDE, one obtains a gPDE with presymplectic structure. Moreover, the resulting intrinsic  action is invariant under the background gauge transformations which affect both the background fields and the dynamical ones. If a particular background solution is fixed, the residual symmetries are precisely the space-time symmetries of the system in this background.

In the case of fields on homogeneous spaces or gravitational backgrounds the presymplectic structure as well as $\chi$ and $H$ can be defined on $(E_0,Q_0)$ so that one can disregard the spacetime and work solely with the bundle $(E_0,Q_0)$.

\subsection{Example: Maxwell field on AdS space as a gPDE over background} \label{sec:example-M}
For further use, let us illustrate the approach recalled above in the case of Maxwell theory coupled to a flat AdS background. 

As the total space of the background gPDE we take $B= \mathfrak{g}_\Lambda [1] \times T[1]X \xrightarrow{} T[1]X$, where $\mathfrak{g}_\Lambda$ is the (A)dS symmetry algebra, and the total $Q$-structure factorizes 
\begin{equation}
    \gamma = \mathrm{d}_X + \mathrm{d}_{\algg}\,,
\end{equation}
where $\mathrm{d}_{\algg}$ is the Chevalley-Eilenberg differential of $\mathfrak{g}_\Lambda$, seen as a homological vector field $\mathfrak{g}_\Lambda [1]$. Introducing coordinates on $\mathfrak{g}_\Lambda [1]$ by  $\{\xi^a, \rho^{ab} \}$ corresponding to translations and Lorentz sectors, the action of $\gamma$ reads as
\begin{equation}
\begin{gathered}
    \gamma \xi^a =\xi^b\rho_{b}{}^{a}\,, \\
    \gamma \rho^{ab} = \rho^a{}_c \rho^{cb} - \frac{2\Lambda}{(D-1)(D-2)} \xi^a \xi^b \,,
\end{gathered}
\end{equation}
where $\Lambda$ is the cosmological constant and indices are raised and lowered using the flat metric.  Solutions to $(B,\gamma,T[1]X)$ are flat connections and gauge transformations are the usual gauge transformations of a connection. It goes without saying that $\sigma^*(\xi^a)$ is nondegenerate as it is interpreted as the background frame field (or soldering form in the Cartan geometry language).

To describe Maxwell theory on this background we  first describe the corresponding off-shell system, i.e. system, where Maxwell equations are not imposed but gauge transformations are given.
A minimal gauge PDE formulation of this system as well as the on-shell one is well-known  in one or another version e.g. from  local BRST cohomology considerations of Einstein-Yang-Mills theory, see e.g.~\cite{Barnich:1995ap} or unfolded formulation of AdS gauge fields~\cite{Lopatin:1987hz,Vasiliev:2005zu}, so here we only state the result: the fiber is a space of formal solutions to the Bianchi identity $\d_{[a}F_{bc]}$
taken at one space-time point and extended by an additional degree $1$-coordinate $C$ corresponding to the gauge potential and the gauge parameter. 
As fiber coordinates of $(E,Q)\xrightarrow{}(\mathfrak{g}_\Lambda[1]\times T[1]X,\gamma)$ we take $\{C,F_{ab},\dots\}$, where $F_{ab}=-F_{ba}$, $gh(C)=1$, $gh(F_{ab})=0$, and $\dots$ denotes the coordinates parametrising the higher jets of $F_{ab}$. Following \cite{Grigoriev:2025gvk}, we introduce the vector field
\begin{align}\label{nabla}
    \nabla_a=[Q,\frac{\partial}{\partial\xi^a}]\,.
\end{align}
In so doing we assume that degree $1$ coordinate functions $c,\xi^a,\rho^{ab}$ are fixed so that the vector field $\dl{\xi^a}$ is defined unambiguously in the sense that $\dl{\xi^a}$ doesn't depend on the choice of degree-$0$ coordinates. Functions $\nabla_{a_1}\dots\nabla_{a_n}F_{bc}$ together with $C$ form an overcomplete coordinate system on the fiber of $E$. A genuine coordinate system is obtained by picking an independent subset by e.g. taking only the totally symmetrized combinations $\nabla_{(a_1}\dots\nabla_{a_n}F_{b)c}$.

The advantage of working  in terms of overcomplete coordinates $\nabla_{a_1}\dots\nabla_{a_n}F_{bc}$ is that the action of $Q$ on these functions takes a particularly simple form. Namely, it is sufficient to define the action of $Q$ on $C$ and $F_{ab}$ and the rest follows from $[\nabla_a,Q]=0$ and \eqref{nabla}. More precisely, we take
\begin{equation}
\begin{split}
    Q C &= \half F_{ab}\xi^a \xi^b \,, \\ 
    Q F_{ab} &= \xi^c \nabla_c F_{ab} + \rho_a{}^c F_{cb} + \rho_b{}^c F_{ac}\,.
\end{split}
\end{equation}
It is important to note that the relations between the above overcomplete coordinates are consequences of the nilpotency of $Q$. In particular, $Q^2=0$
implies:
\begin{equation}
\begin{split}
\nabla_{[a}F_{bc]}&=0\,,\\
    [\nabla_{a},\nabla_b]F_{cd}&=\frac{2\Lambda}{(D-1)(D-2)}(g_{ca}F_{bd}-g_{cb}F_{ad}+g_{ad}F_{cb}-g_{bd}F_{ca})\,,
\end{split}
\end{equation}
where $g_{ab}$ is a constant metric. The remaining relations are obtained by $\nabla^a$ prolongations of the ones above. Note that nothing in these considerations involved $T[1]X$ so that indeed $E=E_0 \times T[1]X$, where $E_0$ is a fiber bundle over $\algg_\Lambda[1]$ which encodes all the information about the theory.  
Finally, the Maxwell theory corresponds to the restriction to the zero locus of the ideal $\eI^\maxw$ generated by the $\nabla_a$-prolongations of
\begin{align}
    \nabla^bF_{bc}\,.
\end{align}

The compatible presymplectic structure is given by
\cite{Alkalaev:2013hta}
\begin{equation}\label{maxwell-presymp}
    \omega = d(\xi_{ab}\ms{D-2}F^{ab})dC \,,
\end{equation}
where $D = dim X$ and
\begin{align}
    \xi\ms{N}_{a_1\dots a_{D-N}}\equiv\frac{1}{N!}\epsilon_{a_1\dots a_{D-N}b_1\dots b_{N}}\xi^{b_1}\dots\xi^{b_{N}}\,.
\end{align}
One can explicitly verify that 
\begin{equation}
    i_Q \omega + d(-\half \xi\ms{D} F_{ab} F^{ab}) + d\xi(\dots) +\eI^\maxw= 0\,.
\end{equation}
Again, this presymplectic structure is a pullback of the presymplectic structure on $E_0 \to \algg[1]$.

Choosing a solution of the background $\sigma_0: T[1]X \xrightarrow{} \mathfrak{g}_\Lambda [1] \times T[1]X $, corresponding to a flat Cartan connection parameterized by 
\begin{equation}
    \sigma_0^*\xi^a \equiv e^a(x,\theta)\equiv  e^a_{\mu}(x)\theta^\mu\,, \qquad \sigma_0^*\rho^{ab} \equiv \rho^{ab}_{\mu}(x)\theta^\mu\,,
\end{equation}
one arrives at a gPDE $(\sigma_0^*E,Q,T[1]X)$ which is $(E,Q)$ pulled back by $\sigma_0$. Introducing the coordinates on its space of sections via 
\begin{equation}
    \begin{gathered}
        \sigma^*C \equiv A(x,\theta)\equiv A_\mu(x) \theta^\mu \,, \qquad F_{\mu\nu}\equiv e^a{}_{\mu} e^b{}_{\nu} \sigma^* F_{ab}\,,
    \end{gathered}
\end{equation}
the corresponding intrinsic action takes the form
\begin{equation}
    S[A, F| e^a, \rho^{ab}] = \int_X d^Dx\det(e)(2F^{\mu\nu}\partial_\mu A_\nu - \half F^{\mu\nu}  F_{\mu\nu}) \,.
\end{equation}
where indices are raised using the metric $\gamma_{\mu\nu}=e^{a}{}_{\mu}g_{ab}e^{b}{}_{\nu}$.
 Passing to the space of supersections and taking the symplectic quotient produces the corresponding BV formulation~\cite{Dneprov:2022jyn,Grigoriev:2022zlq}.

\subsection{Structures induced on submanifolds/boundaries}
\label{sec:boundary-induce}

A substantial advantage of the gauge PDE approach and its presymplectic version is that gauge PDEs behave well with respect to the restriction to spacetime submanifolds. More precisely, let $(E,Q) \xrightarrow{} (T[1]X, \mathrm{d}_{X})$ be a gPDE and  $b: \Sigma \xrightarrow{} X$ a spacetime submanifold. In particular, if $X$ is a manifold with boundaries one can take $\Sigma=\d X$. It is easy to see that $T[1]\Sigma$ is naturally a $Q$-submanifold in $T[1]X$ and moreover the restriction of $E$ to $T[1]\Sigma \subset T[1]X$ can be seen as a $Q$-submanifold in $E$. It follows that the pullback bundle $(b^*E, Q) \xrightarrow{} (T[1]\Sigma, \mathrm{d}_\Sigma)$ is again a gauge PDE.   What does this gPDE describe? Since every $Q$-section of $(E,Q) \xrightarrow{} (T[1]X, \mathrm{d}_{X})$ (i.e. every solution) restricts to a $Q$-section of $(b^*E, Q_b) \xrightarrow{} (T[1]\Sigma, d_\Sigma)$ this theory should be interpreted as the theory of possible submanifold (boundary) values of the  fields described by $(E,Q,T[1]X)$, see~\cite{Grigoriev:2022zlq,Grigoriev:2023kkk} for further details.

Although in this work we do not discuss systems with nontrivial boundary conditions, let us mention that this setup leaves a natural room for the boundary conditions. Namely, boundary conditions correspond to sub-gPDE $E_\Sigma$ of $(b^*E, Q,T[1]\Sigma)$ and the solutions of $(E,Q,T[1]X)$ satisfying the boundary conditions $(E_\Sigma, Q_\Sigma,T[1]\Sigma)$ are defined to be solutions of $(E,Q,T[1]X)$ such that their restrictions to $T[1]\Sigma$ solve $(E_\Sigma, Q_\Sigma,T[1]\Sigma)$. Note that the boundary conditions defined in this way are the generalized ones because they can restrict not only fields but also gauge parameters,  see~\cite{Grigoriev:2023kkk} for more details and examples. For a discussion of boundary conditions in the AKSZ setup see e.g.~\cite{Pulmann:2019vrw}.

If we are interested in Lagrangian systems, the gPDE in question comes equipped with a compatible presymplectic structure $\omega$. Its pullback to $b^*E$, which we denote by $\Omega_\Sigma$, makes $(b^*E,Q,T[1]\Sigma)$ into a gauge PDE with shifted presymplectic structure. Indeed,
if $\dim{X}-\dim{\Sigma}=k$ then the degree of the presymplectic structure $\int_{T[1]\Sigma} \ev^* \Omega_\Sigma$ induced on the space of supersections $S\Gamma(b^*E)$ is $k-1$. In particular, if 
$\Sigma$ is of codimension $1$ (e.g. $\Sigma$ is a Cauchy surface) this is a degree $0$ presymplectic structure and it is naturally interpreted as a BFV presymplectic structure. Moreover, under some regularity assumptions,\footnote{which are not fulfilled automatically for e.g. the minimal formulation of gravity~\cite{Grigoriev:2020xec}, see also \cite{Canepa:2020ujx}.} the shifted BV system determined by $(b^*E,Q,T[1]\Sigma)$ with $\Omega_\Sigma$ on the symplectic quotient of $S\Gamma(b^*E)$ is the BFV-like formulation of the starting point system.

As an illustration of how a gauge PDE with presymplectic structure induces its BFV formulation when pulled back to a Cauchy surface let us get back to the example of the Maxwell theory given in Section~\bref{sec:example-M} and for simplicity assume that the cosmological constant vanishes. Assuming that we work in Cartesian spacetime coordinates $x^a=\{x^0,x^i\}$ we set background fields to the simplest particular solution $\sigma_0^*(\xi^a)=\theta^a$ and $\sigma_0^*(\rho^{ab})=0$. Then following~\cite{Grigoriev:2022zlq} we pull back $E$ to the submanifold singled out by $\theta^0=0$, $x^0=0$. The $Q$ structure and the presymplectic structure now reads as
\begin{equation}
\begin{gathered}
    Q C = \half F_{ij}\theta^i \theta^j\,, \\ 
    Q F_{ab} = \theta^i \nabla_i F_{ab} 
\end{gathered}
\end{equation}
and  
\begin{equation}\label{maxwell-presymp2}
    \omega_\Sigma = d(\epsilon_{ijk}\theta^j\theta^k F^{0i})dC \,, \qquad \epsilon_{ijk}=\epsilon_{0ijk}\,,
\end{equation}
where for simplicity we assumed $\dim{X}=4$. It is a matter of a straightforward computations to check that the induced presymplectic structure on the space of supersections $S\Gamma(b^*E)$ is given by: 
\begin{equation}
   \int_{\Sigma} \delta c \wedge \delta J^i_i + \delta A_i \wedge \delta \pi^i\,,
\end{equation}
where the coordinates on $S\Gamma(b^*E)$ are introduced as follows:
\begin{equation}
\begin{gathered}
    \ev^* C = c(x) + A_i(x) \theta^i + F^*_{ij}(x) \theta^i \theta^j + ... \\
    \ev^* F^{0i} = \pi^{i}(x) + J^{i}_m (x) \theta^m + J^{i}_{mn} (x) \theta^m \theta^n + ...
\end{gathered}
\end{equation}
The presymplectic structure is obviously regular and the respective symplectic quotient can be identified with the BFV phase space for the minimal BFV formulation of Maxwell field. The AKSZ-like functional reads as
\begin{equation}
    S_\Sigma = \int_\Sigma \pi^i \partial_i c 
\end{equation}
and is precisely the standard expression for the BFV-BRST charge.

\section{Gauge PDEs with asymptotic boundaries}
\label{sec:general}

\subsection[Q-boundaries]{$Q$-boundaries}

We first need to recall the notion of a graded manifold with boundary. A graded manifold $M$ can be defined by specifying its body $M_0$ (which can be assumed a smooth manifold) together with the sheaf of graded commutative algebras on $M_0$, see e.g.\cite{Deligne:1999qp} for further details.  Replacing $M_0$ with a real manifold with boundary $\d M_0$ and specifying a codimension-$1$ submanifold $\d M$ such that its body coincides with $\d M_0$ one gets the desired notion. $\d M$ is referred to as the boundary of the graded manifold $M$.

In applications to gauge theories and, first of all, theories on manifolds with boundaries we are interested in generalised boundaries. Namely, let $M$ be a graded manifold whose body $M_0$ is a real manifold with boundary. A generalised boundary of $M$ is a submanifold $N \subset M$ such that $\body{N}=\d M_0$, where $\d M_0$ is the boundary of $M_0$. In this work we assume that on $M$ there are no coordinates of vanishing degree besides those originating from  $M_0$.

Given a graded manifold $M$ with a generalised boundary $N \subset M$ one can define the interior
of $M$ by restricting the structure sheaf to the interior $Int(M_0)$ of the body. In particular,
if $M$ is a manifold with a (generalised) boundary $N$ the interior does not depend on the choice of $N$ (of course by definition $\body{N}$ coincides with $\d (\body{M})$).
 \footnote{ To illustrate the concept of generalised boundary consider the following  example: $M=\fR_{\geq 0}\times \fR[1]\times \fR[-1]$, where $\fR_{\geq 0}$ is the halfline $\Omega \geq 0$. One can define distinct (generalised) boundaries, e.g.:
\begin{equation}
\begin{aligned}
 \d M:&\qquad \Omega=0\\
N_1:& \qquad \Omega=0\,, \quad \xi=0\\
N_2:& \qquad \Omega=0\,, \quad \xi=0\,, \quad \rho=0\,,\\
N_3:& \qquad \Omega+\xi\rho=0\,, 
\end{aligned}
\end{equation}
where $\xi$ and $\rho$ are coordinates on $\fR[1]$ and $\fR[-1]$ respectively. All of the above submanifolds are (generalised) boundaries whose bodies is just a point $\Omega=0$ of $\fR_{\geq 0}$. Note that $\d M$ and $N_3$ are usual (not generalised) boundaries but they are different, though isomorphic. Note also that if $M$ is non-negatively graded there is only one boundary which is determined by the usual boundary of the body. Indeed, by the Batchelor theorem  a nonegatively graded manifold is associated to a vector bundle over its body so restricting this bundle to the boundary of the body gives a vector bundle associated to the boundary. Let us mention that in this example the interior is given by $M=\fR_{>0}\times \fR[1]\times \fR[-1]$, where $\fR_{>0}$ is the interior of $\fR_{\geq 0}$, irrespective of the choice of (generalised) boundary.}

For our purposes the following specific version of generalised boundary is sufficient:
\begin{definition}
Let $(M,Q)$ be a $Q$-manifold with boundary, i.e. a graded manifold with boundary equipped with a homological vector field $Q$. A submanifold $\d^Q M\subset M$ of codimension $(1,1)$ is called a $Q$-boundary if $Q$ is tangent to $\d^Q M$, $\body{\d^Q M}=\d M_0, M_0=\body{M}$, and in a suitable  local coordinate system $(\Omega, \xi, z^A)$ near the boundary, $\d^Q M$ is singled out by 
\begin{equation}
        \Omega = 0\,, \qquad \xi = 0
\end{equation}
and, moreover, 
\begin{equation}
        Q\Omega = \xi\,.
\end{equation}
\end{definition}
Without loss of generality one can assume that $\Omega$ is a globally defined function such that 
its pullback to $M_0$ is positive on the interior $Int(M_0)$.
According to the above definition, $\Omega$ vanishes on $\d^Q M$ and hence on $\d M_0$ so that it is natural to call $\Omega$ a boundary defining function.

A basic and very useful example of a manifold with $Q$-boundary is as follows: let $M_0$ be a usual smooth manifold with boundary $\d M_0$. Then $M=T[1]M_0$ equipped with the de Rham differential seen as a vector field $\mathrm{d}_{M_0}$, is naturally a $Q$-manifold with $Q$-boundary given by $T[1]\partial M_0$. In what follows we denote the embedding $\d^Q M \hookrightarrow{} M$ by $b$ so that  
    \begin{equation} \label{embed-map}
        b^*\Omega = b^*\xi = 0\,.
    \end{equation}
Because $Q$ is by definition tangent to $\d^Q M$, $b$ is a $Q$-map and hence\maxim{corrected}
\begin{equation}
    Q b^* \alpha = b^* Q \alpha \,, \quad \forall \alpha \in \cC^\infty(T[1]M) \,.
\end{equation}

Given a manifold $M$ with (generalised) boundary
it is easy to see that $Int(M)$ does not depend on the choice of the generalised boundary (of course we assume that $M_0$ and its boundary $\d M_0$ are fixed). Note that in the case of a $Q$ manifold the interior is naturally a $Q$-manifold as well. We say that $\tilde{\alpha} \in \bigwedge^\bullet Int(M)$ extends to the boundary if $\tilde{\alpha} = \alpha|_{Int(M)}$ for some $\alpha \in \bigwedge^\bullet M$.

\subsection{Gauge PDEs with asymptotic boundaries}

Let us recall Penrose's definition of an asymptotically simple space \cite{Penrose:1962ij}: $(\tilde X,\tilde g)$,
where $\tilde g$ is a pseudo-Rie\-man\-ni\-an metric, is called asymptotically simple if there is another pseudo-Riemannian manifold $(X,g)$ with boundary such that  $\tilde X=Int(X)$ and on $Int(X)$ one has: $g=\Omega^2  \tilde g$ for some $\Omega \in \cC^\infty(X)$ satisfying  $\Omega=0$, $d\Omega \neq 0$ on $\d X$ and $\Omega>0$ on $Int(X)$.  Strictly speaking, there is also an additional completeness condition but it can be ignored in the present context. If in addition $\tilde g$ is subject to the vacuum Einstein equation the space-time is called asymptotically flat or (A)dS depending on the value of cosmological constant.

The idea is to think of metric $\tilde g$ as a particular solution of a certain gauge theory on $\tilde X$ and try to reformulate the notion in terms of the theory itself rather than its particular solution. In particular, this implies that there is a gauge theory defined on $X$ such that its restriction to $Int(X)$ is equivalent to the starting point theory. This vague idea can be formalised using the language of gauge PDEs as follows:

\begin{definition}\label{def:as-gPDE}
A gPDE $(E,Q,T[1]X)$ is called a gPDE with asymptotic boundary if
\begin{itemize}
    \item $X$ is a real manifold with boundary $\d X$.
    \item There is a $Q$-subbundle, denoted $\d_v^Q E \subset E$, such that its restriction to $T[1](Int(X))$ is a $Q$-boundary of $E|_{T[1](Int(X))}$
    and at the same time its restriction to $T[1]\d X$ is a $Q$-boundary of $E|_{T[1]\d X}$.
  \end{itemize}
  
    Let $(E,Q,T[1]X)$ be a gPDE with asymptotic boundary, 
\begin{itemize}
    \item 
We call  $(\d^Q(E|_{T[1]\d X}),Q,T[1]\d X)$ the induced boundary gPDE.
    \item

Gauge PDE $(Int(E|_{T[1](Int(X))}),Q,T[1](Int(X)))$ is denoted $Int(E)$ and called the interior of $(E,Q,T[1]X)$. 

\item
A section $\sigma: T[1]X\to E$ is a solution if $\sigma^*\circ Q=\dx \circ \sigma^*$
and its restriction to $T[1]\d X$ is a solution of the induced boundary gPDE $\d^Q (E|_{T[1]\d X})$ and the restriction of $\sigma$ to $T[1](Int(X))$ is a solution to $Int(E)$. Moreover,
$\dx \sigma^*(\Omega)$ does not vanish at $\d X \subset  T[1]X$, where $\Omega$ is a boundary defining coordinate on the fiber.
\end{itemize}
\end{definition}

In the case when the bundle $E$ of gPDE with asymptotic boundary is a globally trivial bundle this definition simply means that both fiber and the base possess a $Q$-boundary. The conditions on the solutions ensure that pullback by a solution of the boundary defining coordinate on the fiber gives a boundary defining function on the spacetime. 

\begin{definition}
Let $(E,Q,T[1]X)$ be a gPDE with asymptotic boundary.
\begin{enumerate}
\item A vector field on $E$ is called compatible if it is projectable, tangent to $\d_v^Q E \subset E$ and its projection is tangent to $\d^Q(T[1]X)$.

\item A symmetry is a compatible vector field $V$ on $E$ satisfying in addition $[Q,V]=0$.\footnote{Usually we require $\gh{V}=0$ but $V$ of nonvanishing degree are also of interest.}

\item A gauge parameter is a compatible vector field of degree $-1$. 

\item A gauge symmetry is a vector field of the form $V=[Q,Y]$, where $Y$ is a gauge parameter. It is clear that gauge symmetries  form a Lie subalgebra of all symmetries.

\end{enumerate}
\end{definition}

Given a (gauge) symmetry $V$, it naturally defines an infinitesimal  transformation of sections: 
\begin{equation}
\delta_V \sigma^*= \sigma^* \circ V-v\circ \sigma^*\,, \qquad v=\pi_*V\,.
\end{equation}
It is easy to check that it sends solutions to solutions. In a similar way one can handle gauge for gauge symmetries by taking gauge parameters of lower ghost degree, i.e. $\gh{Y_k}=-k, k>1$. The corresponding infinitesimal transformations do not act on sections but act on (gauge-for-) gauge parameters, e.g. $\delta_{Y_2} Y=\commut{Q}{Y_2}$.

\begin{definition} \label{def:compactification}
Let $\tilde X$ be a real manifold without boundary and $(\tilde E,Q,T[1]\tilde X)$ a gPDE. We say that a gPDE with asymptotic boundary $(E,Q,T[1]X)$ is its  compactification if
\begin{enumerate}
    \item $\tilde X =Int(X)$
    \item The interior (denoted as $(Int(E),Q,T[1] Int(X))$) of $(E,Q,T[1]X)$ restricted to \\ $T[1](Int(X))$ is equivalent to $(\tilde E, Q,T[1]\tilde X)$ as a gPDE.
\end{enumerate}
\end{definition}

Any solution of $E$ can be restricted to a solution of $(Int(E),Q,T[1] Int(X))$ and hence determines a solution (not necessarily unique) to $(\tilde E,\tilde Q,T[1] \tilde X)$ thanks to the equivalence. The converse is generally not true as a given solution of $\tilde E$ could diverge near the boundary. 

Similarly, any symmetry of $E$ defines a symmetry of $\tilde E$ but the converse is not true in general. For instance, a generic symmetry of $\tilde E$
does not necessarily arise as a restriction of a vector field from $E$. We also emphasize that for the given gPDE $(\tilde E,Q,T[1]\tilde X)$ there may exist different compactifications. In the examples presented later in this text, the choice of compactification reflects the choice of asymptotic behavior of the fields.

For our purposes it is often enough to restrict to gPDEs which are globally-trivial as $Q$-bundles. In this case the notion of a gPDE with asymptotic boundary simplifies as follows. Namely,  $(E,Q,T[1]X)$ is globally trivial, i.e.
\begin{equation}
\begin{gathered}
    (E,Q,T[1]X)=(T[1]X,\dx)\times (F,q)
\end{gathered}
\end{equation}
and $(F,q)$ is a $Q$-manifold with $Q$-boundary. In this case it is also natural to restrict ourselves to the situation where the initial gauge PDE is also globally-trivial, i.e. is given by $(T[1] \tilde X,\mathrm{d}_{\tilde X})\times (\tilde F,q)$
and $(\tilde F,q)$ is equivalent to $(F,q)$. Of course, it goes without saying that $\tilde X=Int(X)$.

\begin{rem}

    Suppose that $F$ is parameterized by coordinates $\{{z}^A,\Omega, \xi\}$
    and $\d^Q F$ is singled out by $\Omega=0,\xi=0$. Restricting these coordinates to $Int(F)$ gives the coordinate system therein. Let $\sigma: T[1]X \to F$ be a solution, meaning that $\dx \circ \sigma^* =\sigma^*\circ q$ and $\sigma^*(\Omega)$ vanishes on $\d X$ and $\sigma^*(\Omega)>0$ on $Int(X)$. Note that $\sigma^*(\xi)|_{\d^Q(T[1]X)}=0$ thanks to the equations of motion.

    Now, let $\tilde z^A$ be a coordinate system on $Int(F)$ and we do not assume that $\tilde z^A$  arise as restrictions of some coordinates on $F$. The restriction of $\sigma$ to $Int(X)$
    can be seen as a map $T[1] \tilde X \to Int(F)$.
Note that $\sigma^*(\tilde z^A)$ do not generally extend to $T[1](\d X)$ and one might need to rescale $\tilde z^A$ by a power of $\Omega$ in order for $\sigma^*(\tilde z^A)$ to be well-defined on the boundary. This is obvious in the example where on $Int(F)$ the coordinates are related by
$z^A= \Omega^k \tilde{z}^A$. Because $F$ is a compactification of $Int(F)$, it is clear that the choice of compactification determines the behavior of fields of $Int(F)$ as $\Omega \to 0$.

\end{rem}

\begin{rem} \label{rem:asymptotic-for-implicit}
In applications it is often convenient to employ gPDEs defined implicitly. In the globally-trivial setup which we restrict ourselves to, this means that the gPDE is defined as a submanifold $Z_\eI$ in $(F,q)$, which is the zero locus of a $q$-invariant ideal $\eI \subset \cC^\infty(F)$. In what follows, we denote such an implicit realisation of $(Z_\eI,q)$ by $(F,q,\eI)$.

The respective modification of Definition~\bref{def:as-gPDE} is that both $\tilde  E$ and $E$ are equipped with $Q$-invariant ideals denoted, respectively, as  $\tilde \eI$ and $\eI$ such that their zero loci are $Q$-subbundles and moreover
$(\tilde E, Q, \tilde \eI, T[1]X)$ is equivalent to $(Int(F),Q,\eI^\prime)$ where $\eI^\prime$ is the ideal generated by $\eI$ restricted to $Int(F)$.
\end{rem}

\subsection[Induced Q-cocycles on the boundary]{Induced $Q$-cocycles on the boundary}\label{sec : renorm}

Assuming that we deal with a globally-trivial system the analysis reduces to that of the fiber $(F,Q)$ which is a $Q$-manifold with $Q$-boundary. What we are interested in now is the near-boundary analysis so that without loss of generality we can assume $F$ to be  a sufficiently small  neighborhood of its $Q$-boundary. In this case one can take $\Omega$ and $\xi\equiv Q\Omega$ to be 
globally-defined so that the $Q$-boundary is a submanifold singled out by $\Omega=0$ and $\xi=0$. 

In applications we often encounter objects which are well-defined on $\tilde F \equiv Int(F)$ only. Suppose that we are given a $Q$-cocycle $\tilde v\in \bigwedge^\bullet(\tilde F)$, i.e. $L_Q\tilde v=0$, such that $v=\Omega^p \tilde v$ extends to the boundary and hence can be seen as a form on $F$. It turns out that under some mild technical assumptions such a cocycle determines a pair of cocycles of the boundary system:
\begin{equation}
        \begin{gathered}
            v^\cR \in \bigwedge^\bullet(\d^Q F)\,, \qquad  gh(v^\cR) = gh(\Tilde{v})\,, \\ 
            v^\cA \in \bigwedge^\bullet(\d^Q F)\,, \qquad gh(v^\cA) = gh(\Tilde{v}) - 1
        \end{gathered}
    \end{equation}
whose cohomology classes in $(\bigwedge^\bullet(\d^Q F),L_Q)$ are determined by the cohomology class of $\tilde v$ in $(\bigwedge^\bullet(\tilde F),L_Q)$.

\def\md{\mathfrak{d}}
To construct these cocycles we assume that $F$ admits a vector field $\md$, $\gh{\md}=-1$ such that  
\begin{equation} \label{md-defining-properties}
\md \Omega=0\,, \qquad \md \xi=1\,, \qquad \md^2=0\,.
\end{equation}
Recall, that we assumed that the boundary defining coordinate $\Omega$ is a fixed globally-defined function and $\xi\equiv Q\Omega$. That such $\md$ exists locally is obvious by taking $\md\equiv \dl{\xi}$ in a local coordinate system $\{\Omega,\xi,z^A\}$.  Its global (meaning in a neighborhood of the boundary) existence is a simplifying assumption which is anyway fulfilled in all our examples. 
Having chosen $\md$,  one defines another  vector field
\begin{equation}
    \nabla = [Q, {\md}]
\end{equation}
obeying $[Q, \nabla] = 0$. It follows that
\begin{align}
    \nabla\Omega=1,\quad \nabla\xi=0.
\end{align}
For any $v \in \bigwedge^\bullet(F)$ we define 
\begin{equation}\label{eq-v(N)}
\begin{gathered}
    v^{(N)} \equiv (L_{\nabla})^N v\,, \qquad \Bar{v}^{(N)} \equiv  (L_{\nabla})^{N} L_{\md} v\,.
\end{gathered}
\end{equation}

We have the following:
\begin{theorem}\label{th-renorm}
    Let $(F,Q)$ be a $Q$-manifold with $Q$-boundary, $b: \d^Q F \hookrightarrow F$ and $\Tilde{v} \in \bigwedge^{\bullet}(\tilde F)$ be a $Q$-cocycle such that $v = \Omega^p \Tilde{v}$ smoothly extends to $\d^Q F$ (i.e. $v$ is a restriction to $Int(F)$ of an element from $\bigwedge^\bullet(F)$). Then 
    \begin{equation}
    \label{RA-map}
    \begin{gathered}
        v^\cR = b^* \frac{1}{p!} v^{(p)}\,, \\ 
        v^\cA = b^* (-1)^{tot(v)}\frac{1}{(p-1)!} \Bar{v}^{(p-1)}\,.
    \end{gathered}
    \end{equation}
    are $Q$-cocycles in $\bigwedge^\bullet(\d^Q F)$ and their cohomology classes do not depend on the choice of $\md$ and $\Omega$. Here $tot(v) = \fdeg{v} + \gh{v}$ is the total degree of $v$.
\end{theorem}

\begin{proof}
    On $\tilde F$ one has
    \begin{equation} \label{eqOmegalq}
        \Omega L_Q v = \Omega L_Q (\Omega^p \Tilde{v}) = p \xi v + \Omega^{p+1} L_Q \Tilde{v} = p \xi v
    \end{equation}
    Applying $L_{\nabla}^{N+1}$ to this equation one gets
    \begin{equation}
        (N+1) L_Q v^{(N)} + \Omega L_Q v^{(N+1)} = p \xi v^{(N+1)}
    \end{equation}
    Because $v \equiv \Omega^p \tilde v$ can be seen as a form on $F$ and $\nabla$ is well-defined on $F$ both sides of the above equations are forms on $F$. Therefore, applying $b^*$ and $b^*\circ Q=Q \circ b^*$ one gets 
    \begin{equation}
        L_Q b^* v^{(N)} = 0\,,\qquad \forall N \geq 0\,.
    \end{equation}
    Although $b^* v^{(N)}$ is a $Q$ cocycle for any $N$, we are going to see later that it is trivial unless $N=p$.
    
   Applying $L_{\md}$ to \eqref{eqOmegalq} one gets
    \begin{align} \label{eq:first-step}
        \Omega L_Q\bar v=\Omega v^{(1)}-pv+p\xi\bar v\,.
    \end{align}
    Then, acting  by $b^{*}(L_\nabla)^N$ one gets
    \begin{align}
        NL_Qb^{*}\bar v^{(N-1)}=(N-p) b^* v^{(N)}\,.
    \end{align}
    This equation can be read in two ways: on the one hand, it states that $b^*v^{(N)}$ is $Q$-exact $\forall N \neq p$, on the other hand it verifies 
    \begin{equation}
        L_Q b^{*}\bar v^{(p-1)} = 0\,,
    \end{equation}
    showing the first part of the statement. 

    The independence of the cohomology class of $v^\cR, v^\cA$ on the choice of $\md$ is shown in Appendix~\bref{app:cocycles-ind}.
\end{proof}

\begin{rem} \label{rem:p-choice-independ}
    The number $p$ featured in the above theorem in $v=\Omega^p\tilde v$ is not required to be minimal, the resulting expressions for   $v^\cR$ and $v^\cA$ coincide for any such $p$ satisfying the conditions of the above Theorem.
\end{rem}

    Indeed, if $\Omega^p \tilde{v}$ is smoothly prolonged to the boundary, then so is $\Omega^{p+1} \tilde{v}$. The expression for $v^\cR$ calculated using the $p$ expression is 
    \begin{equation}
        v^\cR(p) = b^* \frac{1}{p!}L_\nabla^p \Omega^p \tilde{v}
    \end{equation}
    while 
    \begin{equation}
        v^\cR(p+1) = b^* \frac{1}{(p+1)!}L_\nabla^{p+1} \Omega^{p+1}\tilde{v} = b^*(\frac{1}{p!}L_\nabla^p \Omega^p \tilde{v} + \Omega \frac{1}{(p+1)!}L_\nabla^{p+1} \Omega^{p} \tilde{v} )\,,
    \end{equation}
    where we used $L_\nabla^2 \Omega = 0$. The second term is automatically zero giving \newline $v^\cR(p+1) =  v^\cR(p)$. A similar calculation proves the claim for $v^\cA$.

Let us explore the general properties of the map determined by~\eqref{RA-map}. It is convenient to define the algebra  $\bigwedge^\bullet_{\fop}(Int(F)) \subset \bigwedge^\bullet(Int(F))$ of forms on $Int(F)$ that are allowed to have poles of finite order as $\Omega \to 0$. More precisely,
\begin{equation} \label{def:poly-forms}
\bigwedge^\bullet{}_{\fop}(Int(F)) = \{\tilde\alpha \in \bigwedge^\bullet(Int(F)) | \exists p \in \mathbb{N}, \alpha \in  \bigwedge^\bullet(F) \, :\, \Omega^p \tilde\alpha =\alpha|_{Int(F)} \}\,,    
\end{equation}
where ``fop'' refers to ``finite-order pole''. We note that this definition, as well as the minimal possible value of $p$, does not depend on the choice of the boundary-defining function $\Omega$.  Any smooth vector field on $F$ restricts to a vector field on $Int(F)$, which in turn acts by Lie derivative on  $\bigwedge{}^\bullet_{\fop}(Int(F))$. Then we define the following maps:
\begin{equation} \label{R-map}
\begin{gathered}
    \cR: \bigwedge^\bullet{}_{\fop}(Int(F)) \xrightarrow{} \bigwedge^\bullet(\partial^Q F)\,,\qquad  
    \cR(\tilde{\alpha}) = b^*(\frac{1}{p!}\alpha^{(p)})\,,
\end{gathered}
\end{equation}
\begin{equation} \label{A-map}
\begin{gathered}
    \cA: \bigwedge^\bullet{}_{\fop}(Int(F)) \xrightarrow{} \bigwedge^\bullet(\partial^Q F)\,,  \qquad
    \cA(\tilde{\alpha}) =   b^*(\frac{(-1)^{tot(\tilde{\alpha})}}{(p-1)!}\Bar{\alpha}^{(p-1)})
\end{gathered}
\end{equation}
where $p$ is any number featured in \eqref{def:poly-forms} and $tot(\tilde\alpha) = form(\tilde\alpha) + ghost(\tilde\alpha)$ is the total degree of $\tilde\alpha$.
Recall that we assumed that $\Omega$ and $\md$ are fixed and globally defined. Remark~\bref{rem:p-choice-independ} implies that the above maps do not depend on the choice of $p$. In the setting of Theorem~\bref{th-renorm} one has: $\alpha^\cR= \cR(\tilde{\alpha})$ and $\alpha^\cA=\cA(\tilde{\alpha}) $. Note that maps~\eqref{R-map} and \eqref{A-map} are generally defined for all forms from $\bigwedge^\bullet{}_{\fop}(Int(F))$, not only $Q$-cocycles. Moreover, the following lemma is a straightforward generalization of Theorem~\bref{th-renorm}. 
\begin{lemma} \label{lemma:map-interwines}
    Maps $\cR$ and $\cA$ both commute with $L_Q$ and $d$.
\end{lemma}
The proof is given in Appendix~\bref{app:properties-maps}.

In applications we often encounter equivalence classes of forms modulo certain ideals. Let $\cI \subset \bigwedge^\bullet (Int(F))$ be a $Q$-invariant differential ideal, i.e. $\cI$ is closed under both $d$ and $L_Q$. We say that $\cI$ is extendable if it can be generated by a set of forms $\tilde f_i$ that extend to $\d^Q F$, i.e. $\tilde f_i = f_i|_{Int(F)}$
for some $f_i \in \bigwedge^\bullet (F)$ and moreover $L_{\md}\cI \subset \cI$.
Any extendable $\cI$ defines a differential  ideal $\cI_{\fop} \subset \bigwedge^\bullet_{\fop}(Int(F))$ generated by $\tilde f_i$. At the same time, $\cI$ also defines a differential  ideal $\cI_{\partial} \subset \bigwedge (\partial^Q F)$ generated by $b^* f_i$, where $\tilde f_i = f_i|_{Int(F)}$.

\begin{lemma} \label{lemma:nice-ideal}
    Let $\cI$ be extendable. Then  maps $\cR, \cA$ are well-defined on the corresponding  quotient algebras:
    \begin{equation}
    \begin{gathered}
        \cR: \bigwedge^\bullet{}_{\fop}/\cI_{\fop} \xrightarrow{} \bigwedge^\bullet(\partial^Q F)/\cI_{\partial} \\
        \cA: \bigwedge^\bullet{}_{\fop}/\cI_{\fop} \xrightarrow{} \bigwedge^\bullet(\partial^Q F)/\cI_{\partial}
    \end{gathered}
    \end{equation}
\end{lemma}
\begin{proof}
    The proof is relegated to Appendix~\bref{app:properties-maps}.
\end{proof}

\subsection{Interpretation of the induced cocycles}
We now provide an interpretation of the cocycle $v^\cR=\frac{1}{p!}b^{*} v^{(p)}$. First, note that on $Int(F)$ the form $v^{(p)}$ admits the following equivalent representation: 
\begin{prop}\label{prop:renorm-b}
    In the setting of Theorem~\bref{th-renorm} the following relation holds on $Int(F)$:
        \begin{align}\label{renorm-b}
    v^{(p)}=p!\tilde v+L_Q\left(\sum_{s=1}^{p-1}p(s-1)!\Omega^{-s}\bar v^{(p-s-1)}+\bar v^{(p-1)}\right)-p\frac{\xi}{\Omega}\bar v^{(p-1)}\,,
\end{align}
where, by convention, $\sum_{s=1}^{0}(\dots)=0$ for $p=1$.
\end{prop}
The proof of this Proposition is given in Appendix~\bref{app:renorm}.
The proposition~\bref{prop:renorm-b} can be interpreted by saying that, up to $Q$-exact terms, the polynomially divergent form $\tilde v$ can be represented as the sum of $v^{(p)}$ and a term proportional to $\frac{\xi}{\Omega}$. Moreover, the coefficient of the latter term is directly related to the anomalous boundary cocycle; see \eqref{RA-map}.

\begin{rem}
To make contact with the b-calculus of Melrose (see e.g. \cite{melrose1993atiyah}) 
let us consider a simple example of a $Q$-manifold with $Q$-boundary given by $(M,Q)=(T[1]X,\dx)$, where $X$ is a smooth manifold with boundary $\partial X$. In this case $Q$-cocycles (in functions) of $T[1]X$ are just de Rham cocycles of $X$. Let $\Omega$ be a boundary defining coordinate with $\dx\Omega = \xi$ and consider a cocycle $\tilde v$ on the interior with $p=1$. Using \eqref{renorm-b} we obtain:
    \begin{align}
        \tilde v= \frac{\mathrm \dx\Omega}{\Omega}\bar v^{(0)}+(v^{(1)}-L_{d_X}\bar v^{(0)})\,.
    \end{align}
    This is the standard decomposition of a $b$-form near the boundary and the mapping of the cohomology is the content of the Mazzeo-Melrose theorem which states that 
    \begin{equation} \label{Melrose-decomposition}
        {}^bH^k(X) = H^k (X) \oplus H^{k-1}(\partial X)\,.
    \end{equation}
\end{rem}
Observe that the last term in \eqref{renorm-b} can be rewritten as follows:
\begin{align}\label{remaider-b}
    -p\frac{\xi}{\Omega}\bar v^{(p-1)}=-pL_Q(\ln\Omega\bar v^{(p-1)})+p\ln\Omega L_Q\bar v^{(p-1)}.
\end{align}
In Appendix~\bref{app:renorm}, we show that, since $\bar v^{(p-1)}$ is a $Q$-cocycle modulo terms vanishing at the boundary, the last term in \eqref{remaider-b} can be rewritten iteratively as the sum of a $L_Q$-exact term and a remainder. This leads to the following result:
\begin{corollary}\label{prop-renorm}
    In the setting of Theorem~\bref{th-renorm} the following relation holds on $Int(F)$ for every integer $N\geq1$:
\begin{multline}\label{prop-renorm-formula}
    v^{(p)}=p!\tilde v\,+ \\
    L_Q\left(\sum_{s=1}^{p-1}p(s-1)!\Omega^{-s}\bar v^{(p-s-1)}+\bar v^{(p-1)}-p\sum_{j=0}^{N}\frac{(-1)^j}{j!}\Omega^{j}\ln\Omega\bar v^{(p-1+j)}\right)+r_N[v]\,,
\end{multline}
where
\begin{multline}
        r_N[v]=p\xi\sum_{i=1}^{N}\frac{(-1)^{i}}{i!}\Omega^{i-1}\bar v^{(p+i-1)}\,+
        \\
        \frac{(-1)^N}{N!}\Omega^N\ln\Omega((p+N)L_Q\bar v^{(p+N-1)}-Nv^{(p+N)})\,.
\end{multline}
\end{corollary}

We note two properties of the remainders $r_N[v]$ that are important for us. First, for any $N$, they have a well-defined limit at $\d^Q F$-boundary and vanish there.  Although by assumption $v\equiv \Omega^p \tilde v$ is smooth, the corresponding $r_N[v]$ is only $(N-1)$ times differentiable because of the term $\Omega^N\ln\Omega$.

Corollary~\bref{prop-renorm} shows that, for a given $Q$-cocycle $\tilde v \in \bigwedge^\bullet_{\fop}(Int(F))$, the induced cocycle $v^\cR \equiv \cR(v)$ can be regarded as the renormalized value of the singular form $\tilde v$ at the boundary, where the renormalization amounts to the addition of $L_Q$-exact terms.
\begin{rem}
    Corollary~\bref{prop-renorm} can be interpreted in the spirit of the holographic renormalization procedure known in the AdS/CFT context \cite{deHaro:2000xn, Skenderis:2002wp}. To illustrate this, we consider a globally trivial gPDE with asymptotic boundary $(T[1]X,\mathrm{d}_{X})\times (F,q)$, where
    \begin{align}
        X=\mathbb{R}_{\geq 0}\times\Sigma, \qquad \partial X\cong\Sigma.
    \end{align}
    Let $v$ be a function of degree $\dim(X)$ on $F$, and consider a solution $\sigma$ such that $\rho=\sigma^{*}\Omega$, where $\rho$ is a boundary defining function on $\mathbb{R}_{\geq 0}$. Defining $X_{\varepsilon}\equiv\{p\in X|\rho(p)\geq\varepsilon\}\subset X$ we obtain from 
    \eqref{prop-renorm-formula} 
    \begin{multline}
      \frac{1}{p!}  \int_{T[1]X_\epsilon}\sigma^{*} (v^{(p)}-r_N[v])= \int_{T[1]X_\varepsilon}\sigma^{*}\tilde v+\sum_{s=1}^{p-1}\frac{(s-1)!}{(p-1)!}\varepsilon^{-s}\int_{\rho=\varepsilon}\sigma^{*}\bar v^{(p-s-1)}-\\\frac{1}{(p-1)!}\ln\varepsilon\int_{\rho=\varepsilon}\sigma^{*}\bar v^{(p-1)}+\mathcal{O}(\varepsilon^{0})\,,
    \end{multline}
    where we have used that $\sigma$ is a solution, i.e. that  $\mathrm{d}_{X}\circ\sigma^{*}= \sigma^*\circ Q$, and hence $\int \sigma^* (L_Q \alpha)=\int_{T[1]X_\varepsilon}\d_X \sigma^*(\alpha)=\int_{\rho=\varepsilon}\sigma^*(\alpha)$.
The left-hand side of this formula can be  interpreted as the renormalization of $\int_{T[1]X_\varepsilon}\sigma^{*}\tilde v$ as it has a well-defined limit as $\varepsilon\to 0$ and differs from 
$\int_{T[1]X_\varepsilon}\sigma^{*}\tilde v$ by the $\ln \varepsilon$ term and terms with poles of positive order in $\varepsilon$. The former is interpreted as the anomaly while the latter are the  boundary counterterms. In the context of holographic renormalization, one typically renormalizes the on-shell action, in which case the anomaly is known as the holographic Weyl anomaly. This explains our notation for the  $\cR$ (renormalization) and $\cA$ (anomaly) maps.

Finally, we note that in the limit $\varepsilon\to0$ the expression \eqref{prop-renorm-formula} does not depend on the integer $N$ since, as shown in Appendix~\bref{app:renorm}, we have
\begin{align}
    r_N[v]=pL_Q(\frac{(-1)^N}{(N+1)!}\Omega^{N+1}\ln\Omega\bar v^{(p+N)})+r_{N+1}[v]
\end{align}
and hence the difference gives the boundary term that vanishes in the $\varepsilon\to 0$ limit.
\end{rem}


\subsection{Renormalization of presymplectic structures}\label{sec: renorm-struc}

Our focus is Lagrangian gauge theories. In the gPDE approach Lagrangians are encoded in compatible presymplectic structures defined on the total space of a gPDE. Let us assume that the system in question  is Lagrangian and let $(E,Q,T[1]X)$ be a compactification of its gPDE.  Identifying $(Int(E|_{T[1]Int(X)}),Q,T[1](Int(X)))$ with the initial gPDE we assume that it is equipped with a compatible presymplectic structure, i.e. a 2-form
$\tilde\omega$ satisfying 
\begin{equation}
\label{omegaQ}
\begin{gathered}
    i_Q \tilde{\omega} + d\tilde{H} \in \tilde\cI \,,\qquad
    d\tilde\omega = 0 \,,\qquad
    gh(\tilde\omega) = dim X - 1\,,
\end{gathered}
\end{equation}
where $\tilde \cI$ is the ideal generated by the pullback of forms of positive degree from the base $T[1]Int(X)$ or, more generally, the corresponding background gPDE. In the case of implicit gPDEs it is convenient to assume that $\tilde\cI$ also involves the ideal of the equations. Note that \eqref{omegaQ} implies 
\begin{equation}
    L_Q \tilde{\omega} \in \tilde\cI\,.   
\end{equation}
To proceed further we also need to assume that $\tilde \omega$ and $\tilde{H}$ belong to $\bigwedge^\bullet_{\fop}(Int(E))$. It follows that one can safely replace $\tilde\cI$ with $\tilde\cI_{\fop}\equiv \tilde\cI \cap \bigwedge^\bullet_{\fop}(Int(F))$.
As we are going to see, this assumption is fulfilled in applications.

Now we are in the setting of Lemmas~\bref{lemma:map-interwines} and \bref{lemma:nice-ideal} which  imply that $\tilde \omega$ induces two presymplectic structures on the boundary: $\omega^{\cR} = \cR(\tilde{\omega})$ and $\omega^\cA = \cA(\tilde{\omega})$. By construction one has $gh(\omega^\cR) = dim X - 1 = dim(\partial X)$ and $gh(\omega^\cA) = dim(\partial X) - 1$. Moreover, picking any presymplectic potential $\tilde\chi$ such that $\tilde\omega = d\tilde\chi$ one gets
\begin{equation}
    \begin{gathered}
        \omega^\cR = d\cR(\tilde{\chi})\,, \qquad L_Q \omega^\cR \in \cI_\partial \\ 
        \omega^\cA = d\cA(\tilde{\chi})\,, \qquad L_Q \omega^\cA \in \cI_\partial \,,
    \end{gathered}
\end{equation}
where $\cI_\d$ is the ideal generated by the pullback of the generators of $\tilde \cI_{\fop}$ to the boundary.

Because both $\omega^\cR$ and $\omega^\cA$ are $Q$-invariant modulo the ideal there exist Hamiltonians $H_l^\cR$ and $H^\cA_l$ such that $i_Q \omega^\cR+dH_l^\cR\in\cI_\d$ and 
$i_Q \omega^\cA+dH_l^\cA\in\cI_\d$. Let us stress that $H^{\cR,\cA}_l$ are apparently different from $H^\cA \equiv \cA(\tilde H)$ and $H^\cR\equiv \cR(\tilde H)$, respectively.
Given the presymplectic potential and the Hamiltonian, these data  determine BV-AKSZ action functional on the space of supersections of the induced boundary gPDE. Assuming that the presymplectic structures induced on the space of supersections by $\omega^\cR$ and $\omega^\cA$ are regular one can pass to the respective symplectic quotients. In the case of $\omega^\cR,H_l^\cR$ the respective BV-AKSZ action has degree $1$ and is to be interpreted as the BFV-BRST charge of the BFV system induced on the quotient. It is natural to interpret this system as the renormalised BFV system induced on the asymptotic boundary.
At the same time $\omega^\cA$ has the standard ghost degree and the associated BV-AKSZ action can be identified with the BV extension of the holographic Weyl anomaly, see e.g.~\cite{Skenderis:2002wp}, if one passes to the corresponding symplectic quotient. Recall that this anomaly can be also seen as an action functional 
whose equations of motion is the 
Fefferman-Graham obstruction.

Although $H_l^\cR$ and $H^\cA_l$ can be found directly by solving e.g. $i_Q\omega^\cA+d H_l^\cA\in \cI_\d$ this is not an efficient technique. An efficient way to obtain $\omega^{\cR,\cA}$ and $H^{\cR,\cA}_l$ in one go is to employ the relation~\eqref{chi-l-axioms} which in the case at hand reads as:
\begin{equation}
    \tilde\omega = (d + L_Q)(\tilde\chi + \tilde{l}) + \tilde\cI|_{\fop}
\end{equation}
where $\tilde{l} = i_Q\tilde{\chi} + \tilde{H}$. Indeed, restricting this equation to form degree 1, we find 
\begin{equation}
    L_Q(\tilde\chi) + d(\tilde{l}) \in \tilde\cI|_{\fop}
\end{equation}
which can be rewritten as 
\begin{equation}
    i_Q d\tilde{\chi} + d(\tilde{l} - i_Q \tilde{\chi}) \in \tilde\cI|_{\fop}\,.
\end{equation}
see \cite{Basile:2026kyx} for further details. Note that one can directly express the BV-AKSZ action in terms of $\tilde{\chi} + \tilde{l}$ using the so-called Chern-Weyl map introduced in \cite{Kotov:2007nr}, see also \cite{Basile:2026kyx}.

Using Lemmas~\bref{lemma:map-interwines} and~\bref{lemma:nice-ideal} one immediately finds that the induced presymplectic structures on the boundary obey
\begin{equation}
\begin{gathered}
    \omega^\cR = (d+L_Q)\cR(\tilde\chi+\tilde{l}) + \cI_\partial\,,\\
    \omega^\cA = (d+L_Q)\cA(\tilde\chi+\tilde{l}) + \cI_\partial\,.
\end{gathered}    
\end{equation}
Writing equations in this form allows one to immediately confirm that 
\begin{equation}
    i_Q \omega^\cR + d(l^\cR - i_Q\chi^\cR)=i_Q \omega^\cR + d(H^\cR) \in \cI_\partial\,,
\end{equation}
\begin{equation}
    i_Q \omega^\cA + d(l^\cA - i_Q\chi^\cA) \in \cI_\partial\,.
\end{equation}
It follows that the Hamiltonian of $Q$ w.r.t. $\omega^\cR$ is just given by $H_l^\cR = \cR(\tilde{H})\equiv H^\cR$.

As for the Hamiltonian $H^{\cA}_l$ of $Q$ with respect to $\omega^\cA$, it can be explicitly found in terms of $\tilde{H}, \tilde\chi$ without referring to $\tilde{l}$. Namely:
\begin{multline}\label{anomaly-hamiltonian}
         H^{\cA}_l=l^\cA - i_Q\chi^\cA \,,= \frac{(-1)^{dimX}}{(p-1)!}b^*(L^{p-1}_{\nabla} L_{\md} \Omega^p (\tilde{H} + i_Q \tilde{\chi}) - i_QL_{\md} L^{p-1}_{\nabla} \Omega^p \tilde{\chi})  = \\ 
      =   \frac{(-1)^{dimX}}{(p-1)!}b^*\Big(L^{p-1}_{\nabla} L_{\md} \Omega^p (\tilde{H}) + L^{p-1}_{\nabla} [L_{\md},i_Q]\Omega^p \tilde\chi\Big)\,\, =\\
         \frac{(-1)^{dimX}}{(p-1)!}b^*L^{p-1}_{\nabla}\Big( L_{\md}  {H} + i_\nabla \chi\Big)
\end{multline}
where $H\equiv \Omega^p\tilde H$, $\chi=\Omega^p\tilde\chi$ for sufficiently large $p$. 

\section{Applications and examples: fields on (A)dS space}
\label{sec:AdS}

In this section we study the boundary structure of gauge fields on (A)dS space. This is done by treating the (A)dS structure as a solution of ``flat AdS gravity'', i.e. the truncation of the full gravity where the Weyl tensor is set to zero. In the next step the (gauge) fields on (A)dS space are reformulated as fields defined on the background of flat (A)dS gravity or, more precisely, its conformal compactification. In the gPDE language this corresponds to $Q$-bundles over the $Q$-manifold describing flat (A)dS gravity. We emphasize that the choice of a flat background in the present work is motivated solely by its simplicity. The extension to a curved background is straightforward using the techniques developed in \cite{Grigoriev:2025gvk}.

\subsection{Flat (A)dS gravity and its compactification}\label{sec:AdS-gravity}

We start with a gPDE formulation of flat (A)dS gravity. Because gravity as well as its flat limit is diffeomorphism-invariant the space-time manifold plays a rather passive role and can be effectively disregarded. In the gPDE approach this can be seen as follows:  diffeomorphism-invariant systems are described by gPDEs which are locally-trivial $Q$-bundle, i.e. locally
\begin{equation}
(E,\dx +Q,T[1] X)=(\tilde F,Q)\times (T[1] X,\dx)\,,
\end{equation} 
see e.g.~\cite{Barnich:1995db,Barnich:2010sw,Grigoriev:2019ojp} for more details. To simplify the exposition we also restrict ourselves to globally-trivial $Q$-bundles. It is clear from the above representation that nearly all the information about the system is encoded in the fiber $(\tilde F,Q)$.

The minimal description of the flat AdS gravity is achieved by taking the fiber $Q$-manifold $(\tilde F,Q)$ to be $(\mathfrak g_\Lambda[1],\mathrm{d}_{\algg})$, where $\mathfrak g_\Lambda$ is the (A)dS algebra and $\mathrm{d}_{\algg}$ is its Chevalley-Eilenberg differential seen as a homological vector field on $\mathfrak g_\Lambda[1]$. It is clear that solutions to this gPDE are flat $\mathfrak g_\Lambda$-connections. Assuming the nondegeneracy condition on the soldering form component of the connection one concludes that such a connection defines the (A)dS space structure on $X$, at least locally.

Now we are interested in the compactification of 
$(\mathfrak g_\Lambda[1],\mathrm{d}_{\algg}) \times T[1]X$ which determines the usual conformal compactification of (A)dS space at the level of solutions. To this end, consider the linear graded manifold with coordinates
\begin{equation}\label{AdS-gr-coordinates}
  \Omega,n_a,\tau\quad \Omega \geq 0\,,\,\,  
  n_\Omega>0
  \,, \qquad \quad  \xi^{a},\quad\rho^{ab},\quad \lambda,\quad \lambda^a,\quad \rho^{ab}=-\rho^{ba}\,,
\end{equation}
where coordinates in the first group are of ghost degree $0$, those from the second have degree $1$ and $\epsilon = \pm 1$, with the sign corresponding to the sign of the cosmological constant.  Here, we also assumed that index $a$ is split  as $\{a\}=\{\Omega,A\}$ so that one direction is distinguished.
The action of $Q$ is given by:
\begin{align} \label{Q-int-F}
\begin{split}
        &Q\xi^a=\xi^b\rho_{b}{}^{a}-\xi^a\lambda\,,\qquad Q\lambda=\xi^a\lambda_a\,,\\
    &Q\lambda^a=\lambda^{b}\rho_b{}^a-\lambda\lambda^a\,,\qquad
    Q\rho^{ab}=\rho^{ac}\rho_c{}^b+\delta^{ab}_{cd}\lambda^c\xi^d\,,\\
    &Q\Omega=\xi^a n_a+\lambda\Omega,\qquad Qn_a=-\xi_a\tau+\rho_{a}{}^{b}n_b+\lambda_a\Omega\,
    ,\qquad Q\tau=-\lambda\tau-\lambda^an_a.
\end{split}
\end{align}

The desired manifold $(F,Q)$ is the zero-locus of the constraint: 
\begin{align} \label{scale-ideal}
    \Omega\tau+\frac{1}{2}n_a n^a-\frac{\epsilon}{2}\,.
\end{align}
Here, the indices are raised and lowered by the constant metric 
\begin{equation}
g_{ab}=diag(\epsilon,-\epsilon,-1,\dots,-1)\,,
\end{equation} 
The ideal generated by \eqref{scale-ideal} is denoted by $\eI_{Sc}$.  The zero locus is clearly a smooth submanifold and $Q$ is tangent to it, giving a $Q$-manifold $(F,Q)$. Instead of solving the above constraint explicitly, it can be often convenient to identify the algebra of functions on $F$ with the quotient of the algebra of functions on~\eqref{AdS-gr-coordinates} by the ideal generated by \eqref{scale-ideal}. In other words, we work using global coordinates ~\eqref{AdS-gr-coordinates} and all relations are understood modulo $\eI_{Sc}$. This can be formalised by using the notion of implicit $Q$-manifold, introduced in~\cite{Grigoriev:2025gvk}.

The manifold $(F,Q)$ has the $Q$-boundary singled out by\footnote{As a submanifold in \eqref{AdS-gr-coordinates} it is singled out by $\Omega = 0$, $Q\Omega = 0$ and $\Omega\tau+\frac{1}{2}n_a n^a-\frac{\epsilon}{2}=0$.}
\begin{equation}
    \Omega = 0\,, \quad \xi = 0\,, \qquad \xi\equiv Q\Omega =\xi^a n_a + \lambda \Omega\,.
\end{equation}
Because  $n_\Omega>0$, the function $\xi$ defined in this way can be taken as the coordinate $\xi$ appearing in the definition of $Q$-boundary.

\begin{prop}\label{prop-Ads-equiv}
    $(Int(F),Q)$ is equivalent as a $Q$-manifold to $(\mathfrak g_\Lambda[1],\mathrm{d}_{\algg})$  where $\mathfrak g_\Lambda$ is the (A)dS algebra and $\mathrm{d}_{\algg}$ is its Chevalley-Eilenberg differential.
\end{prop}
\begin{proof}
Coordinates \eqref{AdS-gr-coordinates} induce coordinates on $Int(F)$, where one can replace $\xi^a,\rho^{ab}$ with 
    \begin{align} \label{eq:trans-functions-ads-gr}
        \tilde\xi^a\equiv\Omega^{-1}\xi^a,\quad \tilde\rho_{a}{}^{b}\equiv\rho_{a}{}^{b}-\Omega^{-1}n_a\xi^b+\Omega^{-1}n^b\xi_a\,,
    \end{align}
giving a new coordinate system on $Int(F)$. The action of $Q$ on the new coordinates is given by:
    \begin{align} \label{eq:Q-AdS-in-tilde-coord}
    \begin{split}
                Q\tilde\xi^{a}&=\tilde\xi^{c}\tilde\rho_{c}{}^{a},\\
        Q\tilde \rho^{bc}&=\tilde\rho^{bd}\tilde\rho_{d}{}^{c}+\tilde \xi^b\tilde \xi^c (2\Omega\tau+n^an_a)=\tilde\rho^{bd}\tilde\rho_{d}{}^{c}+\tilde \xi^b\tilde \xi^c \epsilon \,,\\
            \end{split}
    \end{align}
where in the second relation we made use of $\eI_{Sc}$. 
The restrictions of $\xi^a$, $\rho^{bc}$, $\lambda$, $\lambda^a$, $\Omega$, and $n^a$ to $Int(F)$ are independent coordinates therein because $\tau$ can be solved for. 

It turns out that the reduced system is precisely $(\mathfrak g_\Lambda[1],\mathrm{d}_{\algg})$. Indeed, the subalgebra of functions on $Int(F)$ 
generated by $\tilde{\xi},\tilde{\rho}$ can be identified with 
the functions pulled back from $(\mathfrak g_\Lambda[1],\mathrm{d}_{\algg})$, giving a projection $\pi_{\algg_\Lambda}:(Int(F),Q) \to (\mathfrak g_\Lambda[1],\mathrm{d}_{\algg})$. Moreover, \eqref{eq:Q-AdS-in-tilde-coord} implies that $\pi_{\algg_\Lambda}$ is compatible with $Q$ so that 
$\pi_{\algg_\Lambda}$ is a $Q$-map.
In terms of the initial coordinates $\xi^a,\rho^{bc},\lambda,\lambda^a,\Omega,n_a$ the projection is given by 
\begin{equation} \label{pi_g}
    \begin{gathered}
        \pi_{\algg_\Lambda}^* \tilde{\xi}^a = \Omega^{-1}\xi^a\,,\qquad
       \pi_{\algg_\Lambda}^* \tilde\rho_{a}{}^{b}\equiv\rho_{a}{}^{b}-\Omega^{-1}n_a\xi^b+\Omega^{-1}n^b\xi_a\,,
    \end{gathered}
\end{equation}
where we consider $\tilde{\xi}^a$ and $\tilde\rho_{a}{}^{b}$ as coordinates on $\algg_\Lambda[1]$. Finally, one can equivalently reduce $(Int(F),Q)$ to its submanifold singled out by $\Omega=1,Q\Omega=0,n^a=n_0^a, Qn_a=0$ so that $(\mathfrak g_\Lambda[1],\mathrm{d}_{\algg})$ is an equivalent reduction of $(Int(F),Q)$. 
Here $n_0^a$ is a fixed constant vector. Note that different choices of $n_0^a$ lead to the same (up to diffeomorphisms) reduced $Q$-manifolds.
\end{proof}
The above proposition gives a systematic way to uplift various objects such as $Q$-cocycles, bundles, etc. defined on $(\algg_\Lambda,\mathrm{d}_{\algg})$
to the interior of its conformal compactification $(F,Q)$, by pulling it back by the projection~\eqref{pi_g}. Because $\pi_{\algg_\Lambda}^*$ originates from the equivalent reduction it determines a quasi-isomorphism of the corresponding complexes.

Let us now get back to the study of $(F,Q)$ which describes the conformal compactification of flat (A)dS gravity. It is convenient to consider an equivalent reduction $(F_\red,Q)$ of $(F,Q)$. More precisely, consider the following submanifold $F_{\red}\subset F$
which is explicitly defined as a submanifold of \eqref{AdS-gr-coordinates}, singled out by
\begin{equation}
n_A=0\,, \quad \rho^{\Omega A}=\lambda^A\Omega\,, \qquad \tau=0,\quad \lambda^\Omega =0\,, \qquad 
n_\Omega=1\,,
\end{equation}
where the last constraint is obtained by solving $\eI_{Sc}$.
\begin{prop}
$(F_\red,Q)$ is an equivalent reduction of $(F,Q)$
\end{prop}
Indeed, thanks to $n_\Omega>0$, functions $n_A,Qn_A$ and $\tau,Q\tau$ form contractible pairs. Moreover, they remain such on $F$ as well. Introducing $\xi\equiv\xi^\Omega+\lambda\Omega$ one finds that  
\begin{align}
\label{Fred-coord}
    \xi,\xi^A,\rho^{AB},\lambda,\lambda^A,\Omega
\end{align}
restricted to $F_{\red}$ form a coordinate system therein. In this coordinate system the action of 
$Q$ reads as:
\begin{align}\label{Q-M-red}
    \begin{split}
        Q\Omega&=\xi,\qquad Q\xi=0,\\
        Q\xi^A&=\xi^B\rho_{B}{}^{A}-\xi^A\lambda+\epsilon(\xi-\lambda\Omega)\lambda^A\Omega \,,\\
        Q\rho^{AB}&=\rho^{AC}\rho_{C}{}^{B}+\delta^{AB}_{CD}\lambda^{C}\xi^{D}-\epsilon\lambda^A\lambda^B\Omega^2\,,\\
        Q\lambda&=\xi^A\lambda_A,\quad Q\lambda^A=\lambda^B\rho_{B}{}^{A}-\lambda\lambda^A\,.
    \end{split}
\end{align}
We can also explicitly construct the $Q$-projection $\pi^\red_{\algg_\Lambda}:Int(F_{\red})\to \algg_\Lambda$ as the composition of $\pi_{\algg_\Lambda}$ with the embedding of $F_{\red}$ into $F$. In coordinates, it takes the form
\begin{align}\label{pi-red}
\begin{alignedat}{4}
(\pi^{\red}_{\algg_\Lambda})^{*}\tilde \xi^\Omega
    &= \Omega^{-1}\xi-\lambda\,,
&\qquad
(\pi^{\red}_{\algg_\Lambda})^{*}\tilde\xi^A
    &= \Omega^{-1}\xi^A\,,
\\
(\pi^{\red}_{\algg_\Lambda})^{*}\tilde\rho_{A}{}^{B}
    &= \rho_{A}{}^{B}\,,
&\qquad
(\pi^{\red}_{\algg_\Lambda})^{*}\tilde\rho_{\Omega}{}^{A}
    &= \epsilon\lambda^A\Omega-\Omega^{-1}\xi^A\, .
\end{alignedat}
\end{align}

The $Q$-boundary $(\d^Q F,Q)$ of $F$ is singled out by $\Omega=0$ and $Q\Omega=0$. We have the following:  
\begin{prop}\label{prop-ads-reduction}
$(\d^Q F,Q)$ is equivalent to $(\algg[1],\mathrm{d}_{\algg})$, where $\algg$ is a conformal algebra in $d$ dimensions and $\mathrm{d}_{\algg}$ is its CE differential.
\end{prop}

\begin{proof}
Instead of performing the equivalent reduction of $(\d^Q F,Q)$ one can first equivalently reduce $F$ to $F_\red$ and then take the $Q$-boundary. The result is easily seen to be the same. The $Q$-boundary $\partial^Q F_{\red}$ of $(F_\red,Q)$ is singled out by $\Omega=\xi=0$ and we denote by $b:\d^QF_{\red}\hookrightarrow F_{\red}$ the corresponding embedding map. Introducing the notation $\hat f\equiv b^{*}f\,, \forall f \in  C^\infty(F_{\red})$, we choose as coordinates on $\partial^QF_{\red}$ the following functions:
\begin{align}
    \hat \xi^A,\hat \rho^{AB}, \hat \lambda, \hat \lambda^A\,,
\end{align} 
where $\xi^A,\rho^{AB},  \lambda,  \lambda^A$ are part of the coordinate system \eqref{Fred-coord} on $F_\red$. In these coordinates the action of $Q$ reads as (cf. \eqref{Q-M-red}):
\begin{align}
    \begin{split}
        Q\hat \xi^A&=\hat \xi^B\hat \rho_{B}{}^A-\hat \xi^A\hat \lambda\,,\quad \\
        Q\hat \rho^{AB}&=\hat \rho^{A}{}_{C}\hat \rho^{CB}+\delta^{AB}_{CD}\hat \lambda^{C}\hat \xi^D\,,\\
        Q\hat \lambda&=\hat\xi^A\hat\lambda_A,\quad Q\hat \lambda^A=\hat \lambda^B\hat \rho_{B}{}^{A}-\hat \lambda\hat \lambda^A\,.
    \end{split}
\end{align}
One can readily see that $\d^Q F_{\red}$ is the Chevalley-Eilenberg  complex for the conformal algebra in $d\equiv D-1$ dimensions.
\end{proof}
The following technical proposition will be useful for us in what follows:
\begin{prop}\label{prop-calculus-ads}
Let  
$\nabla \equiv [Q,\md ]$, where $\md=\frac{\partial}{\partial \xi}$ is defined using the coordinate system \eqref{Fred-coord}, be a vector field on $F_{\red}$. Then for $N>0$ one has:
   \begin{align}
    \begin{split}
                &b^{*}\nabla^{N}\lambda^A=0,\quad b^{*}\nabla^N\lambda=0\,,\\
        &b^{*}\nabla^N\xi^A=\delta_{N,2}\epsilon\hat \lambda^A, \quad b^*\nabla^N\rho^{AB}=0\,,\\
        &b^{*}\nabla^N\xi^\Omega=-\delta_{N,1}\hat\lambda\,,
            \end{split}
    \end{align}
    where $b: \d^Q F_\red \hookrightarrow F_\red$ and $\xi^\Omega\equiv  \xi-\lambda\Omega$.
\end{prop}
The proof is straightforward and based on using \eqref{Q-M-red}. In practice, we will also frequently need to compute expressions of the form $b^{*}\nabla^k \xi{}\ms{N}_{a_1\dots a_{D-N}}$. The corresponding combinatorics is collected in Appendix~\bref{app: Faa}.

Using this statement, we can illustrate the map of cocycles from Theorem~\bref{th-renorm} using the following simple example:
\begin{example}\label{ex-cocycleads}
    Consider the following cocycle on $\mathfrak{g}_\Lambda[1]$
    \begin{align}
        \tilde f\equiv \frac{1}{D!}\epsilon_{a_1\dots a_D}\tilde\xi^{a_1}\dots \tilde\xi^{a_{D}}\,.
    \end{align}
    By Proposition~\bref{prop-Ads-equiv}, this cocycle can be uplifted to $Int(F)$ and subsequently pulled back to  $Int(F_{\red})$, yielding
    \begin{align}
        \tilde f=\Omega^{-D}\frac{1}{D!}\epsilon_{a_1\dots a_D}\xi^{a_1}\dots \xi^{a_{D}}\,.
    \end{align}
    Observing that $f\equiv \Omega^D\tilde f$ admits a smooth extension to the boundary, we can apply Theorem~\bref{th-renorm} and Proposition~\bref{prop-calculus-ads} to obtain cocycles $f^\cA, f^\cR$ on $\partial^QF_{\red}$. For example, for $D=5$, we have
    \begin{align}
        f^\cA=-\frac{1}{16}\hat\lambda^A\hat\lambda^B\epsilon_{ABCD}\hat\xi^C\hat\xi^D,\quad f^\cR=\hat\lambda f^\cA\,.
    \end{align}
    It is straightforward to verify that these are indeed $Q$-cocycles.
    \end{example}

\subsubsection{Field-theoretical interpretation}\label{Ads-sections}
Just like the system $(\algg_\Lambda[1],\mathrm{d}_{\algg})\times (T[1]X,\dx)$, the boundary gPDE $(\partial^Q F_{\red},Q)\times (T[1]\d X,d_{\d X})$ describes flat connections valued in the orthogonal algebra. However, their field-theoretic interpretation differs. The most evident difference is that the dimension of the base differs by one. Moreover, from the viewpoint of Klein geometry, understood as flat Cartan geometry (see, for example, \cite{sharpe2000differential}), an important additional piece of data is which part of the connection is interpreted as the frame (soldering form) and, accordingly, required to be nondegenerate. In the first case, the natural candidate is $\sigma^{*}\tilde\xi^a$, see  \eqref{eq:Q-AdS-in-tilde-coord}, which corresponded to the realization of (A)dS space as the coset  $SO(D-1,2)/SO(D-1,1)$. In the boundary system, interpreting $\sigma^*\hat\xi^A \equiv e^A{}_{\mu}\theta^\mu$ as the soldering form, we arrive at the realization of conformal space as the quotient of the orthogonal group by its parabolic subgroup.

Using the natural gauge transformations encoded in $Q$ and the zero-curvature equations satisfied by $\sigma$ one can set to zero all the components of the connection save for the frame $e^A\equiv\sigma^{*}\hat\xi^A$ which can be set to  
\begin{align}
\label{flat-frame}
    \sigma^{*}\hat\xi^A=\delta^A_\mu\theta^\mu\,.
\end{align}
We refer to this gauge as the coordinate-like one. Note that if $\sigma$ is interpreted as a Cartan connection its natural gauge transformations are not sufficient to achieve~\eqref{flat-frame}. More precisely, to bring the local frame to the above form one also needs to employ the freedom in choosing local coordinates on $\d X$.

Although the coordinate-like gauge simplifies the expressions, it is at the same time very  restrictive and is not useful for keeping track of the background conformal transformations. It is often more convenient to work with a less restrictive gauge (strictly speaking, partial gauge), namely $\sigma^{*}\hat\lambda = 0$, in which all connection components except $e^A{}_{\mu}$ can be expressed in terms of
$\gamma_{\mu\nu} \equiv e^A{}_{\mu} g_{AB} e^B{}_{\nu}$ using the zero-curvature condition. In particular, in this gauge one finds
\begin{align}
    e^A{}_{\mu}\sigma^{*}\lambda_A=-\frac{1}{d-2}(R_{\mu\nu}[\gamma]-\gamma_{\mu\nu}\frac{R[\gamma]}{2(d-1)})\theta^\nu\equiv -P_{\mu\nu}[\gamma]\theta^\nu\,,
\end{align}
where $P_{\mu\nu}[\gamma]$ is the Schouten tensor constructed from the metric $\gamma_{\mu\nu}$. We refer to this gauge as the metric-like gauge, see, for example, \cite{Grigoriev:2025gvk} for details, including the curved case.

\begin{rem}
    Let $\sigma: T[1]\partial X \to \partial^Q F_{\red}$ be a solution in the metric-like gauge, i.e. $\sigma^*(\hat\lambda)=0$. The conserved charge determined by $f^\cA$ from Example \bref{ex-cocycleads},can be  expressed in terms of the metric $\gamma_{\mu\nu}$ and takes the following form:
        \begin{align}
        \int_{T[1]\d X}\sigma^*(f^\cA)=
        \int_{{\d X}} d^4x\sqrt{\gamma}(P_{\mu\nu}P^{\mu\nu}-P^2)\,.
    \end{align}
    This expression formally coincides with the action of conformal gravity up to a topological term. Of course, this should not be interpreted literally as a holographic construction of the conformal gravity action, 
    since the $Q$-map condition on $\sigma$ implies that the metric $\gamma$ is conformally flat. In particular, the above integrand can be set to zero, at least locally, by a gauge transformation.
\end{rem}

\subsection{Fields on compactified (A)dS background}

For the remainder of this section we are interested in Lagrangian theories defined on the background of compactified flat (A)dS space. 

Our starting point is the gPDE describing flat AdS gravity: $ (\mathfrak{g}_\Lambda \times T[1]X, \mathrm{d}_{\algg} + \mathrm{d}_{X})$, where $\algg_\Lambda$ is the AdS algebra.
This system is diffeomorphism-invariant by construction. In the gPDE formulation such systems correspond to gPDEs that are locally trivial as $Q$-bundles. This means that in a suitable trivialisation the total space $Q$-structure is the sum of the base space and the fibre $Q$ structures, as is manifestly so in the case at hand,  where $\mathrm{d}_{\algg} + \mathrm{d}_{X}$. Because the base space $Q$-structure is just the spacetime de Rham differential all the information about the system is encoded in the fiber $Q$-manifold so that one can disregard the base space, at least locally.~\footnote{If gPDE in question is a nontrivial bundle, the additional piece of data on top of the typical fiber and the spacetime manifold is the global structure of the bundle described by e.g. its gluing cocycles. To simplify the exposition, in this work we restrict ourselves to gPDEs which are globally trivial.}

What we are really interested in, is not the flat AdS gravity itself but rather a gauge theory defined on its background. Such a system is naturally described as a gPDE over background, where the base gPDE is precisely the  flat AdS gravity: $(\tilde{E}^\phi, \tilde{Q}) \xrightarrow{} (\mathfrak{g}_\Lambda \times T[1]X, \mathrm{d}_{\algg} + \mathrm{d}_{X})$, where $\phi$ schematically denotes the gauge field of interest. This system is again diffeomorphism-invariant and hence $(\tilde E^\phi,\tilde Q)$ is a product of $(T[1]X,\dx)$ and a $Q$-bundle over $(\mathfrak{g}_\Lambda, \mathrm{d}_{\algg}$). It is clear that all the information about the system is encoded in this bundle and we do not explicitly keep track of the spacetime manifold. We denote this bundle by $(\tilde{F}^\phi, \tilde{Q}) \to (\mathfrak{g}_\Lambda, \mathrm{d}_{\algg})$.

Recall that $(\mathfrak{g}_\Lambda, \mathrm{d}_{\algg} )$ is equivalent to the interior of $(F,Q)$ introduced in Section~\bref{sec:AdS-gravity}. In particular, the projection $\pi_{\mathfrak{g}_\Lambda}: Int(F) \xrightarrow{} \mathfrak{g}_\Lambda$  introduced in Proposition~\bref{prop-Ads-equiv} implements the equivalent reduction. This projection determines the pullback-bundle $\pi_{\mathfrak{g}_\Lambda}^*\tilde{F}^\phi$ over $(Int(F),Q)$, which describes an equivalent reformulation of $(\tilde{F}^\phi, \tilde{Q}) \to (\mathfrak{g}_\Lambda, \mathrm{d}_{\algg} )$ as a gPDE over the interior of the compactified flat AdS gravity.  Then we look for the compactification $(F^\phi,Q)$ of $(\tilde{F}^\phi, \tilde{Q})$, which in turn defines the desired  boundary structure of the initial gauge theory. In the case at hand, compactification means a $Q$-bundle $F^\phi\to F$ whose restriction to $Int(F)$ coincides with $\pi_{\mathfrak{g}_\Lambda}^*\tilde{F}^\phi \to Int(F)$. In practice, this amounts to defining  a near boundary coordinate system on $\tilde F^\phi \to Int(F)$, such that it extends to $E$. In more traditional language this corresponds to choosing the boundary behavior of fields.

In this work our focus is on the extra structure defined on a gPDE, which  encodes its (partially) Lagrangian formulation. More specifically, we are interested in compatible presymplectic structures. 
Suppose $(\tilde{F}^\phi, \tilde{Q}) \xrightarrow{} (\mathfrak{g}_\Lambda, \mathrm{d}_{\algg})$ be equipped with a compatible presymplectic structure:
\begin{equation} \label{eq:523}
    L_{\tilde{Q}} \tilde{\omega} \in \tilde{\cI}_{\mathfrak{g}_\Lambda[1]}\,.
\end{equation}
Then it defines a structure on $Int(F^\phi)$ obeying 
\begin{equation}
    L_Q \tilde \omega \in \cI_{Int(F)}\,,
\end{equation}
where by some abuse of notations we use the same symbol to denote $\pi_{\mathfrak{g}_\Lambda}^* \tilde \omega$. By construction the ideal $\cI_{Int(F)}$ is extendable, therefore, if $\tilde{\omega} \in \bigwedge_{\fop}^2(Int(F^\phi))$ one can apply $\cR, \cA$ maps to obtain cocycles on $\partial^Q F^\phi$. 

In general the starting point gPDE  $(\tilde{F}^\phi, \tilde{Q}) \xrightarrow{} (\mathfrak{g}_\Lambda, \mathrm{d}_{\algg})$ might be implicit, with some ideal $\tilde{\eI}^\phi$ and equation \eqref{eq:523} might only be satisfied up to $\tilde{\eI}^\phi$. The whole construction goes through if the corresponding ideal on $Int(F^\phi)$ is extendable which will be the case in our examples.

Let us collect all the bundles in question into one diagram, each arrow is a $Q$-map and arrows without labels are canonical:
\begin{equation}\nonumber
\begin{tikzcd}
\tilde F^\phi \arrow[d, "\tilde\pi^\phi"] &  \arrow[l,""'](\pi_{\mathfrak{g}_{\Lambda}})^{*}\tilde F^\phi \arrow[d,""'] \arrow[r, hook, "compactification"] & [5em] F^\phi \arrow[d,"\pi^\phi"] & \arrow[l,hook] F_{\red}^\phi \arrow[d] &  \arrow[l, hook] \partial^Q F_{\red}^\phi\arrow[d,""']\\
\mathfrak{g}_\Lambda[1] & \arrow[l,"\pi_\mathfrak{g}"'] Int(F) \arrow[r, hook, "compactification"] & F & \arrow[l, hook, "r"] F_{\red} & \arrow[l, hook, "b"] \partial^Q(F_{\red}) \\
\end{tikzcd}
\end{equation}
Here $F^\phi_{\red}\equiv r^*F^\phi$ and $\partial^QF^\phi_\red=b^{*}F^\phi_\red$.

To conclude, under the assumption that the ideals are extendable and the compactification of $\tilde F^\phi$ is fixed, the initial gPDE $(\tilde F^\phi,\tilde Q)$ equipped with the compatible presymplectic structure induces a boundary gPDE equipped with a pair of compatible presymplectic structures $\omega^\cA$ and $\omega^\cR$. The structure $\omega^\cR$ has degree equal to the boundary dimension. Hence, the associated presymplectic BV-AKSZ action has degree one and is to be interpreted as a BFV-BRST-like charge. The structure $\omega^\cA$, on the other hand, carries precisely the right degree to define a compatible presymplectic structure on the induced boundary gPDE. The associated intrinsic action~\eqref{AKSZ-action} is therefore of degree zero and is interpreted as an action functional on the space of asymptotic boundary data.

\subsection{Example: Scalar} \label{sec:scalar-ads}
In this section, we illustrate the general construction with the example of Klein-Gordon (KG) field on the background of (compactified) flat AdS gravity. Instead of constructing the compactification of this system step by step, we present the result and then comment on possible ways to derive it. Namely, consider an  implicit gPDE
\begin{align}
    (F^\varphi,Q,\eI^\varphi)\to(F,Q)
\end{align}
with fiber coordinates $\nabla_{(a)}\varphi\equiv\{\nabla_{(a_1}\dots\nabla_{a_n)}\varphi,n\geq0\}$.
Here, following \cite{Grigoriev:2025gvk} and by analogy with the example of Section~\bref{sec:example-M}, we use convenient coordinates generated by the vector fields $\nabla_a=\left[Q,\frac{\partial}{\partial \xi^a}\right]$.
Let us note that, for the chosen background $F$, the nilpotency of $Q$ implies, in particular, $[\nabla_a,\nabla_b]=0$.
The action of $Q$ on the fiber coordinates is given by:
\begin{align}\label{Qphi-comp-ads}
Q\varphi=\xi^a\nabla_a\varphi+w\lambda\varphi
\end{align}
and prolonged to other coordinates using  $[\nabla_a,Q]=0$. Note that we are actually dealing 
with the family of systems parameterized by $w\in \fR$. The ideal $\eI^\varphi$ is generated by $\nabla_b$-prolongations of
\begin{multline}\label{Scalar-ideal}
    K~~=~~\Omega^{2}\nabla_{a}\nabla^{a}\varphi~+~ \Omega((1-2w-d)n^a\nabla_a \varphi ~+~(d+1)w\tau \varphi) ~+\\ (w(d+w)n_a n^{a} +m^2)\varphi\,.
\end{multline}
Recall that by $\nabla_a$ prolongations of $K$ we mean $\nabla_{(b_1} \ldots \ldots \nabla_{b_k)}K, \, k>0$. It is easy to check that on $Int(F)$, i.e. where $\Omega>0$, this system is just a reformulation  of the usual gPDE description of  KG field. Indeed, the equivalent reduction to $\Omega=1,  Q\Omega=0, n^a=n_0^a, Qn^a=0$ brings it back.

This system can be obtained in two ways. The first is to view it as the zero-curvature limit of the Klein-Gordon equation on a general asymptotically AdS background, as described in \cite{Grigoriev:2025gvk}. The second approach is to consider the standard implicit gPDE for the Klein--Gordon equation on the background $\algg_\Lambda[1]$, with fiber coordinates $\tilde\nabla_{(a)}\tilde\varphi$, and then lift it to $Int(F)$ using the projection $\pi_{\algg_\Lambda}$ introduced in Proposition \bref{prop-Ads-equiv}, thereby introducing new near-boundary coordinates
\begin{align}\label{scalar-near-b}
    \nabla_{(a)}\varphi \equiv \nabla_{(a)}(\Omega^w \pi_{\algg_\Lambda}^{*}\tilde\varphi)
\end{align}
and rewriting accordingly the ideal generated by $\tilde K\equiv\tilde\nabla_a\tilde\nabla^a\tilde\varphi + m^2 \tilde\varphi$ so that $\eI^\varphi$ is generated by the $\nabla_a$-prolongations of $\Omega^w \pi_{\algg_\Lambda}^{*}\tilde K$.

Let us constrain the parameters $m,w$ via 
\begin{align}
    w=\ell-\frac{d}{2},\quad m^2=-\epsilon w(d+w),\quad \ell\in\mathbb{N}/2,
\end{align}
where the choice of mass is standard in the AdS/CFT context; the restriction of $\ell$ to half-integer  numbers, as we anticipate, will be related to the divergence behaviour of the presymplectic potential. Note that there can be nontrivial boundary systems even if the above relations do not hold but they are not of direct interest in the present context. See however, the discussion in~\cite{Grigoriev:2025gvk}.

 We now apply the reduction described in Proposition \bref{prop-ads-reduction}. We denote by $F^\varphi_{\red}\equiv r^*F^\varphi$, where $r$ denotes the embedding of $F_{\red}$ into $F$, the corresponding pullback bundle. Splitting the index $\{ a\} = \{\Omega, A \}$ and introducing notation $\varphi^{(N)}\equiv (\nabla)^N\varphi$ coordinates  of the fiber can be chosen as follows:
\begin{align}\label{scalar-near-bound}
    \nabla_{(A)}\varphi^{(N)}, \quad \varphi^{(N)}\equiv\nabla^N\varphi,\quad N\geq0\,.
\end{align}
Pulling back to the $Q$-boundary of the base, we obtain
\begin{align}
    (\partial^QF_{\red}^{\varphi},Q, \hat \eI^\varphi)\to(\partial^QF_{\red},Q)
\end{align}
with fiber coordinates $\nabla_{(A)}\hat\varphi$.  Here we use the notation $b^{*}f\equiv \hat f$ for all $ f\in\cC^\infty(F^\varphi_{\red})$ and also $\hat\eI^\varphi\equiv b^{*}\eI^\varphi.$ Let us note that the vector fields $\nabla_A$ on $F^\varphi$ are tangent to both $F^\varphi_{\mathrm{red}}\subset F^\varphi$ and $\partial^Q F^\varphi_{\mathrm{red}}\subset \partial^Q F^\varphi$. Therefore, we denote their restrictions by the same symbol $\nabla_A$.

The ideal $\hat\eI^\varphi$ is generated by $\nabla_A$-prolongations of
\begin{align}\label{phi-gran-eq}
(N-2\ell)\hat\varphi^{(N)}+(N-1)\varepsilon\nabla_A\nabla^A\hat\varphi^{(N-2)}\,,\quad N\geq1\,.
\end{align}
It follows that 
\begin{align}
    \begin{split}
        \hat\varphi^{(i)}\propto\Box\hat\varphi^{(i-2)}+\hat\eI^\varphi\,,\quad i\neq0,1,2\ell\,,
    \end{split}
\end{align}
where $\Box\equiv\nabla_A\nabla^A.$ By iterating this procedure one concludes that all $\hat \varphi^{(i)}$ can be expressed in terms of the remaining ones modulo $ \eI^\varphi$, except for $\hat \varphi^{(0)}$ and $\hat \varphi^{(2\ell)}$, which are interpreted as the leading and subleading values of the scalar field, respectively. Moreover, when $\ell$ is an integer, we have 
\begin{align}
    \Box^\ell \hat\varphi\in\hat \eI^\varphi
\end{align}
which can be interpreted as the flat limit of the GJMS equation; see \cite{Grigoriev:2025gvk} for an analogous analysis in the curved case.

The action of $Q$ on $\partial^QF_{\red}^{\varphi}$ can be easily obtained using Proposition \bref{prop-calculus-ads} and the action of $Q$ in the bulk \eqref{Qphi-comp-ads}:
\begin{align}
    \begin{split}
    &Q\hat\varphi^{(i)}=\hat\xi^A\nabla_A\hat\varphi^{(i)}
+\frac{i(i-1)}{2}\epsilon\hat\lambda^A\nabla_A\hat\varphi^{(i-2)}+(\ell-i-d/2)\hat \lambda\hat \varphi^{(i)}\,.\\
    \end{split}
\end{align}

Now let us turn our attention to presymplectic structures. The presymplectic potential for KG on the background of flat AdS gravity is given by, see e.g.~\cite{Dneprov:2025eoi}
\begin{equation}\label{chi-stand}
\tilde \chi'_\varphi=\tilde\xi\ms{d}_a\tilde\nabla^a\tilde\varphi d\tilde\varphi\,,
\end{equation}
where we assumed that the background flat AdS gravity is described by a minimal gPDE $(\algg_\Lambda[1],\dg)\times (T[1]X,\dx)$ with fiber coordinates $\tilde \xi^a,\tilde \rho^{ab}$. Note that the same expression holds true for generic gravitational background~\cite{Alkalaev:2013hta}. This presymplectic potential can be pulled back to $(Int(F^\varphi_\red),Q)$ by the canonical  projection $\pi^\red_{\algg_\Lambda}:Int(F_\red) \to \algg_\Lambda[1]$ given by~\eqref{pi-red}, making 
$(Int(F^\varphi_\red),Q)$ into an equivalent gPDE with compatible presymplectic structure. Indeed,
representing $(Int(F_\red),Q)$ as a product $Q$-manifold $(\algg_\Lambda[1],\dg)$ and a contractible piece one finds that $(Int(F^\varphi_\red),Q)$ factorizes into the starting point gPDE with presymplectic structure over $(\algg[1],\dg)$ and the contractible piece. The latter is clearly in the kernel of the presymplectic structure and hence does not contribute to the corresponding  BV formulation and the intrinsic action in particular.

Details of the derivation of the explicit form of $\chi'_\varphi$ pulled back to $Int(F^\varphi_\red)$
are relegated to Appendix~\bref{app:scalar-derivation}. Here we immediately present the equivalent (differing by the exact 1-form and terms from the ideal) presymplectic potential and state its properties:
\begin{prop}\label{ads-scalar-presymp} 1-form 
    \begin{align}
        \tilde\chi_\varphi=\Omega^{-d/2-\ell+1}(\xi\ms{d}_\Omega\epsilon \varphi^{(1)}+\xi\ms{d}_A\nabla^A\varphi)d(\Omega^{d/2-\ell}\varphi)
    \end{align}
 on $(Int (F_{\red}^\varphi),Q)$ is a compatible presymplectic potential., i.e. it satisfies 
    \begin{align}
    L_Q\tilde\chi_\varphi=d(-\tilde H_\varphi-i_Q\tilde\chi_\varphi)+\cI_{Int(F)}+\eI^\varphi\,.
    \end{align}
Covariant Hamiltonian $H_\varphi$ is given explicitly by:
\begin{multline}
\tilde H_\varphi=\frac{(-1)^{d+1}}{2}\Omega^{-2\ell}[\xi\ms{d+1}(\Omega(\epsilon(\varphi^{(1)})^2+\nabla_A\varphi\nabla^A\varphi)\,\,+\\
(d-2\ell)\epsilon\varphi \varphi^{(1)})+(d/2-\ell)\Omega \varphi^2\lambda^A\xi_A\ms{d}]\,.
\end{multline}
Moreover, this system is equivalent to the starting point minimal gPDE over $\algg_\Lambda[1]$ with presymplectic structure~\eqref{chi-stand}
in the sense that it can be represented as the product of this gPDE and the contractible gPDE with trivial presymplectic structure. Strictly speaking this applies to $\tilde\chi_\varphi$ modified by the $d$-exact term and the term from the ideal $\cE^\varphi$.
\end{prop}
It is clear that $\chi_\varphi\equiv \Omega^{2\ell}\tilde\chi_\varphi$, $H_\varphi\equiv \Omega^{2\ell}\tilde H_{\varphi}$ extend to the boundary and hence $\chi^\cR$ and $\chi^\cA$ are well-defined and can be explicitly obtained by the procedure of Section~\bref{sec: renorm-struc} with $p=2\ell$.  This explains why we assumed  $\ell \in \mathbb{N}/2$.

Now we calculate the induced structures on the boundary.

\begin{prop}\label{ads-scalar-odd}
    Let $p=2\ell$ be odd. Then
    \begin{align}
        \begin{split}
                        &\chi^\cR_\varphi=\frac{\epsilon}{(2\ell-1)!}\hat\xi\ms{d}\hat\varphi^{(2\ell)} d\hat\varphi+\hat\eI^\varphi\,,
                \qquad\chi^\cA_\varphi\in\hat \eI^\varphi\,,
        \end{split}
    \end{align}
\end{prop}
where $\hat\xi\ms{N}_{A_1\dots A_{d-N}}\equiv b^{*}\xi\ms{N}_{\Omega A_1\dots A_{d-N}}$ and $d=D-1$.
\begin{proof}
On $F^\varphi_{\red}$ we calculate:
\begin{align}\label{formulas-phi}
    \begin{split}
        &\chi_\varphi=\Omega(\xi_\Omega\ms{d}\epsilon\varphi^{(1)}+\xi_A\ms{d}\nabla^A\varphi)d\varphi+(d/2-\ell)\varphi(\xi_\Omega\ms{d}\epsilon\varphi^{(1)}+\xi_A\ms{d}\nabla^A\varphi)d\Omega\,,\\
        &\bar\chi_\varphi=-\Omega\xi\ms{d-1}_{\Omega A}\nabla^A\varphi d\varphi +d\Omega(reg)\,,\\
    \end{split}
\end{align}
To prove the result, note that \eqref{phi-gran-eq} implies that for odd $i<2\ell$:
\begin{align}
    \hat\varphi^{(i)}\in\hat\eI^\varphi\,.
\end{align}
We also note that Proposition \bref{prop-calculus-ads} implies that, for odd $i$,
    \begin{align}
             b^{*}(\nabla)^i\xi_\Omega\ms{d}=b^{*}(\nabla)^i\xi\ms{d-1}_{\Omega A}=b^{*}(\nabla)^{i-1}\xi\ms{d}_A=0\,.
    \end{align}
\end{proof}
Using the coordinate-like gauge on the boundary it is easy to see that the presymplectic structure induced by $d\chi_\varphi^\cR$ on the space of supersections is proportional to
\begin{align}
\int d^dx \delta \varphi^{(0)}(x)\wedge \delta \varphi^{(2l)}(x) \,,
\end{align} 
where by some abuse of conventions $\varphi(x) \equiv \ev^*(\hat \varphi)$, $\varphi^{(r)}(x)\equiv \ev^*(\hat \varphi^{(r)})$. The respective symplectic quotient does not have coordinates of nonvanishing degree and can be identified with the phase-space of the scalar field in $D=d+1$ dimensions. Note that the respective BV-AKSZ action (which is to be identified with the BFV-BRST charge) clearly vanishes by the degree reasoning.

Further examples follow from direct calculations, by applying $L_\nabla$ to \eqref{formulas-phi}.
\begin{example}
Let $\ell=1$. Then, up to $\hat\eI^\varphi$,
\begin{align}
\begin{split}
    &\chi_\varphi^\cR=(\hat\xi\ms{d}\epsilon\hat \varphi^{(2)}+\hat\lambda\hat\xi_{A}\ms{d-1}\nabla^A\hat\varphi)d\hat\varphi
    \,,\qquad\chi_\varphi^\cA=(-1)^{d}\hat\xi_A\ms{d-1}\nabla^A\hat\varphi d\hat\varphi\,,\\
    &H_{l}^\cA=\hat\xi\ms{d}\frac{1}{2}\nabla^A\hat\varphi\nabla_A\hat\varphi+\frac{d-2}{4}\hat\varphi^2\hat\lambda^A\hat\xi\ms{d-1}_{A}-\frac{d-2}{2}\hat \varphi\hat\lambda\hat\xi\ms{d-1}_{A}\nabla^A\hat\varphi\,,
\end{split}
\end{align}
where $H^{\cA}_l$ is the covariant Hamiltonian associated with the presymplectic structure $\omega^\cA=d\chi^\cA$, obtained via \eqref{anomaly-hamiltonian}.

By some abuse of notation we define coordinates on the space of sections $\sigma: T[1]\partial X \xrightarrow{} \partial^Q F^{\varphi}_{\red}$ by 
\begin{equation} 
\begin{gathered}
    \sigma^*\hat\varphi = \varphi\,,    \quad \sigma^*\nabla_A\hat\varphi = \varphi_A \,,\quad
    \sigma^{*}\hat\xi^A\equiv e^{A}{}_{\mu}\theta^\mu\,, \quad
    \sigma^*\hat\lambda_A = \lambda_{A|\mu}\theta^\mu\,, \qquad \sigma^* \hat\lambda = \lambda_{|\mu}\theta^\mu
\end{gathered}
\end{equation}
and denoting $\partial_A\equiv e^{\mu}{}_{A}\partial_\mu$, $\lambda_{A|B} = \lambda_{A|\mu}e_B{}^{\mu}$, $\lambda_{|B} = \lambda_{|\mu}e_B{}^{\mu}$, the intrinsic action associated with $\chi_\varphi^\cA$ takes the form:
    \begin{align}
        S[\varphi,\varphi^A]=-\int e\ms{d}(\varphi^A\partial_A\varphi-\frac{1}{2}\varphi^A\varphi_A+\frac{2-d}{4}\varphi^2\lambda^A{}_{|A}+\frac{d-2}{2}\varphi\lambda_{|A}\varphi^A)\,.
    \end{align}
    Picking a background configuration subject to the metric-like gauge
    introduced in Section \bref{Ads-sections}, $\gamma_{\mu\nu}\equiv e^{A}{}_{\mu}g_{AB}e^{B}{}_{\nu},$ and eliminating the auxiliary field $\varphi^{A}$, the action can be rewritten as
\begin{align} \label{conf-scalar-action}
    S[\varphi]=-\frac{1}{2}\int d^dx\sqrt \gamma(\partial_\mu\varphi\partial^\mu\varphi+\frac{d-2}{4(d-1)}\varphi^2R[\gamma])\,,
\end{align}
which is the standard form of the conformally coupled scalar action. By construction this action is invariant under the background gauge transformations preserving the metric-like gauge. Note that we restricted ourselves to flat-backgrounds and hence the invariance is guaranteed only for them. However, in this particular case the above action is conformally invariant even for curved $\gamma$ but this is not by construction. In fact, in this case it is straightforward to replace flat AdS gravity background with the full gravity and obtain the boundary system as in~\cite{Grigoriev:2025gvk}, giving a conformally-coupled scalar on general background.   
\end{example}

\begin{example}

Let $\ell=2$. Then, up to $\hat\eI^\varphi,$ 
    \begin{align}
\chi_\varphi^\cA=\epsilon\frac{(-1)^d}{4}(\hat \xi_A\ms{d-1}\nabla^A\hat \varphi d\Box\hat \varphi+\hat \xi_A\ms{d-1}\nabla^A \Box \hat \varphi d\hat \varphi+2\hat \lambda^B\hat\xi\ms{d-2}_{AB}\nabla^A\hat \varphi d\hat \varphi)\,,
   \end{align}
   and
\begin{multline}
            H^{\cA}_l=-\frac{\epsilon}{4}(\hat\xi\ms{d}(\frac{1}{2}(\Box\hat \varphi)^2-\nabla_A\Box\hat \varphi\nabla^A\hat \varphi)-\hat\lambda^B\hat\xi_B\ms{d-1}\nabla_A\varphi\nabla^A\varphi+ 
            \\
            \frac{4-d}{2}\hat\lambda^A\hat\xi_A\ms{d-1}\hat\varphi\Box\hat \varphi+\frac{4-d}{2}\hat\lambda^A\hat \lambda^B\hat\xi\ms{d-2}_{AB}\varphi^2+\frac{d}{2}\hat\lambda\hat\xi\ms{d-1}_A\nabla^A\varphi\Box\varphi+
            \\
            (4-d)\hat \lambda^B\hat \lambda\hat \xi_{AB}\ms{d-2}\nabla^A\hat \varphi\hat\varphi+\frac{d-4}{2}\hat \lambda\hat \xi_A\ms{d-1}\varphi\nabla^A\Box\varphi)\,.
\end{multline}
The corresponding intrinsic action can be immediately read-off from the above data in terms of the conformal Cartan geometry of the background. For brevity, we illustrate its form in the metric-like gauge of the background:
    \begin{multline}
S[\varphi,\varphi_\alpha,\Box\varphi,\Box_\alpha\varphi]=-\frac{\epsilon}{4}\int d^dx\sqrt{\gamma}((\Box^\mu\varphi+2P^{\sigma}{}_{\nu}\delta^{\nu\mu}_{\rho\sigma}\varphi^\rho)\partial_\mu\varphi+\varphi^\mu\partial_\mu\Box\varphi+\\\frac{1}{2}(\Box\varphi)^2-\Box_\mu\varphi\varphi^\mu+P\varphi_\mu\varphi^\mu+\frac{d-4}{2}P\varphi\Box\varphi-\frac{d-4}{2}P^\mu{}_{\sigma}P^\nu{}_{\rho}\delta^{\sigma\rho}_{\mu\nu}\varphi^2)\,,
    \end{multline}
where $P^{\sigma}{}_{\nu}$ is the Schouten tensor of $\gamma$. All fields except for $\varphi$ are auxiliary, and eliminating them yields, up to a boundary term:
\begin{multline}
    S[\varphi]\propto\int d^{d}x\sqrt{\gamma}((\nabla_\alpha\nabla^\alpha\varphi)^2-(4P^{\mu\nu}-(d-2)Pg^{\mu\nu})\partial_\mu\varphi\partial_\nu\varphi+\\\frac{d-4}{2}(-\nabla_\mu\nabla^\mu P+\frac{d}{2}P^2-2P^{\mu\nu}P_{\mu\nu})\varphi^2)\,.
\end{multline}
Just like in the previous example this action happens to be conformally invariant on generic background and its equation of motion is the Paneitz equation~\cite{Paneitz:1983}. For $d=4$ this action was proposed even earlier in~\cite{Fradkin:1985am}. 
\end{example}

\begin{example}\label{ads-scalar-even}
Finally, let us consider the general case with $\ell\in\mathbb{N}$ and to simplify the considerations restrict ourselves to the coordinate-like gauge where $\sigma^*(\xi^A)= \theta^A$ and all other background fields are set to zero. On this background
    \begin{align}
    \begin{split}
                \chi_\varphi^\cR&=\frac{\epsilon}{(2\ell-1)!}\theta\ms{d}\sum_{i=0}^{\ell-1}C^{2i}_{2\ell-1}\hat\varphi^{(2\ell-2i)}d\hat \varphi^{(2i)}\,,\\
        \chi_\varphi^\cA&=\frac{(-1)^d}{(2\ell-2)!}\theta\ms{d-1}_{A}\sum_{i=0}^{\ell-1}C^{2i}_{2\ell-2}\nabla^{A}\hat \varphi^{(2\ell-2-2i)}d\hat \varphi^{(2i)}\,,\\
        H^{\cA}_l&=-\frac{1}{2(2\ell-2)!}\theta\ms{d}\sum_{i=0}^{\ell-1}(\epsilon C^{2i-1}_{2\ell-2}\hat\varphi^{(2\ell-2i)}\hat\varphi^{(2i)}- C^{2i}_{2\ell-2}\nabla_A\hat\varphi^{(2\ell-2-2i)}\nabla^A\hat\varphi^{(2i)})\,,
    \end{split}
\end{align}
up to terms from $\hat\eI^\varphi$. The intrinsic action associated with $\chi^\cA_\varphi$ takes the form 
\begin{multline} \label{1st-order-scalar}
    S[\varphi^{(2i)},\varphi^{(2i)|A}]\propto \sum_{i=0}^{\ell-1}\int d^{d}x[C^{2i}_{2\ell-2}\varphi^{(2\ell-2-2i)|A}\partial_A \varphi^{(2i)}+\\\frac{1}{2}(\epsilon C^{2i-1}_{2\ell-2}\varphi^{(2\ell-2i)}\varphi^{(2i)}- C^{2i}_{2\ell-2}\varphi_A^{(2\ell-2-2i)}\varphi^{(2i)|A})]\,.
\end{multline}
One can easily see that all the fields save for $\varphi^{(0)}$ are auxiliary and their elimination gives
\begin{align}
   S[\varphi^{(0)}] \propto \int d^d x \varphi^{(0)} (\partial_A\partial^A)^\ell\varphi^{(0)}\,.
\end{align}
In other words,  \eqref{1st-order-scalar} is the first-order action of the higher-order singleton field. Although this action can easily be obtained by e.g. applying the intrinsic action construction of~\cite{Grigoriev:2016wmk},
the present derivation is based on holographic considerations and guarantees the conformal invariance. 
 It can also be generalised to a non-flat case in a systematic way.
\end{example}

\subsection{Example: Maxwell} \label{sec:maxwell-ads}
Consider an implicit gPDE over background
\begin{align}
    (F^\maxw,Q,\eI^\maxw)\to(F,Q)
\end{align}
whose fiber is coordinatized by the following overcomplete coordinates
\begin{align}
    C,\quad \nabla_{(a)}F_{bc}, \qquad gh(F_{bc})=0\,,\quad \,gh(C)=1\,,
\end{align}
subject to the $\nabla_a$ prolongations of the relations
\begin{align}\label{Max-Bianchi}
    \nabla_{{[a}}F_{bc]}=0,\qquad F_{ab}=-F_{ba}\,.
\end{align}
The action of $Q$ is determined by 
\begin{align}\label{Q-ads-Max}
            &QC = \half \xi^a \xi^b F_{ab}\,,\qquad
    Q F_{ab} = \xi^c \nabla_c F_{ab} +  \rho_a{}^cF_{cb} + \rho_b{}^cF_{ac}- 2 \lambda F_{ab}
\end{align}
and the ideal $\eI^\maxw$ is generated by the $\nabla_a$-prolongations of
\begin{align}\label{max-eq}
    Y_a=(4-D)n^bF_{ba}+\Omega\nabla^bF_{ba}\,.
\end{align}

Just like in the case of  scalar field, this system can be obtained in two ways. The first is to take the flat-background and abelian limit of the conformal-like Yang--Mills system described in~\cite{Grigoriev:2025gvk}. The second is to start from the usual minimal gPDE for the Maxwell field on $\algg_\Lambda[1]$-background, see Section~\bref{sec:example-M},
and pull it back to $Int(F)$ using the projection $\pi_{\algg_\Lambda}$. A subtlety is that if $\tilde C, \nabla_{(c)}\tilde F_{ab}$ denote the usual overcomplete coordinates on the fiber of the minimal gPDE for Maxwell, the action of $Q$ in these coordinates pulled back by $\pi^*_{\algg_\Lambda}$ becomes singular at $\Omega=0$. This can be resolved by introducing new coordinates on $Int(F^\cM)$:
\begin{align}
\label{ads-maxwell-uplift}
\nabla_{(a)}F_{bc}\equiv \nabla_{(a)}(\Omega^{-2}\pi_{\algg_\Lambda}^*\tilde F_{bc})\,,
\end{align}
and these are precisely the coordinates in which the action of $Q$ is given by \eqref{Q-ads-Max}. Note that, in contrast to the scalar field case, the conformal weight of the Maxwell field is fixed by the requirement of regularity of $Q$, which, in turn, reflects the compatibility between the Maxwell gauge transformations and the asymptotic behavior of fields.

Now we can use the background reduction described in Section \bref{sec:AdS-gravity} and define $F^\maxw_{\red}\equiv r^{*}F^\maxw$ which is by construction a gPDE over background $F_\red$. As usual, we also split the index as $\{a\}=\{\Omega,A\}$. Then  a part of the relations \eqref{Max-Bianchi} can be solved explicitly in terms of $J_{B}^{(N)}\equiv\nabla_{(A)}\nabla^N F_{\Omega B}$:
\begin{align}\label{ads-solved-bianchi}
    \nabla^N F_{BC}=\nabla_B J^{(N)}_C-\nabla_{C} J^{(N)}_B\,,
\end{align}
giving the following overcomplete coordinate system on the fibers of $F^\maxw_{\red}$:
\begin{align}\label{max-coord}
    C\,,\qquad \nabla_{(A)}F_{BC}\,,\qquad \nabla_{(A)}J_B^{(N)}\,.
\end{align}
The boundary gPDE is then given by a pullback of $F^\maxw_{\red}$ to $\partial^QF_{\red}$ and can be represented as:
\begin{align}
    (\partial^QF^\maxw_{\red},Q,\hat \eI^\maxw)\to(\partial^QF_{\red},Q)\,.
\end{align}
We denote by $b$ the embedding $\partial^QF^\maxw_{\red} \hookrightarrow F^\maxw_{\red}$.
Using as fiber coordinates the restriction of \eqref{max-coord} to $\partial^QF^\maxw_{\red}$, the action of $Q$ can be explicitly computed from \eqref{Q-ads-Max} using Proposition \bref{prop-calculus-ads}.

For further computations, it is useful to formulate the following analogue of Proposition \bref{prop-calculus-ads}:
\begin{prop}\label{calculus-ads-maxwell}
    \begin{align}
        \begin{split}
            b^{*}\nabla^NF_{AB}&=\nabla_A\hat J_B^{(N-1)}-\nabla_B\hat J_{A}^{(N-1)}\,,\\
            b^{*}\nabla^NC&=\hat \xi^B\hat J_B^{(N-1)}+\epsilon\frac{(N-1)(N-2)}{2}\hat \lambda^B\hat J_B^{(N-3)}\,,
        \end{split}
    \end{align}
    where $N\geq1$.
\end{prop}
\begin{proof}
    The first line follows directly from \eqref{ads-solved-bianchi}; the second follows from
    \begin{align}
        QC=\frac{1}{2}\xi^a\xi^bF_{ab}=(\xi-\lambda\Omega)\xi^BJ_B+\frac{1}{2}\xi^A\xi^BF_{AB}\
    \end{align}
    which holds on $F^\cM_\red$, together with Proposition \bref{prop-calculus-ads}.
\end{proof}

Acting with $b^{*}\nabla^N$ on the generators of the ideal $\eI^{\maxw}_{\red}$, we obtain that $\hat\eI^\maxw$ is  generated by $\nabla_A$-prolongations of 
\begin{align}
    \begin{split}
        &\hat Y_\Omega^{(N)}=-N\nabla^A\hat J^{(N-1)}_A,\qquad N\geq1\,,\\
        
        &\hat Y_{A}^{(N)}=(4-D+N)\epsilon \hat J_A^{(N)}+N\nabla^B b^{*}\nabla^{N-1} F_{BA}\,,\qquad N\geq0\,,
    \end{split}
\end{align}
where $ b^{*}\nabla^{N-1} F_{BA}$ can be calculated using Proposition \bref{calculus-ads-maxwell}.

 Note that, since $\nabla^A\nabla^B F_{AB}=0$, we have
 \begin{align}\label{maxwell-noether}
(4-D+N)\hat Y_\Omega^{(N)}\propto \nabla^B\hat Y_B^{(N-1)}     
 \end{align} 
 and hence the generators $\hat Y_\Omega^{(N)}$ for $N\neq D-4$ can be excluded from the generating set. This can be interpreted as a consequence of the Noether identities for the Maxwell equations, see \cite{Grigoriev:2025gvk}, where a nonlinear version of these equations is also given.

It is easy to see that, among all $\hat J_B^{(N)}$, the only independent coordinates on the zero locus of $\hat \eI^\maxw$ are $\hat J_B^{(D-4)}$, which corresponds to the subleading value of the Maxwell field. Moreover, in the case where $D-1\geq4$ and is even, we have
\begin{align}
    \Box^{\frac{D-5}{2}}\nabla^A\hat F_{AB}\in\hat\eI^\maxw\,.
\end{align}
This is interpreted as the obstruction equation for the leading value.

Now let us turn our attention to the structures induced on the boundary gPDE by the compatible  presymplectic potential in the bulk. First of all, we define a suitable version of the Maxwell presymplectic potential in terms of $Int(F^\maxw_{\red})$:
\begin{prop}\label{ads-max-presymp} 1 form 
    \begin{align}
    \label{chi-Fred-Maxw}
\tilde \chi_\maxw=\Omega^{4-D}(2\epsilon\xi\ms{D-2}_{\Omega B}J^B+\xi\ms{D-2}_{AB}F^{AB})dC
    \end{align}
on $(Int(F^\maxw_{\red}),Q,F_\red)$ is a compatible presymplectic potential. More precisely,
\begin{align}
\begin{aligned}
L_Q\tilde\chi_\maxw
    &= d(-\tilde H_{\maxw}-i_Q\tilde\chi_{\maxw})+\cI_{Int(M)}+\eI_{\red}^\maxw\,,\\
\tilde H_\maxw
    &= -\frac{1}{2}\Omega^{4-D}\xi\ms{D}
       (2\epsilon J_BJ^B+F^{AB}F_{AB})\,.
\end{aligned}
\end{align}
Moreover, this system is equivalent to the gPDE  with presymplectic structure described in Section~\bref{sec:example-M}
in the sense that it can be represented as the product of the gPDE with presymplectic structure from Section~\bref{sec:example-M} and the contractible gPDE with trivial presymplectic structure.
\end{prop}
\begin{proof}
The proof is analogous to that of Proposition \bref{ads-scalar-presymp}. Namely, the presymplectic structure of the gPDE over $\algg_\Lambda[1]$ is given by \eqref{maxwell-presymp}
\begin{align}\label{ads-maxwell-initial}
\tilde\xi\ms{D-2}_{ab}\tilde F^{ab}d C\,.
\end{align}
It can be pulled back to $Int(F_\red^\maxw)$ by the $Q$ projection $\pi^\red_{\algg_\Lambda}: Int(F_\red) \to \algg_\Lambda[1]$ given by \eqref{pi-red}, giving the 
\eqref{chi-Fred-Maxw} in the coordinate system~\eqref{ads-maxwell-uplift}.
\end{proof}
We see that $\tilde\chi_\maxw$ satisfies the assumptions of Theorem ~\bref{th-renorm} for $p=D-4$, and we define
\begin{align}
    \chi_\maxw\equiv \Omega^{D-4}\tilde\chi_\maxw\,,\qquad H_\maxw\equiv\Omega^{D-4}\tilde H_{\maxw}\,.
\end{align}

\begin{prop} Let the dimension of the boundary  $d \equiv D-1$ be odd. Then
    \begin{align}\label{ads-maxwell-odd}
        \begin{split}
        &\chi_{\maxw}^{\cR}=\frac{2\epsilon}{(d-3)!}\hat\xi_B\ms{d-1}\hat J^{(d-3)|B}dC+\hat\cI^\maxw\,,\qquad \chi^\cA_\maxw\in\hat\cI^\maxw\,.
        \end{split}
    \end{align}
\end{prop}
\begin{proof}
    Similarly to Proposition~\bref{ads-scalar-odd}, one observes that 
    \begin{align}
            \hat J^{(2i)}_B\in \hat \eI^\maxw,\qquad 0\leq i< \frac{d-3}{2}
    \end{align}
    and then uses Proposition~\bref{calculus-ads-maxwell}.
\end{proof}
It is instructive to consider the presymplectic structure induced by $\chi^\cR_\maxw$ on the space of supersections. Introducing components according to $\ev^*(\hat C)=\cC(x)+\theta^B A_B(x)+\ldots$, and $\ev^*(\hat J^{(d-3)|A})=J_0^A(x)+\theta^B \cP^A_B(x)+\ldots$, where we have assumed the background configuration in the coordinate-like gauge, the presymplectic structure induced on the space of supersections is given by (up to a coefficient)
\begin{equation}
    \int d^dx \, \left(\delta J^B_0 \wedge \delta A_B+ \delta \cP \wedge \delta \cC  \right)\,, \qquad \cP\equiv \cP^A_A\,.
\end{equation}
This symplectic structure has degree $0$ and is obviously regular. It determines a symplectic quotient with coordinates $A_B(x), J^B_0,\cC(x),\cP(x)$, which is easily seen to be identical with the BFV phase space of the Maxwell system. 
Moreover, the presymplectic BV-AKSZ action determined by $Q$ and $\chi^\cR$ is just a usual BFV-BRST charge $\int d^dx J^A \d_A \cC$ (up to an overall factor and boundary term). This is not a coincident. Indeed, had we started with Maxwell theory on $\d X \times \fR^1_{\geq 0}$ and formulated it as a presymplectic gauge PDE, we would have obtained the above presymplectic gPDE formulation of the BFV system for Maxwell as its pullback to the boundary, see \cite{Grigoriev:2022zlq} for more details. This is of course just a presymplectic gPDE derivations of the BFV Hamiltonian formulation of the system. However, in the present setup the boundary is asymptotic and the above derivation involved renormalization of the presymplectic structure. Apparently, it was not guaranteed that the result is isomorphic to the BFV phase space.  
\begin{example}
    For $d=4$
    we have
    \begin{align}
        \begin{split}
            &\chi^\cR=(\epsilon\hat\xi\ms{3}_{B}\hat J^{(1)|B}-\frac{1}{2}\hat \lambda\hat\xi_{AB}\ms{2}\hat F^{AB})d\hat C\,,\qquad \chi^\cA=-\hat\xi\ms{2}_{AB}\hat F^{AB}d\hat C\,,
        \end{split}
    \end{align}
modulo terms in $\hat \eI^\maxw$.
  Note that the above presymplectic potential $\chi^\cA$ coincides with \eqref{ads-maxwell-initial} for $d=4$. We also see that $\chi^\cR$, in contrast to the odd-dimensional boundary case \eqref{ads-maxwell-odd}, acquires an additional correction necessary for its $Q$-invariance. 

  The presymplectic BFV-AKSZ sigma model determined by $\chi^\cR,Q$ again gives the BFV formulation for Maxwell but in contrast to the odd-$d$ case it is somewhat modified to maintain the coupling to background conformal geometry.  Nevertheless, in the metric-like and coordinate-like background gauges the expressions for the presymplectic structure and the BRST charge are formally the same. At the same time, it is clear from the explicit form of $Q$ and $\chi^\cA$ on $\d^Q F_\red^\maxw$ that the associated presymplectic BV-AKSZ system describes Maxwell in 4d. The bonus of the present formulation is that the conformal invariance is manifestly captured by the background gPDE.

\end{example}
\begin{example}
    For  $d=6$ 
    we have
    \begin{align}
    \chi^\cA_\maxw=-\frac{\epsilon}{2}[(\hat \xi_{AB}\ms{4}\nabla^A \nabla_C\hat F^{CB}+\hat \lambda^C\hat \xi_{ABC}\ms{3}\hat F^{AB})d\hat C+\frac{1}{2}\hat \xi_{AB}\ms{4}\hat F^{AB}d(\hat \xi^C\nabla^D\hat F_{DC})]\,,
\end{align} 
and
\begin{multline}
            H^{\cA}_l=-\frac{\epsilon}{2}[\hat \xi\ms{6}(\frac{1}{2}\nabla_A\hat F^{AB}\nabla^C\hat F_{CB}-\nabla^A\nabla_C\hat F^{CB}\hat F_{AB})-
            \\
            \frac{1}{2}\hat \lambda^C\hat \xi_{C}\ms{5}\hat F^{AB}\hat F_{AB}+2\hat \lambda\hat \xi_A\ms{5}\hat F^{AC}\nabla^D\hat F_{DC}]\,,
\end{multline}         
modulo terms in $\hat \eI^\maxw$.
    
Introducing coordinates on the space of sections by 
\begin{align}
\begin{alignedat}{4}
\sigma^*\hat{C}
    &= A_\mu\theta^\mu\,,
&\qquad
F_{\mu\nu}
    &= e^{A}{}_{\mu}e^B{}_{\nu}\sigma^*\hat F_{AB}\,,
\\
j_\mu
    &= e^{A}{}_{\mu}\sigma^{*}\nabla^C\hat F_{CA}\,,
&\qquad
k_{\mu\nu}
    &= e^A{}_{\mu}e^{B}{}_{\nu}\sigma^{*}\nabla_A\nabla^C\hat F_{CB}\,.
\end{alignedat}
\end{align}

The intrinsic action in the metric-like gauge of the background is given by
\begin{multline}
    S[A,F,j,k]=-\frac{\epsilon}{2}\int d^6x\sqrt{\gamma}(2k^{\mu\nu}\partial_{[\mu} A_{\nu]}-\delta^{\delta \sigma\rho}_{\alpha\beta \gamma}P^{\gamma}{}_{\delta}F^{\alpha\beta}\partial_\sigma A_\rho+F^{\mu\nu}\partial_\mu j_\nu+\\+\frac{1}{2}j^\mu j_\mu-k^{\mu\nu}F_{\mu\nu}-\frac{1}{2}PF^{\mu\nu}F_{\mu\nu})\,.
\end{multline}
After elimination of auxiliary fields
\begin{align}
    \begin{split}
        &F_{\mu\nu}=2\partial_{[\mu}A_{\nu]}\,,\quad j^\nu=\frac{1}{\sqrt \gamma}\partial_\mu(\sqrt\gamma F^{\mu\nu})\,,\quad k_{[\mu\nu]}=\partial_{[\mu}j_{\nu]}-\delta^{\delta\sigma\rho}_{\mu\nu\gamma}P^\gamma{}_\delta\partial_\sigma A_\rho-P F_{\mu\nu}
    \end{split}
\end{align}
we get
\begin{align}\label{6d-maxwell}
    S[A]=-\frac{\epsilon}{4}\int d^6x \sqrt\gamma(j_\mu j^\mu-PF^{\mu\nu}F_{\mu\nu}-4P^{\gamma}{}_{\alpha}F^{\alpha\beta}F_{\beta\gamma}+2F^{\mu\nu}\partial_{\mu}j_{\nu})\,.
\end{align}
 As in the scalar case, this functional turns out, somewhat accidentally, to be well defined also on a general conformal-geometric background, and not only on a flat one. The extrema of this action are described by a gauge field $A_\mu$ satisfying the conformally invariant generalization of the Maxwell equation in six dimensions, of the form $\Box\nabla^\mu F_{\mu\nu}+\dots=0$. The corresponding conformally invariant differential operator  was already constructed by Branson in~\cite{branson1985differential}, where it was denoted by $D_{4,1}$. We also note that the action \eqref{6d-maxwell}, up to a boundary term, coincides with the abelian limit of the functional referred to in~\cite{Gover:2023rch} as the $d=6$ Yang-Mills energy and obtained there via holographic renormalization.

 \end{example}

\begin{rem} Just like in 
the scalar field Example~\bref{ads-scalar-even}, it is straightforward to compute the general form of $\chi^\cA$ in the coordinate-like gauge for an arbitrary even-dimensional boundary. The corresponding AKSZ action is a first-order action for the equation 
\begin{align}
    (\partial_\rho\partial^\rho)^{\frac{d-4}{2}}\partial^\mu F_{\mu\nu}[A]=0\,,
\end{align} but we omit the details here. Of course, Lagrangian for generic conformal gauge fields on conformally-flat spaces are known in the literature~\cite{Vasiliev:2009ck}, see also~\cite{Chekmenev:2020lkb} where such Lagrangians are also derived from bulk considerations. 
\end{rem}

\section{Applications and examples: fields on Minkowski space}
\label{sec:flat}

\subsection{Flat gravity and its compactification}
As in the AdS case, we consider a graded $Q$-manifold with coordinates
\begin{align}
\label{coord-lin}
    \Omega,\,n_a,\,\tau\,\quad \Omega\geq0,\, n_\Omega>0,\qquad \xi^a,\,\rho^{ab},\,\lambda,\,\lambda^a\,,
\end{align}
where the index $a$ is decomposed as $\{a\}=\{\Omega,u,A\}$, $|A|=D-2$, and the action of $Q$ is given by \eqref{Q-int-F}. The $Q$-manifold $(F,Q)$ is then defined as the zero locus of
\begin{align}\label{poincare-scale}
    \Omega\tau+\frac{1}{2}n_an^a
\end{align}
and indices are raised and lowered by the following  constant metric\footnote{Note that the constraint \eqref{poincare-scale} coincides with \eqref{scale-ideal} for $\epsilon=0$, i.e. in the case of vanishing cosmological constant, while the non-diagonal form of the constant metric is chosen so that this constraint is well-defined at $\Omega=0,\,n_\Omega>0.$ }
\begin{equation} \label{flat-metric}
g_{ab}=
\begin{pmatrix}
0 & 1 & 0  & \cdots & 0\\
1 & 0 & 0 & \cdots & 0\\
0 & 0 & -1  & \cdots & 0\\
\vdots & \vdots &  \vdots & \ddots & \vdots\\
0 & 0 & 0 &  \cdots & -1
\end{pmatrix}\,.
\end{equation}
$(F,Q)$ is a manifold with $Q$ boundary which can be identified with the submanifold  singled out by $\Omega=0, \,Q\Omega=0$.

\begin{prop}\label{poincare-equiv}
    $(Int(F),Q)$ is equivalent, as a $Q$-manifold, to $(\mathfrak{iso}[1],\mathrm{d}_{\algg})$,  where $\mathfrak{iso}$ is the Poincaré algebra in $D$ dimensions and $\mathrm{d}_{\algg}$ is its Chevalley-Eilenberg differential. 
\end{prop}
Notice that this proposition implies that  $(F,Q)$ is a compactification of $(\mathfrak{iso}[1],d_{\algg})$.
\begin{proof}
This proposition is a direct analogue of Proposition~\bref{prop-Ads-equiv}. It is enough to check that for $\Omega>0$ the constraint \eqref{poincare-scale} can be solved for $\tau$ so that the remaining coordinates can be taken as independent coordinates on its zero locus. Then  $\Omega-1$, $Q\Omega$ as well as $n_a-g_{a\Omega}$ and $Qn_a$ are contractible pairs provided $\Omega>0$ and the reduced $Q$-manifold is diffeomorphic to $(\mathfrak{iso}[1],\dg)$.
\end{proof}
It is also useful to have the explicit  form of the $Q$-projection  map $\pi_\mathfrak{iso}:(Int(F),Q)\to (\mathfrak{iso}[1],\mathrm{d}_{\algg})$. If 
$\tilde\xi^a$, $\tilde \rho_{a}{}^{b}$ are linear coordinates on $\mathfrak{iso}[1]$ such that the differential reads as
\begin{align}
\label{pi-F-iso}
    \mathrm{d}_\algg \tilde\xi^a=\tilde\xi^b\tilde\rho_{b}{}^{a},\qquad \mathrm{d}_\algg \tilde\rho_{a}{}^{b}=\tilde{\rho}_{a}{}^{c}\tilde\rho_{c}{}^{b}\,,
\end{align}
the projection  $\pi_\mathfrak{iso}:(Int(F),Q)\to (\mathfrak{iso}[1],\mathrm{d}_{\algg})$ is given by
    \begin{align}
        \pi_\mathfrak{\mathfrak{iso}}^*\tilde\xi^a=\Omega^{-1}\xi^a\,,\qquad \pi_\mathfrak{\mathfrak{iso}}^{*}\tilde{\rho}_{a}{}^{b}=\rho_{a}{}^{b}-\Omega^{-1}n_a\xi^b+\Omega^{-1}n^b\xi_a\,.
    \end{align}

As in the AdS case, it is more convenient to work with a $Q$-submanifold $F_{\red}\subset F$ 
which can be defined as a surface in the coordinate space \eqref{coord-lin}, singled out by 
\begin{align}
    n_\Omega-1=n_\alpha=0\,, \quad \rho^{\Omega \alpha}=\lambda^\alpha\Omega\,,\qquad \tau=0,\quad\lambda^\Omega=0\,,
\end{align}
where $\{\alpha\}=\{u,A\}$ is a shorthand notation for the corresponding index set. It is clear that this surface  belongs to the zero locus of \eqref{poincare-scale} so that indeed $F_{\red}\subset F$.
\begin{prop}
    $(F_\red,Q)$ is an equivalent reduction of $(F,Q)$.
\end{prop}
\begin{proof}
Indeed, since $n_\Omega>0$, the functions $n_A$, $Qn_A$, $n_\Omega$, $Qn_\Omega$, $\tau$ and $Q\tau$ form contractible pairs, while $n_u$ can be eliminated by solving the constraint \eqref{poincare-scale}. Introducing $\xi\equiv\xi^\Omega+\lambda\Omega$, one may take the restriction of
\begin{align}
\xi\,,\xi^\alpha\,,\rho^{\alpha\beta}\,,\lambda\,,\lambda^\alpha\,,\Omega\,.
\end{align}
to $F_{\red}$ as independent coordinates. 
In this coordinate system the action of $Q$ is given by:
\begin{align}\label{poincare-Q-Fred}
    \begin{split}
    &Q\Omega=\xi\,,\qquad Q\xi=0\,,
    \\  &Q\rho^{\alpha\beta}=\rho^{\alpha C}\rho_{C}{}^{\beta}+ \delta^{\alpha\beta}_{\gamma\delta}\lambda^\gamma\xi^\delta-\Omega\delta^{\alpha\beta}_{\gamma\delta}\lambda^\gamma\rho^{u\delta}\,,\\
    &Q\xi^\alpha=\xi^B\rho_{B}{}^{\alpha}-\xi^\alpha\lambda+(\xi-\lambda\Omega)\rho^{u\alpha}+\xi^u\lambda^\alpha\Omega\,,\\
    &Q\lambda=\xi^A\lambda_A+ (\xi-\lambda\Omega)\lambda^u\,,\\
    &Q\lambda^\alpha=\lambda^B\rho_{B}{}^{\alpha}-\lambda\lambda^\alpha+\Omega\lambda^u\lambda^\alpha\,.
    \end{split}
\end{align}
By combining \eqref{pi-F-iso} with the embedding $F_\red \subset F$ one finds the explicit form of the $Q$-projection  $\pi^\red_{\mathfrak{iso}}:(F_\red,Q)\to(\mathfrak{iso}[1],\mathrm{d}_\algg)$
\begin{align}\label{pi-red-null}
\begin{aligned}
(\pi^{\red}_{\mathfrak{iso}})^{*}\tilde \xi^\Omega
    &= \Omega^{-1}\xi-\lambda\,,&
(\pi^{\red}_{\mathfrak{iso}})^{*}\tilde\xi^\alpha&= \Omega^{-1}\xi^\alpha\,,\\
(\pi^{\red}_{\mathfrak{iso}})^{*}\tilde\rho^{AB}
    &= \rho^{AB}\,,&(\pi^{\red}_{\mathfrak{iso}})^{*}\tilde\rho^{uB}&= \rho^{uB}-\Omega^{-1}\xi^B\,,\\
(\pi^{\red}_{\mathfrak{iso}})^{*}\tilde\rho^{\Omega u}  &= \lambda^u\Omega+\Omega^{-1}(\xi-\lambda\Omega)\,,&
(\pi^{\red}_{\mathfrak{iso}})^{*}\tilde\rho^{\Omega B} &= \lambda^B\Omega\, ,
\end{aligned}
\end{align}
which we need in what follows.
\end{proof}
\begin{prop}
    $(\partial^QF_\red,Q)$ is equivalent to $(\mathfrak{iso}[1],\mathrm{d}_\algg)$.
\end{prop}
\begin{proof}
    The $Q$-boundary  $\partial^QF_{\red}$ of $F_{\red}$ is defined by $\Omega=\xi=0$ and we denote by $b:\partial^QF_{\red} \hookrightarrow F_{\red}$ the corresponding embedding. It is straightforward to verify that $(\partial^QF_{\red},Q) \subset (\partial^Q F,Q)$ is an equivalent reduction. 
    
    Using the notation $\hat f=b^{*}f$, $\forall f\in C^\infty(F_{\red})$, we take the following functions
    \begin{align}
        \hat \xi^\alpha\,, \hat \rho^{\,\alpha\beta}\,, \hat \lambda\,, \hat \lambda^\alpha\,,
    \end{align}
     as coordinates on $\partial^QF_\red$.  The action of $Q$ on these coordinates is then obtained by applying $b^*$ to \eqref{poincare-Q-Fred}, giving:
     \begin{align}\label{poincare-Qbdry}
\begin{alignedat}{4}
Q\hat\xi^\alpha
    &= \hat \xi^B\hat \rho_{B}{}^{\alpha}-\hat \xi^\alpha\hat\lambda\,,
&\qquad\qquad
Q\hat \rho^{\alpha\beta}
    &= \hat \rho^{\alpha C}\hat \rho_{C}{}^{\beta}
       + \delta^{\alpha\beta}_{\gamma\delta}\hat \lambda^\gamma\hat \xi^\delta\,,
\\
Q\hat \lambda
    &= \hat \xi^A\hat \lambda_A\,,
&\qquad\qquad
Q\hat \lambda^\alpha
    &= \hat \lambda^B\hat \rho_{B}{}^{\alpha}
       -\hat \lambda\hat \lambda^\alpha\,.
\end{alignedat}
\end{align}
Identifying $\hat \xi^u,\hat \rho^{\,uA}, \hat \lambda^u$ as the translation ghosts, one readily recognizes \eqref{poincare-Qbdry} as the Chevalley--Eilenberg differential of the Poincar\'e algebra in $D$ dimensions. 

\end{proof}

\begin{rem} 
    On the boundary, this algebra is usually interpreted not as the Poincaré algebra in $D$ dimensions but as the global conformal Carrollian algebra in $D-1$ dimensions. We also note that the Chevalley--Eilenberg complex of the Carroll algebra in $D-1$ dimensions can be obtained by imposing the conditions $\hat\lambda=\hat\lambda^\alpha=0$.
\end{rem}

A complete analogue of the technical Proposition~\bref{prop-calculus-ads} can also be proven in this case:
\begin{prop}\label{utv-nabla-poincare}
For $N\geq1$ we have
\begin{align}
    \begin{split}
            b^{*}\nabla^N\xi^A&=\delta_{N,1}\hat C^A\,,\quad b^{*}\nabla^N\lambda=\delta_{N,1}\hat \lambda^u\,,\\
            b^{*}\nabla^N\xi^\Omega&=-\delta_{N,1}\hat \lambda-2\delta_{N,2}\hat \lambda^u\,,\\
            b^{*}\nabla^N\rho^{\alpha\beta}&=b^{*}\nabla^N\xi^u=b^{*}\nabla^N\lambda^\alpha=0\,,
    \end{split}
\end{align}
where we have introduced the following notation:
\begin{align}
    \hat C^A\equiv \hat\rho^{\,uA}\,.
\end{align}
\end{prop}
\begin{example}
    Consider the following cocycle on $\mathfrak{iso}[1]$:
    \begin{align}
        \tilde f\equiv\frac{1}{D!}\varepsilon_{a_1\dots a_D}\tilde\xi^{a_1}\dots\tilde\xi^{a_D}\,.
    \end{align}
    By Proposition~\bref{poincare-equiv} it can be uplifted to $Int(F)$. Applying Theorem~\bref{th-renorm} and Proposition~\bref{utv-nabla-poincare} one can easily see that
    \begin{align}
        f^\cR\propto \hat \xi^u\hat \lambda^u\hat C^{A_{1}}\dots \hat C^{A_{D-2}}\varepsilon_{\Omega u A_1\dots A_{D-2}},\qquad f^\cA=0\,.
    \end{align}
\end{example}
\subsubsection{Field-theoretical interpretation}\label{poincare-field}
Just like in the (A)dS case, the boundary gPDE  $\partial^QE\equiv(\partial^QF_{\red},Q)\times(T[1]\partial X,d_{\partial X})$ describes flat connections, except that they now take values in the Poincar\'e algebra. Indeed, the $Q$-map condition on sections $\sigma: T[1]\partial X\to \partial^Q E$ is the flatness condition, as can be seen from \eqref{poincare-Qbdry}.
However, the field-theoretic interpretation of this connection is not entirely standard, since the base is $d\equiv (D-1)$-dimensional. Introducing $\sigma^{*}\hat\xi^{\alpha}\equiv e^\alpha(x,\theta)\equiv e^\alpha{}_{\mu}(x)\theta^\mu$ we require $e^{\alpha}{}_{\mu}$ to be nondegenerate, which is again not entirely standard, since from the viewpoint of the Poincaré algebra they do not correspond to the translational part. In other words, from the viewpoint of Klein geometry, we consider a different coset, see \cite{Herfray:2021qmp} for a related discussion.

It is easy to see that 
\begin{equation}
\label{c-gauge}
    \sigma^*(\xi^\alpha)=\theta^\alpha\,, \qquad  \sigma^*(\text{the rest})=0
\end{equation} 
is clearly a solution satisfying the above nondegeneracy condition. Moreover,  as for any flat connection, any other solution is locally related to the above one by a gauge transformation. We will refer to the gauge \eqref{c-gauge} as the coordinate-like gauge.

As in the (A)dS case, we can introduce a more general gauge. To this end, we assume that $\partial X$ has the topology $\mathbb{R}\times \Sigma$ for some $\Sigma$, and we introduce adapted coordinates  $\{x^\mu\}\equiv \{u,y^i\}$, where $u$ is a coordinate on $\mathbb{R}$. We then fix the notation $\sigma^{*}\hat C^I \equiv C^I(x,\theta)\equiv C^I{}_{|\mu}(x)\theta^\mu$ for all coordinates except $\hat\xi^\alpha$, for which we keep the convention $\sigma^*\hat\xi^\alpha = e^{\alpha}{}_{\mu}(x)\theta^\mu$,  and require
\begin{align}
  e^A{}_{u}=0\,,\quad  e^u{}_{\mu}=\partial_\mu u\,, \quad e^{i}{}_{A}C^A{}_{|i}=0\,,\quad \lambda_{|\mu}=0\,,
\end{align}
where $e^{i}{}_{A}$ is the inverse of $e^A{}_{i}$. Note that the nondegeneracy of $e^\alpha{}_{\mu}$ together with $e^u{}_{\mu}=\partial_\mu u$ implies the nondegeneracy of $e^A{}_{i}$, so the inverse is well defined. Following \cite{Herfray:2021qmp}, we will refer to this gauge as the Bondi gauge, although in this gauge we do not require the scalar curvature of $\gamma_{ij}\equiv e^{A}{}_i g_{AB}e^B{}_{j}$ to be constant. In this gauge, the flatness condition implies in particular 
\begin{equation}\label{Bondi-formulas}
\begin{gathered}
\partial_u\gamma_{ij}
    = \lambda_{A|u}
    = \lambda_{[i|j]}
    = C_{B|u}
    = C_{[i|j]}
    = 0\,,
\qquad
\lambda_{i|j}\big|_{tf}
    = \partial_u C_{i|j}\,,
\\
2(d-2)\lambda^{i}{}_{|i}
    = -R[\gamma]\,,
\qquad
2(d-2)\lambda^{u}{}_{|u}
    = \frac{R[\gamma]}{d-1}\,,
\qquad
(d-2)\lambda^u{}_{|i}
    = -D^j C_{j|i}\, ,
\end{gathered}
\end{equation}
where we used $e^{A}{}_{i}$ to convert indices, e.g. $\lambda_{i|j}\equiv e^{A}{}_{i}\lambda_{A|j}$. Here $R[\gamma]$ is the Ricci scalar of $\gamma_{ij}$, $|_{tf}$ denotes the trace-free part, and $D_i$ is the Levi--Civita connection, both defined with respect to $\gamma_{ij}$. In the context of gravity at null infinity, $C_{i|j}$ is usually referred to as the asymptotic shift.
\subsection{Example: Scalar}

The examples of fields on a background with vanishing cosmological constant closely follow those on the $\Lambda\neq0$ background discussed in Sections~\bref{sec:scalar-ads} and \bref{sec:maxwell-ads}. For this reason, we will be rather brief in what follows, focusing primarily on the key formulas and the final results.

We start directly with the reduced model $(F^\varphi_{\red},Q,\eI_{\red}^{\varphi}) \xrightarrow{} (F_{\red},Q)$. As coordinates on the fiber of this bundle, we choose
\begin{align}\label{poincare-scalar-coord}
    \nabla_{(\alpha)}\varphi^{(N)}\equiv\nabla_{(\alpha)}\nabla^N\varphi\,,\qquad N\geq0\,,
\end{align}
with the action of $Q$ given by
\begin{align}\label{Qphi-poinc}
    Q\varphi=\xi\varphi^{(1)}+\xi^\alpha\nabla_\alpha\varphi+(\ell-d/2)\lambda\varphi-\lambda\Omega\varphi^{(1)}\,,
\end{align}
where $\ell\in\mathbb{N}/2$. The ideal $\eI^\varphi_{\red}$ is generated by $\nabla_a$ prolongations of
\begin{align} \label{poincare-scalar-ideal}
\begin{split}
\Omega(2\nabla_u\varphi^{(1)}+\nabla_A\nabla^A\varphi)+ (1-2\ell)\nabla_u \varphi
    \end{split}
\end{align}
which can be obtained as the pullback of \eqref{Scalar-ideal} to $F_{\red}$  upon setting $m=0$, $w=\ell-\frac{d}{2}$.

The induced boundary gPDE for $(F^\varphi_{\red},Q,\eI_{\red}^{\varphi}) \xrightarrow{} (F_{\red},Q)$
is given by $(\partial^{Q}F^{\varphi}_{\red},Q,\hat\eI^\varphi)\to(\partial^QF_{\red},Q)$. Using \eqref{poincare-scalar-coord} pulled back to $\partial^QF_{\red}$ as the fiber coordinates and employing \eqref{Qphi-poinc} and Proposition~\bref{utv-nabla-poincare} one finds:
\begin{multline}\label{scalar-poincare-Q-bdry}
Q\hat\varphi^{(N)}=\hat \xi^\alpha\nabla_\alpha\hat\varphi^{(N)}+N\hat C^A\nabla_A\hat\varphi^{(N-1)}+\\N\hat \lambda^u(\ell-\frac d2-N+1)\hat \varphi^{(N-1)} + \hat \lambda(\ell-\frac{d}{2}-N)\hat \varphi^{(N)}\,.
\end{multline}
The ideal $\hat \eI^\varphi=b^*\eI^\varphi$ determined by \eqref{poincare-scalar-ideal}, is generated by  $\nabla_\alpha$ prolongations of: 
\begin{align}\label{eq-null-scalar}
    (2N-2\ell +1)\nabla_u\hat\varphi^{(N)}+N\nabla^A\nabla_A\hat\varphi^{(N-1)}\,,\quad N\geq0\,.
\end{align}

Now we are ready to turn our attention to the  presymplectic potentials. In complete analogy with Proposition~\bref{ads-scalar-presymp}, we have
\begin{prop}
$1$-form \begin{align}
\tilde\chi_\varphi=\Omega^{-d/2-\ell+1}(\xi_{\Omega}\ms{d}\nabla_u\varphi+\xi\ms{d}_{u}\varphi^{(1)}+\xi\ms{d}_{A}\nabla^A\varphi)d(\Omega^{d/2-\ell}\varphi)\,,
\end{align}
on $Int(F^\varphi_\red)$ is a compatible presymplectic potential, i.e. it satisfies
    \begin{equation}
        L_Q\tilde\chi_\varphi=d(-\tilde H_{\varphi}-i_Q\tilde\chi_\varphi)+\cI_{Int(M)}+\eI^{\varphi}\,,
            \end{equation}
    where
    \begin{multline}
        \tilde H_\varphi=\frac{(-1)^{d+1}}{2}\Omega^{-2\ell}[\xi\ms{d+1}(\Omega(2\varphi^{(1)}\nabla_u\varphi+\nabla_A\varphi\nabla^A\varphi)+\\(d-2\ell)\varphi \nabla_u\varphi)+
        (\frac{d}{2}-\ell)\Omega\varphi^2\lambda^\alpha\xi_\alpha\ms{d}]\,.
    \end{multline}
Moreover, this system is equivalent to the scalar field minimal gPDE over $\mathfrak{iso}[1]$ with presymplectic structure~\eqref{chi-stand}, provided $\tilde\chi_\varphi$
is modified by a $d$-exact term and terms belonging to $\eI^{\varphi}$.
\end{prop}

It is clear that $\chi_\varphi=\Omega^{2\ell}\tilde\chi_\varphi$ and $H_\varphi\equiv\Omega^{2\ell}\tilde H_\varphi$ admit smooth extensions to the $Q$-boundary, and therefore we choose $p=2\ell$.
Unlike in the (A)dS case, the distinction between integer and half-integer $\ell$ is less substantial. We therefore proceed directly to the examples.
\begin{example}
Let $\ell=1/2.$ Then
    \begin{align}
        \chi^\cR_\varphi=\hat\xi\ms{d}\nabla_u\hat \varphi d\hat \varphi,\qquad \chi_\varphi^\cA=0\,.
    \end{align}
   In general, the presymplectic structure $\omega^{\cR}\equiv d\chi^\cR$ induces a presymplectic structure on the space of configurations of the boundary values. In this particular case, this can be related to the Hermitian product introduced by R. Sachs in~\cite{Sachs:1962zza}. 
\end{example}
  
\begin{example}
Let $\ell=1$. Then
    \begin{align}
    \begin{split}
                &\chi^\cA_\varphi=(-1)^d(\hat\xi\ms{d-1}_u\hat\varphi^{(1)}+\hat\xi_A\ms{d-1}\nabla^A\hat\varphi)d\hat\varphi\,,\\
                &H^{\cA}_l=\frac{1}{2}\hat\xi\ms{d}\nabla_A\hat\varphi\nabla^A\hat\varphi+\frac{d-2}{2}\hat\varphi(\frac{\hat\varphi}{2}\hat\lambda^\alpha\hat\xi_\alpha\ms{d-1}-\hat\lambda\hat\xi_u\ms{d-1}\hat\varphi^{(1)}-\hat\lambda\hat\xi_A\ms{d-1}\nabla^A\hat\varphi)\,.
    \end{split}
    \end{align}
Introducing coordinates on the space of sections of $\sigma: T[1]\partial X \to\partial^QF^\varphi_{\red}$ according to  
$\sigma^{*}\nabla_{\alpha_1}\dots\nabla_{\alpha_n}\hat\varphi^{(N)}\equiv\varphi^{(N)}_{\alpha_1\dots \alpha_n}(x)$,
the intrinsic action takes the following form:
\begin{multline}    
        S[\varphi^{(1)},\varphi,\varphi^A]=-\int e\ms{d}[\varphi^{(1)}\partial_u\varphi+\varphi^A\partial_A\varphi-\frac{1}{2}\varphi_A\varphi^A+\frac{2-d}{4}\varphi^2\lambda^{\alpha}{}_{|\alpha}+ \\  \frac{d-2}{2}\varphi(\lambda_{|u}\varphi^{(1)}+\lambda_{|A}\varphi^{A})]\,,
\end{multline}
where $\partial_\alpha \equiv e^\mu{}_{\alpha}\partial_\mu$, see Section~\bref{poincare-field} for the notation in the background sector. For $d\geq3$, working in the Bondi gauge for the background and eliminating $\varphi^A$, the action can be rewritten using \eqref{Bondi-formulas} as
\begin{align}\label{magnetic-scalar}
             S[\varphi,\varphi^{(1)}]=-\int dud^{d-1}y\sqrt{\gamma}(\varphi^{(1)}\partial_u\varphi+\frac{1}{2}\partial^{i}\varphi\partial_{i}\varphi+\frac{d-2}{8(d-1)}\varphi^2R[\gamma])\,,
\end{align}
where $\{x^\mu\}=\{u,y^i\}$ and $\gamma_{ij}\equiv e^A{}_{i}g_{AB}e^{B}{}_{j}.$
The corresponding equations of motion are:
\begin{align}
    \begin{split}
        &\partial_u\varphi=0\,,\qquad \partial_u\varphi^{(1)}+\frac{1}{\sqrt \gamma}\partial_{i}(\sqrt \gamma\partial^{i}\varphi)-\frac{d-2}{4(d-1)}\varphi R=0\,.
    \end{split}
\end{align}
The theory defined by these equations is known in the literature as the magnetic Carrollian scalar. For a constant metric $\gamma$, the action \eqref{magnetic-scalar} was obtained in \cite{Henneaux:2021yzg} as the ultra-relativistic limit of the Klein--Gordon action. See also \cite{Bekaert:2024itn}, where the magnetic scalar was studied at the level of equations of motion from the holographic perspective.
\end{example}
\begin{example}
    Let $\ell=3/2.$ Then
    \begin{multline}
            \chi^\cA_\varphi=(-1)^d[(\hat\xi\ms{d-1}_u\hat \varphi^{(2)}+\hat C^A\hat \xi\ms{d-2}_{uA}\hat \varphi^{(1)}+\hat  \xi_A\ms{d-1}\nabla^A\hat \varphi^{(1)}+\hat C^{B}\hat \xi_{AB}\ms{d-2}\nabla^A\hat \varphi)d\hat\varphi+\\(\hat \xi_u\ms{d-1}\hat \varphi^{(1)}+\hat \xi\ms{d-1}_A\nabla^A\hat \varphi)d\hat \varphi^{(1)}]\,,
    \end{multline}
\begin{multline}
   H^{\cA}_l=\hat\xi\ms{d}\nabla_A\hat \varphi^{(1)}\nabla^A\hat \varphi+\frac{1}{2}\hat C^B\hat\xi_B\ms{d-1}\nabla_A\hat \varphi\nabla^A\hat \varphi  + \frac{d-3}{4}\hat \lambda^\alpha\hat \varphi(2\hat \xi_\alpha\ms{d-1}\hat \varphi^{(1)}+\hat C^B\hat \xi\ms{d-2}_{\alpha B}\hat \varphi)-\\\frac{d-1}{2}\hat \varphi^{(1)}\hat \lambda(\hat\xi_u\ms{d-1}\hat \varphi^{(1)}+\hat \xi_A\ms{d-1}\nabla^A\hat \varphi)-\frac{d-3}{2}\hat\varphi[(\hat \lambda^u\hat \xi\ms{d-1}_{u} + 2\hat \lambda\hat C^C\hat \xi_{uC}\ms{d-2})\hat \varphi^{(1)}+ \\\hat \lambda\hat \xi_u\ms{d-1}\hat \varphi^{(2)}+(\hat \lambda^{u}\hat \xi\ms{d-1}_{A}+2\hat \lambda\hat C^C\hat \xi\ms{d-2}_{AC})\nabla^A\hat \varphi+\hat \lambda\hat \xi\ms{d-1}_{A}\nabla^A\hat\varphi^{(1)}]\,.
\end{multline}
From these data, one can directly construct the intrinsic action with the background written in terms of a Poincaré-valued connection. For simplicity, we present it in the Bondi gauge introduced in Section~\bref{poincare-field}. Introducing coordinates on the space of sections analogously to the previous example and using \eqref{Bondi-formulas}, we have for $d\geq3$,
    \begin{multline}
        S=-\int dud^{d-1}y\sqrt{\gamma}[\varphi^{(2)}\partial_u\varphi+(\varphi^{(1)|i}-C^{i}{}_{|j}\varphi^{j})\partial_{i}\varphi + \varphi^{i}\partial_{i}\varphi^{(1)} +\varphi^{(1)}\partial_u\varphi^{(1)}-\varphi_{i}^{(1)}\varphi^{i}-\\\frac{d-3}{2(d-2)}\varphi\varphi^{i}D^jC_{j|i}+\frac{d-3}{8}\varphi^2\partial_u(C^{i|j}C_{i|j})+\frac{d-3}{4(d-2)}R\varphi\varphi^{(1)}]\,,
    \end{multline}
    where $D_i$ is the Levi-Civita covariant derivative for the $d-1$ dimensional metric $\gamma$. The Euler-Lagrange equations can be split into the algebraic 
\begin{align}
\varphi_{i}=\partial_{i}\varphi,\quad   \varphi^{(1)}_{i}=\partial_{i}\varphi^{(1)}-C_{i|j}\partial^{j}\varphi-\varphi\frac{d-3}{2(d-2)}D^jC_{j|i}\,,
\end{align}
and the remaining ones:
\begin{align}\label{l=3/2}
    \partial_u\varphi=0\,,
\qquad
D_iD^i\varphi-\frac{d-3}{4(d-2)}R\varphi=0\,,
\end{align}
\begin{multline}
\partial_u\varphi^{(2)}
+D_i\varphi^{(1)|i}
-C^{i|j}D_iD_j\varphi
-\frac{d-1}{2(d-2)}D^i\varphi\, D^jC_{j|i}-\\
\frac{d-3}{4}\varphi\,\partial_u\!\left(C^{i|j}C_{i|j}\right)
-\frac{d-3}{4(d-2)}R\varphi^{(1)}
=0 .
\end{multline}

Note that \eqref{l=3/2} simply states that $\varphi$ describes a solution of the conformal Laplace equation on a $(d-1)$-dimensional space. We also note that in the case $d=3$ the formulas simplify considerably.
\end{example}

Similarly to Example~\bref{ads-scalar-even}, by putting the background in the coordinate-like gauge, we can obtain $\chi^\cR$ and $\chi^\cA$ in the general case:
\begin{prop}
    Restricting to the solution of the background given by $\hat\xi^\alpha=\theta^\alpha$ and all other background fields equal to zero, for $\ell\geq1$ we have 
\begin{align}
\begin{split}
&    \chi^\cR_\varphi=\frac{1}{(2\ell-1)!}\theta\ms{d}\sum_{i=0}^{2\ell-1}C^{i}_{2\ell-1}\nabla_u\hat\varphi^{(2\ell-1-i)}d\hat \varphi^{(i)}\,,\\
            &\chi^\cA_\varphi=\frac{(-1)^d}{(2\ell-2)!}\sum_{i=0}^{2\ell-2}C^{i}_{2\ell-2}(\theta\ms{d-1}_u\hat \varphi^{(2\ell-1-i)}+\theta_A\ms{d-1}\nabla^A\hat \varphi^{(2\ell-2-i)})d\hat \varphi^{(i)}\,,\\
            &H^{\cA}_l=\frac{1}{2(2\ell-2)!}\theta\ms{d}\sum_{i=0}^{2\ell-2}C^{i}_{2\ell-2}\nabla_A\hat \varphi^{(2\ell-2-i)}\nabla^A\hat \varphi^{(i)}\,.
        \end{split}
    \end{align}
The intrinsic action associated with $\chi_\varphi^\cA$ is given by
    \begin{align}\label{null-scalar-action}       S[\varphi^{(i)},\varphi_A^{(i)}]\propto\sum_{i=0}^{2\ell-2}C^i_{2\ell-2}\int du d^{d-1}y(\varphi^{(2\ell-1-i)}\partial_u\varphi^{(i)}+\varphi_A^{(2\ell-2-i)}\partial^A\varphi^{(i)}-\frac{1}{2}\varphi_A^{(2\ell-2-i)}\varphi^{(i)|A})\,.
    \end{align}
Eliminating the auxiliary fields $\varphi^{(i)}_A=\partial_A\varphi^{(i)}$, the solutions of this system are described by a collection of scalar fields $\varphi^{(N)}$, $N=0,\dots,2\ell-1$, satisfying the equations
\begin{align}\label{null-scalar-sections}
    (2N-2\ell+1)\partial_u\varphi^{(N)}+N\partial_A\partial^A\varphi^{(N-1)}=0\,,\quad 0\leq N\leq 2\ell-1\,.
\end{align}
Their  left-hand-sides have the same form as the generators of the ideal \eqref{eq-null-scalar}, truncated at the level $2\ell-1$. The equations \eqref{null-scalar-sections} can, of course, also be derived from the analysis of the massless scalar field on the Minkowski background, see e.g. \cite{Satishchandran:2019pyc, Bekaert:2024tkv}. However, to the best of our knowledge, the action \eqref{null-scalar-action} has not appeared in the literature before.

\end{prop}
\begin{rem}\label{remark-scalar-bundle}
    Note that, as follows from \eqref{scalar-poincare-Q-bdry}, $\partial^{Q}F^\varphi_{\red}$ carries the structure of an infinite $Q$-bundle 
    \begin{align}
            \dots\to(\partial^QF^\varphi_{\red})^{(i)}\to (\partial^QF^\varphi_{\red})^{(i-1)}\to\dots \to (\partial^QF^\varphi_{\red})^{(0)}\to \partial^QF_{\red}\,,
\end{align}
 where the fiber of $(\partial^QF^\varphi_{\red})^{(i)}\to (\partial^QF^\varphi_{\red})^{(i-1)}$ is coordinatized by $\nabla_{(\alpha)}\varphi^{(i)}$. The ideal $\hat\eI^{\varphi}$ can also be split into a union of $Q$-invariant ideals $(\hat\eI^{\varphi})^{(i)}$ in the obvious way.    Moreover, the presymplectic potentials $\chi^\cR$ and $\chi^\cA$ depend only on a finite number of $\hat\varphi^{\,(N)}$. More specifically, both of them can be obtained as the pullback from $(\partial^QF^\varphi_{\red})^{(2\ell-1)}$.
\end{rem}

\subsection{Example: Maxwell}
The exposition here parallels that of Section~\bref{sec:maxwell-ads} until the introduction of the reduced model $(F^\maxw_{\red},Q,\eI^\maxw_{\red})\to(F_{\red},Q)$ with an overcomplete set of fibre coordinates $C$, $\nabla_{a_1}\dots\nabla_{a_n}F_{bc}$. Using the splitting of the index $\{ a\} = \{\Omega, \alpha \} = \{\Omega, u, A \}$ we introduce the following notation
\begin{align}
    J_A\equiv F_{\Omega A}\,,\quad \nu \equiv F_{\Omega u}\,.
\end{align}

Then we can use this decomposition of the index to explicitly solve part of the Bianchi identities \eqref{Max-Bianchi}: 
\begin{align}\label{poincare-Max-Bianchi}
    \begin{split}
        &\nabla F_{uA}=\nabla_uJ_A-\nabla_A\nu\,,\\
        &\nabla F_{AB}=\nabla_AJ_B-\nabla_BJ_A\,,\\
    \end{split}
\end{align}
passing to overcomplete coordinates
\begin{align}\label{poincare-max-coord}
   C\,,\quad  \nabla_{(\alpha)}F_{\beta\gamma}\,, \nabla_{(\alpha)}\nu^{(N)}\,,\nabla_{(\alpha)}J_B^{(N)}\,,\qquad N\geq0\,.
\end{align}
 The ideal $\eI^\maxw_{\red}$ is generated by the $\nabla_{a}$-prolongations of
\begin{align}\label{poincare-max-eq}
    \begin{split}
        &Y_{\Omega}=(D-4)\nu-\Omega\nu^{(1)}-\Omega\nabla^AJ_A\,
        ,\quad Y_{u}=\Omega(\nabla_u\nu-\nabla^AF_{uA})\,,\\
        &Y_A=(4-D)F_{uA}+\Omega(2\nabla_uJ_A-\nabla_A\nu+\nabla^BF_{BA})\,,
    \end{split}
\end{align}
which can be obtained by pulling back \eqref{max-eq} to $\eI^\maxw_{\red}$ and using the above decomposition.

The action of $Q$ can be deduced by substituting the corresponding indices into
\begin{align}\label{poincare-max-Q}
\begin{split}
    QF_{ab}&=\xi\nabla F_{ab}+\xi^\alpha\nabla_\alpha F_{ab}+\rho_{a}{}^{c}F_{cb}+\rho_{b}{}^{c}F_{ac}-2\lambda F_{ab}-\lambda\Omega\nabla F_{ab}\,,\\
    QC&=(\xi-\lambda\Omega)\xi^\alpha F_{\Omega\alpha}+\frac{1}{2}\xi^\alpha\xi^\beta F_{\alpha\beta}
    \end{split}
\end{align}
and using $[Q, \nabla] = 0 = [Q, \nabla_\alpha]$.

The asymptotic boundary gPDE for $(F^\maxw_{\red},Q,\hat\eI^\maxw)\to(F_{\red},Q)$ is given by 
\begin{align}
    (\partial^QF^\maxw_{\red},Q,\hat\eI^\maxw)\to(\partial^QF_{\red},Q)\,.
\end{align}
Using as fiber coordinates on $\partial^QF^\maxw_{\red}$ functions \eqref{poincare-max-coord}, the action of $Q$ on them is obtained from \eqref{poincare-max-Q} using Proposition~\bref{utv-nabla-poincare}.

Using \eqref{poincare-Max-Bianchi} together with \eqref{poincare-max-eq}, one can check that the ideal $\hat \eI^{\maxw}\equiv b^*\eI^\maxw_{\red}$ is generated by the  $\nabla_{\alpha}$-prolongations of
\begin{align}\label{maxw-null-eom}
\begin{split}
    &\hat Y_\Omega^{(N)}=(d-3-N)\hat \nu^{(N)}-N\nabla^A\hat J_A^{(N-1)}\,,\\
    &\hat Y_u^{(d-2)}=(d-2)( \nabla_u\hat \nu^{(d-3)}-\nabla^A\hat F^{(d-3)}_{uA})\,,\\
    &\hat Y_A^{(0)}=(3-d)\hat F_{uA}\,,\\
    &\hat Y_A^{(N+1)}=(5-d+2N)\nabla_u\hat J_A^{(N)}+(d-4-N)\nabla_A\hat\nu^{(N)}+(N+1)\nabla^B\hat F_{BA}^{(N)}\,,
    \end{split}
\end{align}
where $N\geq 0$. The expressions for $\hat F_{AB}^{(N)}$ and $\hat F_{uA}^{(N)}$ in terms of the boundary coordinates can be obtained simply by applying $b^{*}\nabla^{N-1}$ to the relations~\eqref{poincare-Max-Bianchi}. Here we also used that the generators $\hat Y_u^{(N)}$ with $N\neq d-2$ are not independent and can be removed from the generating set of the ideal.  An analogous feature already appeared in the AdS case, see \eqref{maxwell-noether}, and can be traced back to the fact that the Maxwell equations are not independent and satisfy Noether identities.

In complete analogy with Proposition~\bref{ads-max-presymp} one obtains the following:
\begin{prop}
    The presymplectic potential defined by 
    \begin{align}
                &\tilde \chi_\maxw= \Omega^{4-D}(-2\xi_{\Omega u}\ms{D-2}\nu+2\xi_{uB}\ms{D-2}J^B+2\xi_{\Omega A}\ms{D-2}F_{u}{}^{A}+\xi_{AB}\ms{D-2}F^{AB})dC
    \end{align}
    on $Int(F^\maxw_\red)$ satisfies
    \begin{align}
    \begin{split}
               & L_Q\tilde\chi_\maxw=d(-\tilde H_{\maxw}-i_Q\tilde\chi_{\maxw})+\cI_{Int(M)}+\eI_{\red}^\maxw\,,\\
        &\tilde H_\maxw=-\frac{1}{2}\Omega^{4-D}\xi\ms{D}(-2\nu^2+4F_{u}{}^{A}J_A+F^{AB}F_{AB})\,.
    \end{split}
\end{align}
Moreover, the present system is equivalent to the gPDE with presymplectic structure described in Section~\bref{sec:example-M}, in the sense that it splits as the product of that gPDE and a contractible gPDE with trivial presymplectic structure.
\end{prop}
It follows that, upon introducing $\chi_\maxw\equiv\Omega^{D-4}\tilde\chi_\maxw$, we can apply the procedure described in Section~\bref{sec: renorm-struc} to obtain $\chi_\maxw^\cR$, $\chi_\maxw^\cA$ and the corresponding Hamiltonians. In the rest of this section we explicitly present them.
\begin{example}
    Let $d=3$. 
    \begin{align}
        \chi_\maxw^\cR=-2(\hat\xi_u\ms{2}\hat\nu-\hat\xi_A\ms{2}\hat F_{u}{}^{A})d\hat C, \qquad H^\cR=0\,,\quad \chi^\cA_\maxw=0\,.
    \end{align}
With the notation
\begin{align}
    \ev^*(\hat C)=\cC+\dots,\quad \ev^*\hat\nu=\nu+\dots,\quad \ev^{*}\hat F_{uA}=F_{uA}+\dots
\end{align}
 and assuming the coordinate-like gauge for the background, the BV--AKSZ action corresponding to $\chi_\maxw^\cR$ takes the form
\begin{align}
    \int dud^{2}y(\partial_u\nu-\partial^AF_{uA})\cC
\end{align}
up to a boundary term and an overall constant.
The coefficient of $\cC$ takes a form similar to that of $Y_u^{(d-2)}$ in \eqref{maxw-null-eom}.
\end{example}
\begin{example}
    Let $d=4$. Then, up to $\hat\eI^\maxw$,
    \begin{align}
        \begin{split}
    \chi_\maxw^\cA=-(2\hat\xi_{uB}\ms{2}\hat J^B+\hat\xi\ms{2}_{AB}\hat F^{AB})d\hat C,\quad H^{\cA}_l =\frac{1}{2}\hat\xi\ms{4}\hat F_{AB}\hat F^{AB}\,.
        \end{split}
    \end{align}
Choosing the Bondi gauge for the background, we introduce the following coordinates on the space of sections:
\begin{equation}
    \begin{gathered}      
     F_{ij}\equiv e^{A}{}_{i}e^B{}_{j}\sigma^{*}\hat F_{AB}\,, \qquad \sigma^*\hat C = A_{\mu} \theta^\mu\,, \qquad
       J_i=e^{A}{}_i \sigma^* \hat J_A\,.
    \end{gathered}
\end{equation}
The intrinsic action takes the following form:
    \begin{align}
        S[A,F,J]=-\int_{\partial X}dud^{D-2}y\sqrt\gamma(2J^i(\partial_u A_i-\partial_i A_u)+2F^{ij}\partial_{i}A_j-\frac{1}{2}F_{ij}F^{ij})\,.
    \end{align}
An equivalent action can be obtained via a ``magnetic'' contraction of the Maxwell action, see \cite{Henneaux:2021yzg}. The equations of motion are given by:
    \begin{align}
        \begin{split}
            &F_{ij}=\partial_i A_j-\partial_j A_i\,,\qquad
            \partial_uA_i-\partial_i A_u=0\,,\\
            &\partial_i(\sqrt\gamma J^i)=0\,,\qquad 
            \partial_uJ^j+\frac{1}{\sqrt\gamma }\partial_i(\sqrt\gamma F^{ij})=0\,.
        \end{split}
    \end{align}
One also has, up to $\hat\eI^\maxw$.:
    \begin{align}
        \chi^\cR_\maxw=-2(\hat\xi\ms{3}_{u}\hat\nu^{(1)}+\hat\xi_{A}\ms{3}\nabla^B\hat F^{BA}+\hat\lambda(\hat\xi_{uB}\ms{2}\hat J^B+\frac{1}{2}\hat\xi_{AB}\ms{2}\hat F^{AB}))d\hat C\,.
    \end{align}
\end{example}

\begin{prop}
For $d\geq 5$ we have 
\begin{multline}
    (-1)^{d+1}(d-4)!\chi^\cA_\maxw=2(\theta_{uB}\ms{d-2}\hat J^{(d-4)|B}+\theta_{AB}\ms{d-2}\nabla^A\hat J^{(d-5)|B})dC+\\(2\theta_{uB}\ms{d-2}\hat J^{(0)|B}+\theta_{AB}\ms{d-2}\hat F^{AB})d(\frac{d-5}{2}\theta^u\nabla^A\hat J_A^{(d-6)}+\theta^{C}\hat J_C^{(d-5)})+\\2\sum_{i=1}^{d-5}C^i_{d-4}(\theta\ms{d-2}_{uB}\hat J^{(d-4-i)|B}+\theta\ms{d-2}_{AB}\nabla^A\hat J^{(d-5-i)|B})d(\frac{i-1}{(d-2-i)}\theta^u \nabla^A\hat J^{(i-2)}_A+\theta^C\hat J_C^{(i-1)})
\end{multline}
and 
\begin{multline}
    (-1)^d(d-4)!H^{\cA}_l=\theta\ms{d}[2\hat F_{AB}\nabla^A\hat J^{(d-5)|B}+\\\sum_{i=1}^{d-5}C^i_{d-4}(\frac{i(d-4-i)}{(i+1)(d-3-i)}\nabla^A\hat J^{(i-1)}_A\nabla^B\hat J_B^{(d-5-i)}+2(\nabla_{[A}\hat J_{B]}^{(i-1)}\nabla^A\hat J^{(d-5-i)|B})]\,,
\end{multline}
where we employed the coordinate-like gauge and omitted terms from $\hat\eI^\maxw$.
Also, $\sum_{i=1}^{0}(\dots)$ is assumed to vanish.

The corresponding intrinsic action can be written down immediately. To present it in a more compact form
we parameterize the field configurations according to 
\begin{align}
\sigma^{*}\hat C=\theta^uA_u+\theta^{B}J_B^{(-1)},\qquad
\sigma^{*}\hat J_B^{(i)}=J_B^{(i)},\quad i\geq0,
\end{align}
eliminate auxiliary fields and drop  boundary terms, giving
\begin{multline}\label{null-maxwell-action}
S[A_u,J^{(i)}_B]\propto
\int_{\partial X}dud^{d-1}y[-J^{(d-4)|B}\partial_BA_u+\sum_{i=0}^{d-4}C^{i}_{d-4}(J^{(d-4-i)|B}\partial_uJ_B^{(i-1)}
+\\
\partial_{[A}J_{B]}^{(i-1)}
\partial^{[A}J^{(d-5-i)|B]}
+\frac{i(d-4-i)}{2(i+1)(d-3-i)}
\partial_AJ^{(i-1)|A}
\partial_BJ^{(d-5-i)|B}
)
] \,,
\end{multline}
where $J_B^{(-2)}\equiv 0$. The corresponding equations of motion are
\begin{align}
    \begin{split}
        &\partial_uJ^{(-1)}_B-\partial_BA_u=0\,,\quad \partial_BJ^{(d-4)|B}=0\,,\\
       &(d-5-2i)\partial_uJ^{(i)}_B-2(i+1)\partial^A\partial_{[A}J_{B]}^{(i-1)}-\frac{i(d-4-i)}{d-3-i}\partial_B\partial_AJ^{(i-1)|A}=0\,,
    \end{split}
\end{align}
where $i=0,\dots,d-4$.  After some algebra, it is straightforward to verify that these equations take the same form form as a subset of \eqref{maxw-null-eom}. Similar asymptotic hierarchies for Maxwell fields near null infinity were studied in \cite{Satishchandran:2019pyc,Bekaert:2024tkv}. However, to the best of our knowledge, the action \eqref{null-maxwell-action} has not appeared in the literature before. 
\end{prop}

\section*{Acknowledgments}
We are grateful to Xavier Bekaert,  Nicolas Boulanger, and Matthieu Vilatte for useful discussions.
M.G. also acknowledges fruitful exchanges with Fedor Popelensky. M.M. would also like to thank Viacheslav Krivorol and Bulat Farkhtdinov for valuable discussions.

\appendix
\section{Independence of the Boundary Cocycles}\label{app:cocycles-ind}
Suppose that $(M,Q)$ is a $Q$-manifold with a $Q$-boundary and $\Omega$ is a defining function for the boundary, $\xi\equiv Q\Omega$. Let us choose a vector field $\md$ such that
\begin{align}\label{app-coord-md}
    \md: \quad \gh{\md}=-1\,,\quad \md \xi=1\,,\quad \md^2=0\,.
\end{align}
Let $v$ be a form on $M$ satisfying 
 \begin{align}\label{app-coord-cocycle}
   \Omega L_Qv=p\xi v\,,
 \end{align}
then on the boundary we obtain cocycles $v^\cR(\md)$ and $v^\cA(\md)$, whose definition \eqref{RA-map} a priori depends on the choice of $\md$, as indicated here by the notation. We now show that their cohomology classes are independent of the choice of admissible $\md$.
\begin{prop}
Let $\md'$ be another vector field satisfying the properties \eqref{app-coord-md}. Then the corresponding boundary cocycles satisfy
\begin{align}
    v^\cR(\md')=v^\cR(\md)+L_Q(\dots)\,,\quad v^\cA(\md')=v^\cA(\md)+L_Q(\dots)\,.
\end{align}
In other words, their cohomology classes with respect to $L_Q$ are independent of the choice of $\md$.\end{prop}
\begin{proof}
For a given $\md$, any admissible $\md'$ can be written as $\md'=\md+X$, where the vector field $X$ satisfies
   \begin{align}
         \gh{X}=-1\,,\quad X\xi=0\,,\quad [\md,X]+\frac 12[X,X]=0\,.
    \end{align}
Introducing an auxiliary parameter $t$, we define
\begin{align}
    \md_t\equiv\md+tX\,,\quad \nabla_t\equiv[Q,\md_t]\,.
\end{align}
Clearly, one has
\begin{align}
     X=\frac{d}{dt}\md_t\,,\quad \md_t\xi=1,\quad\dot\nabla_t=[X,Q]\,.
\end{align}
Now consider $p!v^\cR(\md_t)\equiv b^{*}L_{\nabla_t}^{p}v$ and differentiate it with respect to $t$:
        \begin{multline}\label{app-coord-vr}
       p!\frac{d}{dt}v^{\cR}(\md_t)=b^*\sum_{i=0}^{p-1}L_{\nabla_t}^{i}L_{[Q,X]}L_{\nabla_t}^{p-1-i}v=L_Qb^{*}\sum_{i=0}^{p-1}L_{\nabla_t}^{i}L_{X}L_{\nabla_t}^{p-1-i}v+b^{*}\sum_{i=0}^{p-1}L_{\nabla_t}^{i}L_XL_{\nabla_t}^{p-1-i}L_Qv\,.
 \end{multline}
We show that the last sum vanishes. Acting on \eqref{app-coord-cocycle}  with $b^*\sum_{i=0}^{p}L_{\nabla_t}^{i}L_XL_{\nabla_t}^{p-i}$, we see that the right-hand side vanishes, while the left-hand side becomes
 \begin{multline}
     b^{*}\sum_{i=0}^p L_{\nabla_t}^{i}L_XL_{\nabla_t}^{p-i}(\Omega L_Qv)=b^{*}\sum_{i=0}^pL^i_{\nabla_t}(\Omega L_XL_{\nabla_t}^{p-i}L_Qv+(p-i)L_XL_{\nabla_t}^{p-i-1}L_Qv)=\\=(p+1)b^{*}\sum_{i=0}^{p-1}L^i_{\nabla_t}L_XL_{\nabla_t}^{p-i-1}L_Qv\,.
 \end{multline}
Integrating \eqref{app-coord-vr} with respect to $t$, we obtain
\begin{align}
     v^{\cR}(\md')-v^{\cR}(\md)=\frac{1}{p!}L_Q\int_{0}^1dt (b^{*}\sum_{i=0}^{p-1}L_{\nabla_t}^{i}L_{X}L_{\nabla_t}^{p-1-i}v)\,.
 \end{align}
 Similarly, 
 \begin{multline}\label{app-coord-an}
     (-1)^{tot(v)}(p-1)!\frac{d}{dt}v^\cA(\md_t)=\frac{d}{dt}b^{*}L_{\nabla_t}^{p-1}L_{\md_t}v=b^{*}\sum_{i=0}^{p-2}L_{\nabla_{t}}^{i}L_{[Q,X]}L_{\nabla_{t}}^{p-2-i}L_{\md_t}v+b^{*}L_{\nabla_t}^{p-1}L_{X}v=\\=L_Q(b^{*}\sum_{i=0}^{p-2}L_{\nabla_{t}}^{i}L_{X}L_{\nabla_{t}}^{p-2-i}L_{\md_t}v)+b^{*}\sum_{i=0}^{p-2}L_{\nabla_{t}}^{i}L_{X}L_{\nabla_{t}}^{p-2-i}L_QL_{\md_t}v+b^{*}L_{\nabla_t}^{p-1}L_{X}v=\\=L_Q(b^{*}\sum_{i=0}^{p-2}L_{\nabla_{t}}^{i}L_{X}L_{\nabla_{t}}^{p-2-i}L_{\md_t}v)+b^{*}\sum_{i=0}^{p-1}L_{\nabla_{t}}^{i}L_{X}L_{\nabla_{t}}^{p-1-i}v-b^{*}\sum_{i=0}^{p-2}L^{i}_{\nabla_t}L_XL_{\nabla_t}^{p-2-i}L_{\md_t}L_Qv.
 \end{multline}
 Acting on \eqref{app-coord-cocycle} with $\frac{1}{p}b^*\sum_{i=0}^{p-1}L^i_{\nabla_t}L_XL_{\nabla_t}^{p-1-i}L_{\md_t}$ a straightforward computation gives
 \begin{align}
            b^{*}\sum_{i=0}^{p-2}L^i_{\nabla_t}L_XL_{\nabla_t}^{p-2-i}L_{\md_t}L_Qv= b^{*}\sum_{i=0}^{p-1}L^i_{\nabla_t}L_XL_{\nabla_t}^{p-1-i}v\,.
 \end{align}
Substituting this identity into \eqref{app-coord-an} and integrating with respect to $t$ we obtain
 \begin{align}
     v^{\cA}(\md')-v^{\cA}(\md)=\frac{(-1)^{tot(v)}}{(p-1)!}L_Q\int_{0}^1dt(b^{*}\sum_{i=0}^{p-2}L_{\nabla_{t}}^{i}L_{X}L_{\nabla_{t}}^{p-2-i}L_{\md_t}v)\,.
 \end{align}
 \end{proof}

 We then turn to the independence of the map the choice of boundary defining function $\Omega$. The corresponding change in $v^\cR$ and $v^\cA$ arises from the change in $\Omega$ and $v \equiv \Omega^{p}\tilde v$.
\begin{prop}
Let $\tilde v$ be a $Q$-cocycle, and $\alpha\in C^\infty(M)$. For
    \begin{align}
        \Omega'=\exp(\alpha)\Omega,\quad v'=\exp(p\alpha)v\,,
    \end{align}
    we have
    \begin{align}
        v'^\cR-v^\cR=L_Q(\dots),\qquad v'^\cA-v^{\cA}=L_Q(\dots)\,,
    \end{align}
    so that the cohomology classes of $v^\cR$ and $v^\cA$  remain unchanged.
\end{prop}
\begin{proof}
We interpolate between the initial and the rescaled defining functions by introducing the one-parameter family
 \begin{align}
     \Omega_t=\exp(\alpha t)\Omega,\quad v_t=\exp(p\alpha t)v\,,\quad\xi_t\equiv Q\Omega_t\,.
 \end{align}
 The corresponding vector fields $\md_t$ are chosen so that $\md_t\xi_t=1$. It follows that
 \begin{align}
     \dot \md_t\xi_t=-\md_t\dot\xi_t=-\md_tQ(\alpha\Omega_t)=-\nabla_t(\alpha\Omega_t)\,,
     \end{align}
     where $\nabla_t=[Q,\md_t]$. We also have
     \begin{align}\label{app-coord-cocycle2}
         \Omega_tL_Qv_t=p\xi_t v_t\,.
     \end{align}
     Differentiating the definition of $v_t^\cR$ with respect to $t$, we obtain
     \begin{multline}
         p!\frac{d}{dt}v^\cR_t=b^{*}\sum_{i=0}^{p-1}L^i_{\nabla_t}L_{[Q,\dot \md_t]}L_{\nabla_t}^{p-i-1}v_t+pb^{*}L^p_{\nabla_t}(\alpha v_t)=\\=L_Qb^{*}\sum_{i=0}^{p-1}L^i_{\nabla_t}L_{\dot\md_t}L^{p-i-1}_{\nabla_t}v_t+b^{*}\sum_{i=0}^{p-1}L^i_{\nabla_t}L_{\dot\md_t}L^{p-i-1}_{\nabla_t}L_Qv_t+pL^p_{\nabla_t}(\alpha v_t).
     \end{multline}
     Now we act with $b^{*}\sum_{i=0}^{p}L^i_{\nabla_t}L_{\dot\md_t}L^{p-i}_{\nabla_t}$ on both sides of \eqref{app-coord-cocycle2}. The left-hand side gives
     \begin{align}
         b^*\sum_{i=0}^pL^i_{\nabla_t}(\Omega_tL_{\dot\md_t}L_{\nabla_t}^{p-i}+(p-i)L_{\dot\md_t}L^{p-i-1}_{\nabla_t})L_Qv_t=(p+1)b^{*}\sum_{i=0}^{p-1}L^i_{\nabla_t}L_{\dot\md_t}L^{p-i-1}_{\nabla_t}L_Qv_t
     \end{align}
     while the right-hand side becomes
     \begin{multline}
         p b^{*}\sum_{i=0}^{p}L^i_{\nabla_t}L_{\dot\md_t}L_{\nabla_t}^{p-i}(\xi_t v_t)=p b^{*}\sum_{i=0}^{p}L^i_{\nabla_t}(-L_{\nabla_t}(\alpha\Omega_t))L^{p-i}_{\nabla_t}v_t)=\\=-pb^{*}\sum_{i=0}^{p}\sum_{j=0}^{i}(j+1)C^j_iL^j
_{\nabla_t}\alpha L^{p-j}_{\nabla_t}v_t  =-pb^{*}\sum_{j=0}^{p}(\sum_{i=j}^pC^j_i)(j+1)L^j_{\nabla_t}\alpha L_{\nabla_t}^{p-j}v_t   =\\=-pb^{*}\sum_{j=0}^{p}(j+1)C^{j+1}_{p+1}L^j_{\nabla_t}\alpha L_{\nabla_t}^{p-j}v_t=-p(p+1)b^{*}L_{\nabla_t}^{p}(\alpha v_t)\,.
\end{multline}
Substituting the result into the expression for $\frac{d}{dt}v^\cR_t$ and integrating over $t$ gives
\begin{align}
    v'^{\cR}-v^\cR=\frac{1}{p!}L_Q\int_0^1dt(b^{*}\sum_{i=0}^{p-1}L^i_{\nabla_t}L_{\dot\md_t}L^{p-i-1}_{\nabla_t}v_t)\,.
\end{align}
The proof for $v^\cA$ proceeds analogously. Differentiating the definition of $v_t^\cA$ with respect to $t$, we obtain
\begin{multline}
    (-1)^{tot(v)}(p-1)!\frac{d}{dt}v^\cA_t=b^{*}\frac{d}{dt}L_{\nabla_t}^{p-1}L_{\md_t}v_t=\\=b^{*}\sum_{i=0}^{p-2}L_{\nabla_t}^{i}L_{[Q,\dot\md_t]}L_{\nabla_t}^{p-2-i}L_{\md_t}v_t+b^{*}L_{\nabla_t}^{p-1}(L_{\dot\md_t}v_t+p\alpha L_{\md_t}v_t)=\\=L_Qb^{*}\sum_{i=0}^{p-2}L_{\nabla_t}^{i}L_{\dot\md_t}L_{\nabla_t}^{p-2-i}L_{\md_t}v_t+b^{*}\sum_{i=0}^{p-2}L_{\nabla_t}^{i}L_{\dot\md_t}L_{\nabla_t}^{p-2-i}L_QL_{\md_t}v_t+b^{*}L_{\nabla_t}^{p-1}(L_{\dot\md_t}v_t+p\alpha L_{\md_t}v_t)
\end{multline}
Applying $\frac{1}{p}b^{*}\sum_{i=0}^{p-1}L^i_{\nabla_t}L_{\dot\md_t}L_{\nabla_t}^{p-1-i}L_{\md_t}$ to the relation \eqref{app-coord-cocycle2}, a straightforward computation gives
\begin{align}
    b^{*}\sum_{i=0}^{p-2}L^i_{\nabla_t}L_{\dot\md_t}L_{\nabla_t}^{p-2-i}L_{\md_t}L_Qv_t=b^{*}\sum_{i=0}^{p-1}L^i_{\nabla_t}L_{\dot\md_t}L_{\nabla_t}^{p-1-i}v_t+pb^{*}L_{\nabla_t}^{p-1}(\alpha L_{\md_t}v_t)\,.
\end{align}
Substituting this identity into the previous expression, using $\nabla_t=[Q,\md_t]$, and integrating with respect to $t$, we arrive at
\begin{align}
v'^\cA-v^{\cA}=\frac{(-1)^{tot(v)}}{(p-1)!}L_Q\int_{0}^1dt(b^{*}\sum_{i=0}^{p-2}L_{\nabla_t}^{i}L_{\dot\md_t}L_{\nabla_t}^{p-2-i}L_{\md_t}v_t)\,.
\end{align}
\end{proof}
 \begin{rem}
For $p=1$, the $L_Q$-exact term vanishes in both proofs. In other words, in this case $v^\cR$ and $v^\cA$ are independent of the choice of $\md$ and of the boundary defining coordinate $\Omega$.
\end{rem}

\section[Properties of R and A maps]{Properties of $\cR$ and $\cA$ maps} \label{app:properties-maps}

\begin{lemma} 
    The maps $\cR, \cA$ both commute with $L_Q$ and $d$.
\end{lemma}
\begin{proof}
Let $\tilde{\beta} \in \bigwedge^\bullet{}_{\fop}(Int(F))$ and $\tilde{\alpha} = L_Q \tilde{\beta}$. Choose $p$ such that $\beta = \Omega^{p-1}\tilde{\beta}$ is smooth on the boundary. Let us calculate $\cR(L_Q\tilde{\beta}) = \cR(\tilde{\alpha})$:
\begin{multline}
    \alpha^\cR = \frac{1}{p!} b^* L^p_\nabla \Omega^pL_Q \tilde{\beta} =\frac{1}{p!}  b^* ( L_Q (L^p_\nabla \Omega^p\tilde{\beta}) - p\xi L^p_\nabla \Omega^{p-1}\tilde{\beta}) = \\ = L_Q b^* \frac{1}{p!} ( (L^p_\nabla \Omega^p\tilde{\beta})) = L_Q \beta^\cR
\end{multline}
and for $\cA(\tilde\alpha)$:
\begin{multline}
    \alpha^\cA = \frac{(-1)^{tot(\tilde{\beta})+1}}{(p-1)!}b^* L^{p-1}_\nabla \Omega^p L_{\md} L_Q \tilde{\beta} = \frac{(-1)^{tot(\tilde{\beta})+1}}{(p-1)!}b^*(L^{p-1}_\nabla \Omega^p L_\nabla \tilde{\beta} - L^{p-1}_\nabla \Omega^p L_Q L_{\md}  \tilde{\beta}) = \\ = \frac{(-1)^{tot(\tilde{\beta})+1}}{(p-1)!}b^* \Big[ L^{p}_\nabla \Omega^p  \tilde{\beta} - p L^{p-1}_\nabla \Omega^{p-1} \tilde{\beta} - L_Q L^{p-1}_\nabla \Omega^p L_{\md} \tilde{\beta} + p L^{p-1}_\nabla \xi \Omega^{p-1} L_{\md} \tilde{\beta} \Big] = \\ = \frac{(-1)^{tot(\tilde{\beta})+1}}{(p-1)!}b^*\Big[ \Omega L_\nabla^p\beta -  L_Q L^{p-1}_\nabla \Omega^p L_{\md} \tilde{\beta} + p \xi L^{p-1}_\nabla   L_{\md} \beta \Big] = \\ = L_Q \frac{(-1)^{tot(\tilde{\beta})+1}}{(p-1)!}b^* (-L^{p-1}_\nabla \Omega^p L_{\md} \tilde{\beta}) = L_Q \frac{(-1)^{tot(\tilde{\beta})}}{(p-1)!}b^* (L^{p-1}_\nabla \Omega^p L_{\md} \tilde{\beta}) = L_Q \beta^\cA
\end{multline}
Therefore, we have shown that $\cR(L_Q \tilde{\beta}) = L_Q \cR(\tilde{\beta})$ and $\cA(L_Q \tilde{\beta}) = L_Q \cA(\tilde{\beta})$

To show that $\cR,\cA$ intertwine the de Rham differential we consider $\tilde{\alpha} = d\tilde{\beta}$ and choosing $p$ such that $\beta = \Omega^{p-1} \tilde\beta$ is smooth on the boundary, we compute:
\begin{equation}
\begin{gathered}
    \alpha^\cR = b^* \frac{1}{p!}L^p_\nabla \Omega^p d \tilde{\beta} = \frac{1}{p!}b^*( d L^p_\nabla \Omega^p  \tilde{\beta} - p L^p_\nabla d(\Omega) \beta ) = d \beta^\cR
\end{gathered}
\end{equation}
and 
\begin{multline}   
    \alpha^\cA = \frac{(-1)^{tot(\tilde{\beta})+1}}{(p-1)!}b^* L^{p-1}_\nabla \Omega^p L_{\md} d \tilde{\beta} = \\=\frac{(-1)^{tot(\tilde{\beta})+1}}{(p-1)!}b^*(-d L^{p-1}_\nabla \Omega^p L_{\md} d \tilde{\beta} + p (d\Omega )L^{p-1}_\nabla  L_{\md} \beta ) = d \beta^\cA
\end{multline}
proving that $\cR(d\tilde{\beta}) = d \cR(\tilde\beta)$  and $\cA(d\tilde{\beta}) = d \cA(\tilde\beta)$.
\end{proof}
\begin{lemma}
    If $\cI$ is extendable the maps $\cR, \cA$ can be restricted to 
    \begin{equation}
    \begin{gathered}
        \cR: \bigwedge{}^\bullet_{\fop}/\cI|_{\fop} \xrightarrow{} \bigwedge^\bullet(\partial^Q F)/\cI_{\partial} \\
        \cA: \bigwedge{}^\bullet_{\fop}/\cI|_{\fop} \xrightarrow{} \bigwedge^\bullet(\partial^Q F)/\cI_{\partial}
    \end{gathered}
    \end{equation}
\end{lemma}
\begin{proof}
    One has to show that $\cR, \cA$ are well defined on equivalence classes. Suppose we pick an equivalence class $[\tilde{\alpha}]$ and two representatives $\tilde{\alpha}, \tilde{\alpha}' \in [\tilde{\alpha}]$, i.e. $\tilde{\alpha} -  \tilde{\alpha}' = \tilde{\beta} \in \cI_{\fop}$. Because the maps are linear over $\tilde{\alpha}$ we only have to show that they map elements of $\cI|_{\fop}$ to $\cI_\partial$. We have
    \begin{equation}
        \beta^\cR = b^*( \frac{1}{(p+1)!}L^{p+1}_\nabla \Omega^{p+1}\tilde{\beta})
    \end{equation}
    \begin{equation}
        \beta^\cA = (-1)^{tot(\tilde{\beta})}b^*( \frac{1}{p!}L^{p}_\nabla \Omega^{p+1}L_{\md}\tilde{\beta}))
    \end{equation}
    Because $\Omega^p\tilde{\beta}$ is smooth on the boundary and using that $L_Q, L_{\md}$ and consequently $L_\nabla$ preserve $\cI$ we conclude that the claim is true.
\end{proof}

\section{Proof of Proposition~\bref{prop-renorm} } \label{app:renorm}
\begin{proof}
    By definition \eqref{eq-v(N)},
    \begin{multline}
       v^{(p)}=L_\nabla^pv=L_\nabla^p(\Omega^{p}\tilde v) = \sum_{i=0}^p C^{i}_p\nabla^i(\Omega^{p})L_\nabla^{p-i}\tilde v
       =p!\tilde v+\sum_{i=0}^{p-1}C^{i}_p\frac{p!}{(p-i)!}\Omega^{p-i}L_\nabla^{p-i-1}L_QL_{\md}\tilde v\,,
    \end{multline}
        where in the last equality we used $\nabla=[Q,\md]$ and $L_Q\tilde v=0$. Next, we move $L_Q$ to the left and rewrite $\tilde v=\Omega^{-p}v$, $\bar v\equiv L_\md v$:
        \begin{multline}\label{renorm-proof-AB}
        v^{(p)}=p!\tilde v+L_Q\left(\sum_{i=0}^{p-1}C^{i}_p\frac{p!}{(p-i)!}\Omega^{p-i}L_\nabla^{p-i-1}(\Omega^{-p}\bar v)\right)-\\\sum_{i=0}^{p-1}C^{i}_p\frac{p!}{(p-i)!}(p-i)\Omega^{p-i-1}(Q\Omega )L_\nabla^{p-i-1}(\Omega^{-p}\bar v)\equiv p!\tilde v+L_QA+B\,,
    \end{multline}
    where we have introduced the notation $A$ and $B$ for the corresponding sums.

  Let us compute $A$, again using the Leibniz rule and $\nabla^j\Omega^{-p}=(-1)^j\frac{(p+j-1)!}{(p-1)!}\Omega^{-p-j}$:
    \begin{align}
        A=\sum_{i=0}^{p-1}\sum_{j=0}^{p-i-1}C^{i}_pC^j_{p-i-1}\frac{p!}{(p-i)!}(-1)^j\frac{(p+j-1)!}{(p-1)!}\Omega^{-i-j}\bar v^{(p-i-j-1)}\,.
    \end{align}
    Making the change of summation variable $s=i+j$ and interchanging the order of summation, we arrive at
    \begin{align}\label{proof-renorm-A}
        A=\sum_{s=0}^{p-1}(\sum_{i=0}^{s}C^{i}_pC^{s-i}_{p-i-1}\frac{p}{(p-i)!}(-1)^{s-i}(p+s-i-1)!)\Omega^{-s}\bar v^{(p-s-1)}
        \equiv \sum_{s=0}^{p-1}a_{p,s}\Omega^{-s}\bar v^{(p-s-1)}\,.
    \end{align}
    Let us compute $a_{p,s}.$ Clearly, $a_{p,0}=1$. For $s>0,$ factoring out the terms independent of the summation index and rewriting the factorials in terms of binomial coefficients, we get
     \begin{align}
        a_{p,s}=p(-1)^s (s-1)!\sum_{i=0}^{s}(-1)^iC^{i}_pC^{s-i}_{p-i-1}C^{p-i}_{p+s-i-1},.
    \end{align}
    We now use the identities $C^{i}_p=C^{p-i}_p$, $C_{p-1-i}^{s-i}=C^{p-1-s}_{p-1-i}$ together with $C^{p-i}_{p}=\frac{p}{p-i}C^{p-i-1}_{p-1}$:
    \begin{align}
        a_{p,s}=p(-1)^s (s-1)!\sum_{i=0}^{s}(-1)^i\frac{p}{p-i}C^{p-i-1}_{p-1}C^{p-1-s}_{p-1-i}C^{p-i}_{p+s-i-1}\,.
    \end{align}
    Using the standard identity for binomial coefficients,
    $C^{p-i-1}_{p-1}C^{p-1-s}_{p-1-i}=C^{s}_{p-1}C^{i}_s$,
    and factoring out the terms independent of $i$, we obtain
\begin{align}\label{proof-aps}
    a_{p,s}=p^2(-1)^s (s-1)!C^s_{p-1}\sum_{i=0}^{s}(-1)^i\frac{1}{p-i}C^{i}_sC^{p-i}_{p+s-i-1}\,.
\end{align}
To compute this sum, note that $C^{p-i}_{p+s-i-1}$ is a polynomial in $i$ of degree $s-1\geq0$. For $s=1$ it is equal to $1$, while for $s>1$ we expand it around the point $p-i$:
\(C^{p-i}_{p+s-i-1}=1+(p-i)P(i)\),
where $P(i)$ is a polynomial of degree $s-2$. Now observe that for any $0\leq r\leq s-1$ one has
\begin{align}
    \sum_{i=0}^{s}(-1)^iC^{i}_s i^{r}=0\,.
\end{align}
Therefore, the term involving $P(i)$ does not contribute to the sum, and we may rewrite \eqref{proof-aps} as
\begin{align}\label{proof-aps2}
    a_{p,s}=p^2(-1)^s (s-1)!C^s_{p-1}\sum_{i=0}^{s}(-1)^iC^{i}_s\frac{1}{p-i}\,.
\end{align}
Let us compute the sum:
\begin{multline}
    \sum_{i=0}^{s}(-1)^iC^{i}_s\frac{1}{p-i}
    =\sum_{i=0}^{s}(-1)^iC^{i}_s\int_0^1 t^{p-i-1}dt
    =\int_0^1t^{p-1}\left(\sum_{i=0}^{s}(-1)^iC^{i}_st^{-i}\right)dt=\\
    =\int_0^1t^{p-1} \left(1-\frac{1}{t}\right)^sdt
    =(-1)^s\int_0^1t^{p-1-s}(t-1)^sdt
    =(-1)^s\frac{(p-s-1)!s!}{p!}\,,
\end{multline}
where in the last equality we used the standard expression for the Euler beta function. Substituting this into \eqref{proof-aps2}, we obtain
\begin{align}
    a_{p,s}=p(s-1)!\,,\qquad s\geq1\,.
\end{align}
Using this in \eqref{proof-renorm-A}, we obtain the desired expression for $A$.
    We now compute
    \begin{align}
       B\equiv -\sum_{i=0}^{p-1}C^{i}_p\frac{p!}{(p-i)!}(p-i)\Omega^{p-i-1}(Q\Omega) L_\nabla^{p-i-1}(\Omega^{-p}\bar v)\,.
    \end{align}
    We see that compared to the previous case, the only difference is the factor $-(p-i)\Omega^{-1} Q\Omega$. Therefore,
    \begin{multline}
        B=-\sum_{s=0}^{p-1}(\sum_{i=0}^{s}C^{i}_pC^{s-i}_{p-i-1}\frac{p}{(p-i)!}(-1)^{s-i}(p+s-i-1)!(p-i))\Omega^{-s-1}Q\Omega\bar v^{(p-s-1)}\equiv \\ \equiv \sum_{s=0}^{p-1}b_{p,s}\Omega^{-s-1}Q\Omega\bar v^{(p-s-1)}\,.
    \end{multline}
    Repeating the argument used for $a_{p,s}$, we see that only the term with $s=0$ contributes, and hence $b_{p,s}=-p\delta_{s,0}$. Thus,
    \begin{align}
        B=-p\Omega^{-1}Q\Omega\bar v^{(p-1)}\,.
    \end{align}
The lemma follows by substituting $A$ and $B$ into \eqref{renorm-proof-AB}.
\end{proof}
To prove Corollary~\bref{prop-renorm} we note that \eqref{eq:first-step} implies 
\begin{align}\label{app-renorm-formula}
    kL_Q\bar v^{(k-1)}-(k-p)v^{(k)}=p\xi\bar v^{(k)}-\Omega(L_Q\bar v^{(k)}-v^{(k+1)})
\end{align}
and hence
\begin{multline}
    p\ln\Omega L_Q\bar v^{(p-1)}=\ln\Omega(p\xi\bar v^{(p)}-\Omega(L_Q\bar v^{(p)}-v^{(p+1)}))=\\=L_Q(p\Omega\ln\Omega\bar v^{(p)})-\xi p \bar v^{(p)}-\Omega\ln\Omega((p+1)L_Q\bar v^{(p)}-v^{(p+1)})=L_Q(p\Omega\ln\Omega\bar v^{(p)})+r_1[v]
\end{multline}
thus completing the proof for $N=1$. For $N>1$ we note that
\begin{align}
    r_N[v]=pL_Q(\frac{(-1)^N}{(N+1)!}\Omega^{N+1}\ln\Omega\bar v^{(p+N)})+r_{N+1}[v]
\end{align}
which can be proven using $\eqref{app-renorm-formula}$ and $\xi\Omega^{i}=\frac{1}{i+1}Q(\Omega^{i+1})$ and follows from
\begin{multline}
    \Omega^N\ln\Omega((N+p)L_Q\bar v^{(p+N-1)}-Nv^{(p+N)})=\\=\Omega^N\ln\Omega((N+p)(p\xi\bar v^{(p+N)}-\Omega(L_Q\bar v^{(N+p)}-v^{(N+p+1)}))=\\=L_Q(\frac{p}{N+1}\Omega^{N+1}\ln\Omega\bar v^{(p+N)})-\frac{p}{N+1}\Omega^N\bar v^{(p+N)}-\\\frac{1}{N+1}\Omega^{N+1}\ln\Omega((N+1+p)L_Q\bar v^{(N+p)}-(N+1)v^{(N+p+1)})\,.
\end{multline}

\section{A Certain Combinatorial Identity}\label{app: Faa}
We frequently need to compute expressions of the form $X^k\xi\ms{N}_{a_1\dots a_{D-N}}$, where $X$ is a vector field and 
\begin{align}
    \xi\ms{N}_{a_1\dots a_{D-N}}\equiv\frac{1}{N!}\epsilon_{a_1\dots a_{D-N}b_1\dots b_{N}}\xi^{b_1}\dots\xi^{b_{N}}\,.
\end{align}
The first few of these can be computed explicitly:
\begin{align}
    \begin{split}
        &X\xi\ms{N}_{a_1\dots a_{D-N}}=(X\xi^{b_1})\xi\ms{N-1}_{a_1\dots a_{D-N}b_1}\,,\\
        & X^{2}\xi\ms{N}_{a_1\dots a_{D-N}}=(X^{2}\xi^{b_1})\xi\ms{N-1}_{a_1\dots a_{D-N}b_1}+(X\xi^{b_1})(X\xi^{b_2})\xi\ms{N-2}_{a_1\dots a_{D-N}b_1 b_2}\,.
    \end{split}
\end{align}
Iterating this procedure requires keeping track of combinatorial factors, and it is not difficult to observe that it is governed by the well-known Faà di Bruno formula, which is simply a generalization of the chain rule to higher derivatives. Namely,
    \begin{align}
        X^{k}\xi\ms{N}_{a_1\dots a_{D-N}}=\sum_{i=1}^k B_{k,i}(X\xi,\dots,X^{k-i+1}\xi)^{b_1\dots b_i}\xi\ms{N-i}_{a_1\dots a_{D-N}b_{1}\dots b_i}\,,
    \end{align}
where $B_{k,i}$ are the exponential Bell polynomials, which can be computed iteratively or looked up in references (e.g., \cite{comtet2012advanced}, Appendix ``Fundamental Numerical Tables''). The first few of these are given in the following table:
\begin{table}[h!]
\centering
\begin{tabular}{c|ccccc}
$k \backslash i$ & 1 & 2 & 3 & 4 &5\\ \hline
1 & $x_1$ & & & &\\
2 &  $x_2$ & $x_1^2$ & & &\\
3 & $x_3$ & $3x_1x_2$ & $x_1^3$ & &\\
4 & $x_4$ & $4x_1x_3+3x_2^2$ & $6x_1^2x_2$ & $x_1^4$ &\\
5 & $x_5$& $5x_1x_4+10x_2x_3$& $10x_1^2x_3+15x_1x_2^2$&$10x_1^3x_2$&$x_1^5$
\end{tabular}
\caption{Bell polynomials $B_{k,i}(x_1,x_2,\dots,x_{k-i+1})$}
\end{table}

For instance, we can deduce that
\begin{multline} 
    X^4\xi\ms{D}=(X^4\xi^a)\xi\ms{D-1}_a+(4(X\xi^a)(X^3\xi^b)+3(X^2\xi^a)(X^2\xi^b))\xi\ms{D-2}_{ab}+\\+6(X\xi^a)(X\xi^b)(X^2\xi^c)\xi\ms{D-3}_{abc}+(X\xi^a)(X\xi^b)(X\xi^c)(X\xi^d)\xi\ms{D-4}_{abcd}\,.
\end{multline}

\section{Derivation of the presymplectic potentials}\label{app:scalar-derivation}

Although in the main text we simply present the symplectic potential on $F^\varphi_\red$ and check that it is equivalent to the standard one, it is useful to show how it arises from the standard presymplectic potential on $(\tilde F^\varphi,Q)\to(\algg_\Lambda[1],\mathrm{d}_{\algg})$:
    \begin{align}
    \tilde \chi'_\varphi=\tilde\xi\ms{d}_a\tilde\nabla^a\tilde\varphi d\tilde\varphi\,,\quad \tilde H'_\varphi=\frac{(-1)^{d+1}}{2}\tilde\xi\ms{d+1}(m^2\tilde\varphi^{2}+\tilde\nabla_a\tilde\varphi\tilde\nabla^a\tilde\varphi)\,.
\end{align}
Using $\xi^a=\Omega\tilde\xi^{a}$ and \eqref{scalar-near-b} we can uplift $\tilde\chi'_\varphi$ and $\tilde H'_\varphi$ to $Int(F^\varphi)$:
\begin{align}
\begin{split}
    &\tilde \chi_\varphi'=\Omega^{-d-w}\xi\ms{d}_a(\Omega\nabla^a\varphi-w(n^a\varphi))d(\Omega^{-w}\varphi)\,,\\
    &\tilde H_\varphi'=\frac{(-1)^{d+1}}{2}\Omega^{-2w-d-1}\xi\ms{d+1}(\Omega^2\nabla_a\varphi\nabla^a\varphi-2w\Omega\varphi n^a\nabla_a\varphi+(m^2+n_an^a w^2)\varphi^2)\,.
\end{split}
\end{align}
It turns out to be convenient to choose a different representative within the equivalence class of the presymplectic potential. To see this, note that the term
\begin{align}
\begin{split}
&-w(\Omega^{-d-w}\xi\ms{d}_a(n^a\varphi)d(\Omega^{-w}\varphi)=-\frac{w}{2}\Omega^{-d}\xi_a\ms{d}n^ad(\Omega^{-2w}\varphi^{2})=d\alpha+\beta\,,\quad\text{where
}\\
    &\alpha=(-1)^{d+1}\frac{w}{2}\Omega^{-d-2w}\xi_a\ms{d}n^a\varphi^2\,,\qquad \beta=(-1)^d\frac{w}{2}\varphi^{2}\Omega^{-2w}d(\Omega^{-d}\xi_a\ms{d}n^a)\in \cI_{Int(F)}\,.
\end{split}
\end{align}
We exploit the freedom in the definition of the presymplectic potential and define
\begin{align}\label{utv-adsscalar-1}
    \tilde\chi_\varphi\equiv \tilde\chi_\varphi'-d\alpha-\beta= \Omega^{-d-w+1}\xi\ms{d}_a\nabla^a\varphi d(\Omega^{-w}\varphi)\,.
\end{align}
Under this shift, the covariant Hamiltonian also changes:
\begin{multline}\label{utv-adsscalar-2}
    \tilde H_\varphi=\tilde H_\varphi'+i_Q\beta=\frac{(-1)^{d+1}}{2}\Omega^{-2w-d-1}[\xi\ms{d+1}(\Omega^2\nabla_a\varphi\nabla^a\varphi-2w\Omega\varphi n^a\nabla_a\varphi+\\+(m^2+n_an^a w(d+w))\varphi^2)+w\Omega\varphi^2(\xi\ms{d+1}\tau-\Omega\lambda^a\xi_a\ms{d})]\,.
\end{multline}
The pullback of $\tilde\chi_\varphi$ and $\tilde H_\varphi$ to $Int(F^\varphi_{\red})$ completes the proof of the statement.

\setlength{\itemsep}{0em}
\small


\begingroup\raggedright\endgroup

\end{document}